\documentclass[10pt]{article}
\usepackage[lmargin=1in,rmargin=1in,tmargin=1in,bmargin=1in]{geometry}
\usepackage{amsmath,amsthm,amsfonts,amssymb,comment,slashed,mathtools}
\usepackage[english]{babel}
\usepackage[T1]{fontenc}  
\usepackage[utf8]{inputenc}
\usepackage{csquotes}
\usepackage{bm}
\usepackage[shortlabels]{enumitem}
\usepackage[affil-it]{authblk}  

\usepackage[section]{placeins}
\usepackage{soul} 

\usepackage{bm}
\usepackage{mathrsfs, graphicx} 
\usepackage{stmaryrd}
\usepackage{epigraph}

\usepackage[%
bookmarks=true,
colorlinks=true,
linkcolor=black,
urlcolor=black,
citecolor=black,
plainpages=false,
pdfpagelabels,
final]{hyperref}

\usepackage{xcolor}

\usepackage[backend=biber,style=alphabetic,giveninits=true,doi=false,isbn=false,url=false,sorting=nyt,maxbibnames=99,date=year]{biblatex}
\renewbibmacro{in:}{}

\usepackage{chngcntr}

\usepackage[capitalise,nameinlink]{cleveref}

\numberwithin{equation}{section}

\theoremstyle{plain}
\newtheorem{thm}{Theorem}
\newtheorem{prop}{Proposition}[section]
\newtheorem{lem}{Lemma}[section]
\newtheorem*{thm*}{Theorem}
\newtheorem{cor}{Corollary}
\newtheorem*{conj*}{Conjecture}

\theoremstyle{definition}
\newtheorem{defn}{Definition}[section]
\theoremstyle{remark}
\newtheorem{rk}{Remark}[section]
\newcounter{parentnumber}
\makeatletter 
\newenvironment{subtheorem}[1]{%
  \counterwithin*{thm}{parentnumber}
  \def\subtheoremcounter{#1}%
  \refstepcounter{#1}%
  \protected@edef\theparentnumber{\csname the#1\endcsname}%
  \setcounter{parentnumber}{\value{#1}}%
  \setcounter{#1}{0}%
  \expandafter\def\csname the#1\endcsname{\theparentnumber\Alph{#1}}%
  \ignorespaces
}{%
  \setcounter{\subtheoremcounter}{\value{parentnumber}}%
  \counterwithout*{thm}{parentnumber} 
  \ignorespacesafterend
}
\makeatother

\NewDocumentEnvironment{manual}{O{thm}m}
 {%
  \renewcommand\thethm{#2}%
  \begin{#1}
 }
 {\end{#1}}

\crefname{thm}{Theorem}{Theorems}

\renewcommand{\Bbb}{\mathbb}

\newcommand{\ve}{\varepsilon}
\newcommand{\les}{\lesssim}

\newcommand{\Diff}{\operatorname{Diff}}

\newcommand{\grd}{\operatorname{grad}}
\newcommand{\Ric}{\mathrm{Ric}}

\newcommand{\R}{\Bbb R}

\renewcommand{\Re}{\mathrm{Re}}
\renewcommand{\Im}{\mathrm{Im}}

\begin{document}

\title{Gravitational~collapse~to~extremal~black~holes\\ and the third law of black hole thermodynamics}

\author[1]{Christoph~Kehle\thanks{christoph.kehle@eth-its.ethz.ch}}
\author[2]{Ryan Unger\thanks{runger@math.princeton.edu}}
\affil[1]{\small  Institute~for~Theoretical~Studies \& Department of Mathematics,~ETH~Zürich,

Clausiusstrasse~47,~8092~Zürich,~Switzerland \vskip.1pc \ 
}
\affil[2]{\small  Department of Mathematics, Princeton University, 
	Washington~Road,~Princeton~NJ~08544,~United~States~of~America \vskip.1pc \  
	}

 \date{February 15, 2024}
\maketitle

\begin{abstract}
We construct examples of black hole formation from regular, one-ended asymptotically flat Cauchy data for the Einstein--Maxwell-charged scalar field system in spherical symmetry which are exactly isometric to extremal Reissner--Nordstr\"om after a finite advanced time along the event horizon.  
Moreover, in each of these examples the apparent horizon of the black hole coincides with that of a Schwarzschild solution at earlier advanced times.
In particular, our result can be viewed as a definitive \emph{disproof} of the ``third law of black hole thermodynamics.''
  
The main step in the construction is a novel $C^k$ characteristic gluing procedure, which interpolates between a light cone in Minkowski space and a Reissner--Nordström event horizon with specified charge to mass ratio $e/M$. Our setup is inspired by the recent work of Aretakis--Czimek--Rodnianski on perturbative characteristic gluing for the Einstein vacuum equations. However, our construction is fundamentally nonperturbative and is based on a finite collection of scalar field pulses which are modulated by the Borsuk--Ulam theorem.
\end{abstract}

\thispagestyle{empty}
\newpage 
\tableofcontents
\enlargethispage{1in}
\thispagestyle{empty}
\newpage

\section{Introduction} \label{sec:introduction}

Following pioneering work of Christodoulou \cite{DC-thesis} and Hawking \cite{Hawking-PRL-71} on energy extraction from rotating black holes,  Bardeen, Carter, and Hawking \cite{BCH} proposed---via analogy to classical thermodynamics---the celebrated \emph{four laws of black hole thermodynamics}. In particular, letting the surface gravity $\kappa$ of the black hole take the role of its temperature, an identification later vindicated by the discovery of Hawking radiation \cite{Hawking-radiation}, they proposed a \emph{third law}  in analogy to ``Nernst's theorem'' in classical thermodynamics. 
 
 \begin{conj*}[The third law of black hole thermodynamics]
 A subextremal black hole cannot become extremal in finite time by any continuous process, no matter how idealized, in which the spacetime and matter fields remain regular and obey the weak energy condition. 
\end{conj*} 

This version is distilled from the literature, particularly from the work of Israel \cite{Israel-third-law, Israel-pdf} who added explicit mention of regularity and the weak energy condition to avoid previously known examples \cite{Delacruz-Israel, Kuchar, Boulware, Farrugia-Hajicek, Sullivan-Israel, Proszynski} which would otherwise violate the third law. In this paper, we show that the third law is fundamentally flawed in a manner that does not appear to be salvageable by further reformulation. Indeed, we construct counterexamples in the Einstein--Maxwell-charged scalar field model in spherical symmetry, a model which satisfies the dominant energy condition, arising from arbitrarily regular initial data on a one-ended asymptotically flat hypersurface.

\begin{thm}\label{thm-main-page}
Subextremal black holes can become extremal in finite time, evolving from \ul{regular} initial data. In fact, there exist regular one-ended Cauchy data for the Einstein--Maxwell-charged scalar field system 
which undergo gravitational collapse and form an exactly Schwarzschild apparent horizon, only for the spacetime to form an exactly extremal Reissner--Nordstr\"om event horizon at a later advanced time. 

In particular, \ul{the ``third law of black hole thermodynamics'' is false}.

\end{thm}

\begin{figure}[ht]
 \centering{
 \def\svgwidth{14pc}
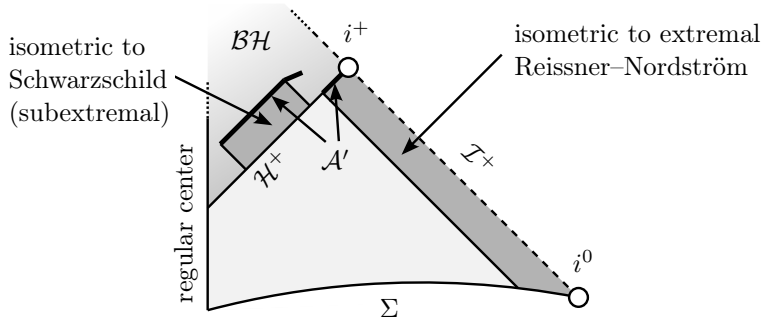
}
\caption{Penrose diagram of our counterexample to the third law arising from regular initial data on $\Sigma$. The northwest edge of the Schwarzschild region is exactly isometric to a section of the $r=2M$ hypersurface in Schwarzschild. The outermost apparent horizon $\mathcal A'$ is initially indistinguishable from Schwarzschild and then jumps out in finite time to be exactly isometric to the event horizon of extremal Reissner--Nordstr\"om. For speculations about the future boundary of the interior, see already \cref{figure-1-comment}. The behavior of our solutions can be modified to be subextremal near $i^0$, see already \cref{rk:Leonhard}.}
\label{fig:first-page}
\end{figure}

Our result also clarifies some issues raised by Israel in \cite{Israel-third-law, Israel-pdf} who seemingly associated a disconnected outermost apparent horizon with a severe lack of regularity of the spacetime metric and/or matter fields. We stress that \ul{our examples are regular despite the disconnectedness of the apparent horizon}. We note moreover that Israel seemed to associate extremization with the black hole ``losing its trapped surfaces.'' This confusion appears to be related to his implicit assumption that the apparent horizon is connected. Since the Einstein--Maxwell-charged scalar field matter manifestly obeys the dominant energy condition,  \ul{trapped surfaces are not lost in any sense, nonetheless, the black hole becomes extremal in finite time}.
 In the examples we construct, there exists an open set of trapped spheres inside the black hole region, which persist for all advanced time until they encounter the Cauchy horizon or a curvature singularity inside the black hole. However, there is a neighborhood of the event horizon which does not contain any (strictly) trapped surfaces. For an extended discussion of these issues, see already \cref{retiring}.

\begin{rk}\label{rk:Leonhard}
Note that in discussions of the third law, the focus is typically on dynamics near the event horizon and apparent horizon, in late advanced time. Our counterexamples depicted in \cref{fig:first-page} are isometric to extremal Reissner--Nordstr\"om for all sufficiently late advanced times and all retarded times to the past of the event horizon, in particular near spatial infinity $i^0$. However, by using a scattering argument as in \cite{KehrbergerI}, one can easily modify our examples so as to be subextremal in a neighborhood of $i^0$, if desired. 
\end{rk}

Our falsification of the third law (\cref{cor:trapped-surfaces-ERN}) is preceded by our following more general result. We construct regular one-ended Cauchy data for the Einstein--Maxwell-charged scalar field system in spherical symmetry whose black hole exterior evolves (in fact is eventually isometric) to a  Schwarzschild black hole with prescribed mass $M>0$ or to a \emph{subextremal} or \emph{extremal} Reissner--Nordström black hole with prescribed mass $M>0$ and prescribed charge to mass ratio $\mathfrak q \doteq e/M \in [-1,1]$.  
The Einstein--Maxwell-charged scalar field  (EMCSF) system reads
\begin{align}\label{eq:Einstein-Maxwell-1}
    R_{\mu\nu}(g)-\tfrac 12 R(g)g_{\mu\nu}&= 2\left(  T^\mathrm{EM}_{\mu\nu}+  T^\mathrm{CSF}_{\mu\nu}\right),\\
    \nabla^\mu F_{\mu\nu}&= 2\mathfrak e\, \Im(\phi\overline{D_\nu \phi})\label{eqn:Maxwell-intro},\\
    g^{\mu\nu}D_\mu D_\nu \phi&=0,\label{eq:Einstein-Maxwell-3}
\end{align}
 for a quintuplet $(\mathcal M,g,F,A,\phi)$,  where $(\mathcal M,g)$ is a  (3+1)-dimensional Lorentzian manifold, $\phi$ is a complex-valued scalar field, $A$ is a   real-valued 1-form, $F=dA$ is a   real-valued 2-form, $D=d+i\mathfrak eA$ is the gauge covariant derivative, $ \mathfrak e \in \mathbb R\setminus \{0\}$ is a fixed coupling constant  representing  the charge of the scalar field, and the energy momentum tensors are defined by 
 \begin{align}
T^\mathrm{EM}_{\mu \nu}&\doteq g^{\alpha \beta}F _{\alpha \nu}F_{\beta \mu }-\tfrac{1}{4}F^{\alpha \beta}F_{\alpha \beta}g_{\mu \nu},\\
  T^\mathrm{CSF}_{\mu \nu}&\doteq  \mathrm{Re}(D _{\mu}\phi \overline{D _{\nu}\phi}) -\tfrac{1}{2}g_{\mu \nu} g^{\alpha \beta} D _{\alpha}\phi \overline{D _{\beta}\phi}.\label{TCSF}
 \end{align}
 We refer to \cref{sec:prelims} for the form of the EMCSF system in spherical symmetry.
\begin{rk}
     All of the results in this paper also hold for the \emph{Einstein--Maxwell--charged Klein--Gordon} system in which the  wave equation \eqref{eq:Einstein-Maxwell-3} is replaced by the Klein--Gordon equation
     \begin{equation*}
         g^{\mu\nu}D_\mu D_\nu \phi=\mathfrak m^2 \phi,
     \end{equation*}
     where $\mathfrak m\in \Bbb R_{>0}$ represents the mass of the scalar field and satisfies $\mathfrak m M  \ll \mathfrak e M$. Here $M$ denotes the mass of the black holes, see already \cref{thm-informal-statement} below.
 \end{rk}

We emphasize that not only are our data in the above examples regular, but the spacetimes arise from gravitational collapse, i.e., the initial data surface is one-ended, has a regular center, lies entirely in the domain of outer communication, and the black hole forms strictly to the future of initial data.  In particular, in contrast to what has been suggested numerically \cite{Alcubierre2014, Pani2020}, there is no upper bound (strictly less than unity) on the charge to mass ratio of a black hole which can be achieved in gravitational collapse for this model.

The key step toward the construction of one-ended Cauchy data evolving to black holes with prescribed mass and charge is a novel \emph{characteristic}/\emph{null gluing} result. The study of the characteristic gluing problem for the  Einstein vacuum equations (outside of spherical symmetry) was recently initiated by Aretakis, Czimek, and Rodnianski \cite{ACR1,ACR2,ACR3} in the perturbative regime around Minkowski space. Our setup is directly inspired by their work. In contrast, however, our null gluing construction (while in spherical symmetry) necessarily exploits the \emph{large data regime} in order to glue a cone of Minkowski space to a black hole event horizon along a null hypersurface within the EMCSF model. The construction of Cauchy data on $\Sigma \cong \R^3$ collapsing to an extremal or subextremal event horizon will then follow from \cref{thm-informal-statement}  as \cref{main-corollary} presented in \cref{subsec:cauchy-data-construction}.

 On the basis of our spherically symmetric horizon gluing construction in \cref{thm-informal-statement}, the results and framework introduced in \cite{ACR1,ACR2,ACR3, Czimek2022-cl,DHR19}, and \cref{softness} below, we formulate the following
\begin{conj*}
There exist regular one-ended Cauchy data for the  Einstein vacuum equations
\begin{equation*}\label{eq:VEE}
    \Ric(g) =0
\end{equation*}
which undergo gravitational collapse and form an exactly Schwarzschild apparent horizon, only for the spacetime to form an exactly extremal Kerr event horizon at a later advanced time. 
In particular,  \ul{already in vacuum},    the ``third law of black hole thermodynamics'' is false.
\end{conj*}

\subsection{Event horizon gluing}
\label{subsec:event-horizon-gluing}
We will now state the rough version of our main null gluing theorem, which concerns gluing a null cone in Minkowski space to a Reissner--Nordstr\"om event horizon.

\begin{thm}[Rough version]\label{thm-informal-statement}
Let $k\in \Bbb N$ be a regularity index, $\mathfrak q\in [-1,1]$  a charge to mass ratio,  and  $\mathfrak e\in \Bbb R\setminus\{0\}$ a fixed coupling constant. For any $M$ sufficiently large depending on $k$, $\mathfrak q$, and $\mathfrak e$, there exist spherically symmetric characteristic data for the Einstein--Maxwell-charged scalar field system with coupling constant $\mathfrak e$ gluing a Minkowski null cone of radius $\tfrac 12M$ to a Reissner--Nordstr\"om event horizon with mass $M$ and charge $e= \mathfrak q M$ up to order $k$. 
\end{thm}

We also refer to  \cref{fig:main-thms} for an illustration of our construction.
\begin{figure}[ht]
 \centering{
 \def\svgwidth{18pc}
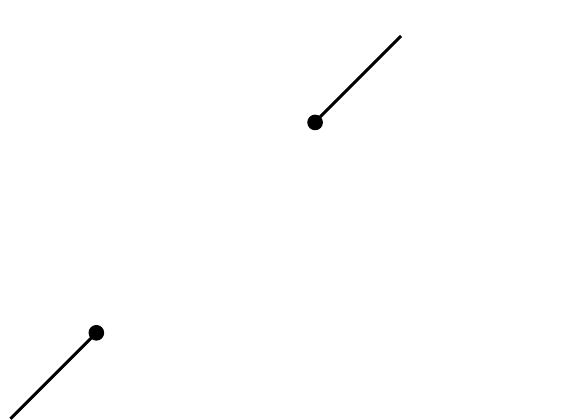
}
\caption{Setup of \cref{thm-informal-statement}.}
\label{fig:main-thms}
\end{figure}

For the precise version of \cref{thm-informal-statement} we refer to \cref{main-thm-Schw} and \cref{main-thm-RN} in  \cref{sec:precise-version-main-theorems}. In fact, more generally, we can replace the Minkowski sphere with certain Schwarzschild exterior spheres at $v=0$, which is important for constructing counterexamples to the third law of black hole thermodynamics (see already \cref{retiring}). Furthermore, when $\mathfrak q=0$ we may take the scalar field to be real-valued, in which case the EMCSF system collapses to the Einstein-scalar field system.

\begin{rk}
For the proofs of \cref{cor:Cauchy-horizon-RN} and \cref{cor:Cauchy-horizon-closes--off} below, we will use versions of \cref{thm-informal-statement} where the top sphere is not located on a horizon. See \cref{thm:interior-gluing} and \cref{thm:backwards-gluing}  in \cref{sec:precise-version-main-theorems} below.
\end{rk}

\begin{rk} 
With our methods one can also construct characteristic data which are exactly Minkowski initially and then settle down, but only asymptotically, to a Schwarzschild or (sub-)extremal Reissner--Nordström event horizon of prescribed mass and charge. The rate of decay can be chosen to be $ |\partial_v \phi| \approx v^{-p}$, $p>\frac 12$, in a standard Eddington--Finkelstein gauge for Schwarzschild or subextremal  Reissner--Nordström black holes. This provides examples of ``global'' characteristic data settling down at certain prescribed rates as assumed in \cite{VdM18,gajic-luk,K-VdM}.
\end{rk}

\subsection{Gravitational collapse from event horizon gluing}
 \label{subsec:cauchy-data-construction}

 For appropriate matter models, the  Einstein equations \begin{equation*}
     \Ric(g)-\tfrac 12 R(g)g=2T
 \end{equation*}  are well-posed (see \cite{MR53338,CBG69} for the vacuum case) as a Cauchy problem for suitable initial data posed on a $3$-manifold $\Sigma$, which will then be isometrically embedded as a spacelike hypersurface in a Lorentzian manifold $(\mathcal M,g)$.
The textbook explicit black hole solutions such as the Schwarzschild spacetime do not contain \emph{one-ended} Cauchy surfaces $\Sigma \cong \mathbb R^3$ but are instead foliated by \emph{two-ended} hypersurfaces $\Sigma\cong\mathbb R\times  S^2$. Thus, a natural and physically relevant problem is to construct regular asymptotically flat data on $\Sigma\cong\R^3$ which evolve to  a black hole spacetime. The first example of \emph{gravitational collapse}, that is a black hole spacetime containing a one-ended Cauchy surface which lies outside of the black hole region, was constructed by Oppenheimer and Snyder \cite{OS39} for the Einstein-massive dust model in spherical symmetry.

Using \cref{thm-informal-statement}, by solving the Einstein equations \emph{backwards}, we construct examples of gravitational collapse where the domain of outer communication is eventually exactly isometric to Reissner--Nordstr\"om with prescribed mass and charge.  
The proof of the following \cref{main-corollary} is given in \cref{subsec:grav-collapse-RN}.
 
 \begin{cor}[Exact Reissner--Nordstr\"om arising from gravitational collapse]\label{main-corollary}
 For any regularity index $k\in \Bbb N$ and charge to mass ratio $\mathfrak q\in[-1,1]$, there exist spherically symmetric, asymptotically flat Cauchy data for the Einstein--Maxwell-charged scalar field system, with $\Sigma\cong \Bbb R^3$ and a regular center, such that the maximal  future globally hyperbolic development $(\mathcal M^4,g)$ has the following properties: 
 \begin{itemize}
\item All dynamical quantities are at least $C^{k}$-regular.
     \item Null infinity $\mathcal I^+$ is complete.
     \item The black hole region is non-empty, $\mathcal B\mathcal H\doteq \mathcal M\setminus J^-(\mathcal I^+)\neq \emptyset$.
     \item The Cauchy surface $\Sigma$ lies in the causal past of future null infinity, $\Sigma\subset J^-(\mathcal I^+)$. In particular, $\Sigma$ does not intersect the event horizon $\mathcal H^+\doteq\partial(\mathcal{BH})$. Furthermore, $\Sigma$ contains no trapped or antitrapped surfaces.
     \item For sufficiently late advanced times $v\ge v_0$, the domain of outer communication, including the event horizon, is isometric to that of a Reissner--Nordstr\"om solution with charge to mass ratio $\mathfrak q$. For $v\ge v_0$, the event horizon of the spacetime can be identified with the event horizon of Reissner--Nordstr\"om. 
 \end{itemize}
 \end{cor}

 \begin{figure}[ht]
\centering{
\def\svgwidth{20pc}
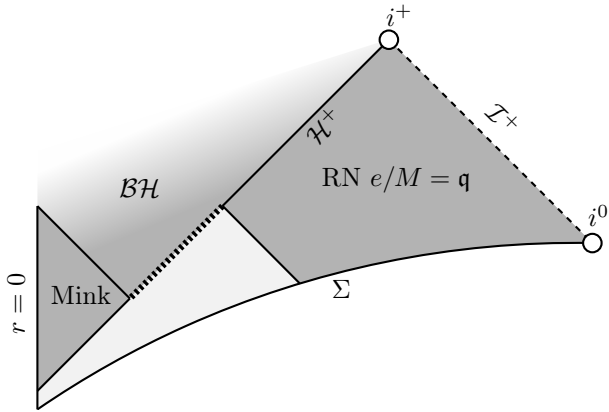}
\caption{Penrose diagram for \cref{main-corollary}. The textured line segment is where the data constructed in \cref{thm-informal-statement} live.}
\label{main-corollary-diagram}
\end{figure}

Note that in the case $|\mathfrak q|=1$, this does not yet furnish a counterexample to the third law of black hole thermodynamics, as the spacetime does not necessarily contain a subextremal apparent horizon. For the counterexample we must defer to \cref{cor:trapped-surfaces-ERN} in 
 \cref{sec:disproof} below. 
 
 However, in our proof of \cref{main-corollary}, forming an extremal black hole with $|\mathfrak q|=1$ is no different from any subextremal charge to mass ratio $|\mathfrak q|<1$ (see already \cref{subsubsec:exceptionality-of-third-law-violations}). In particular, in contrast with what has been suggested by numerical simulations \cite{Alcubierre2014,Pani2020}, there is no universal upper bound (strictly less than unity) for $|\mathfrak q|$. Given that we have now proved that extremal Reissner--Nordstr\"om can arise in gravitational collapse, it would be interesting to rethink the numerical approach to this problem and develop a scheme to construct such solutions numerically. Because our construction is fundamentally teleological (see already \cref{MFGHD}), it might be challenging to directly find suitable data on $\Sigma$ by trial and error.

The formation of black holes is a very well studied problem in spherical symmetry. We mention here only the Einstein-scalar field model, for which Christodoulou \cite{C-91} first showed that concentration of the scalar field can lead to formation of a black hole. This result played a decisive role in Christodoulou's proof of weak cosmic censorship in spherical symmetry \cite{C-WCC}. Dafermos constructed solutions of the Einstein-scalar field system which collapse to the future but are complete and regular to the past \cite{Dafermos-complete-regular-past}. For work on other matter models, see for example \cite{Hakan, AnLim22}.
 
 Outside of spherical symmetry (for the Einstein vacuum  equations), formation of black holes was studied by Christodoulou in the seminal monograph \cite{Christo09}. Christodoulou constructed \emph{characteristic data} for the Einstein vacuum equations containing no trapped surfaces, but whose evolution contains trapped surfaces in the future. Li and Yu \cite{LiYu} showed how to combine Christodoulou's construction with the spacelike gluing technique of Corvino and Schoen \cite{CorvinoSchoen} to construct asymptotically flat \emph{Cauchy data} containing no trapped surfaces, but whose evolution contains trapped surfaces in the future. Later, Li and Mei \cite{LiMei} observed that the Corvino--Schoen gluing can be done ``behind the event horizon,'' which yields a genuine construction of gravitational collapse in vacuum arising from one-ended asymptotically flat Cauchy data. 
 
 However, the constructions of the above type rely on the observation that if an additional restriction is imposed on the seed data in \cite{Christo09}, then the resulting spacetime has a region of controlled size which is close to Schwarzschild. The Corvino--Schoen gluing then selects \emph{very slowly rotating} Kerr parameters for the exterior region. 
 
 We emphasize that our gluing approach yields collapsing spacetimes with exactly specified (in particular, extremal, if desired) parameters, but is so far limited to spherical symmetry. 
 
\begin{rk}\label{softness}
Our derivation of \cref{main-corollary} from \cref{thm-informal-statement} is completely soft and does not make use of spherical symmetry. Therefore, if versions of the main gluing theorems were known for the Einstein vacuum equations (for example, gluing a Minkowski cone to an extremal Kerr event horizon, or more generally a Schwarzschild exterior sphere to an extremal Kerr event horizon), then our procedure would yield vacuum spacetimes arising from gravitational collapse which are eventually isometric to extremal Kerr. Furthermore, such a construction would also yield a disproof of the third law in vacuum. 
 \end{rk}

 \begin{rk}
By the very nature of our gluing procedure, the constructions in this paper have finite regularity ($C^k$ for arbitrarily large $k$). It would be mathematically interesting to create such examples with $C^\infty$ regularity. See already \cref{rk:smooth}.
 \end{rk}

\begin{rk}
The existence of dynamical spacetimes satisfying the dominant energy condition which are extremal at spacelike infinity $i^0$ does not contradict the positive mass theorem ``with charge'' \cite{GHHP, CRT06} because the matter itself carries charge. Concretely, condition (27) in \cite{GHHP} is false for various charged matter models, in particular the Einstein--Maxwell-charged scalar field model with small (or zero) mass. 
\end{rk}

\subsection{Characteristic gluing setup and proof}

\subsubsection{Previous work on characteristic gluing}
 
The gluing problem along characteristic hypersurfaces for hyperbolic equations and associated null constraints already appears for the linear wave equation on Minkowski space. On $\Bbb R^{3+1}$, let $u=\tfrac 12 (t-r)$, $v=\tfrac 12 (t+r)$, and let $\phi$ be a spherically symmetric solution to the wave equation, i.e.,
\begin{equation}
    \partial_u\partial_v(r\phi)=0.\label{mink-wave}
\end{equation}
Let $C\cup \underline C$ be a spherically symmetric bifurcate null hypersurface, that is, $C=\{u=u_0\}\cap\{v\ge v_0\}$ and $\underline C= \{u_1\ge u\ge u_0\}\cap \{v=v_0\}$. The wave equation \eqref{mink-wave} implies that $\partial_u(r\phi)$ \emph{is conserved along the outgoing cone} $C$. This implies that $\partial_u\phi$ cannot be freely prescribed along $C$, but is in fact determined by $\partial_u\phi$ on the bifurcation sphere $C\cap \underline C$. Indeed, the characteristic initial value problem is well posed with just $\phi$ itself prescribed along $C\cup \underline C$---the full 1-jet of $\phi$ can then be recovered from \eqref{mink-wave}. For general spacetimes, the question of null gluing for the linear wave equation was studied by Aretakis \cite{AretakisLinear}. 

For a general wave equation, ingoing derivatives satisfy transport equations along outgoing null cones. The general $C^k$ characteristic gluing problem is to be given two spheres $S_1$ and $S_2$ along an outgoing null cone $C$, and $k$ ingoing and outgoing derivatives of $\phi$ at $S_1$ and $S_2$. One then seeks to prescribe $\phi$ along the part of $C$ between $S_1$ and $S_2$ so that the outgoing derivatives agree with the given ones and the solutions of the transport equations for the ingoing derivatives have the specified initial and final values. In general, the linear characteristic gluing problem is obstructed due to the presence of \emph{conserved charges} stemming from \emph{conservation laws} along $C$.
\begin{rk}\label{rk:smooth}
    Even in the absence of conservation laws at any order, $C^\infty$ gluing of transverse derivatives may be obstructed in linear theory. This can be seen already for the $(1+1)$-dimensional wave equation $\partial_u \partial_v \phi = f(u,v) \phi$ for generic   $f\in C^\infty(\mathbb R^{1+1})$. For such an $f$, there are no conservation laws at any order and by imposing trivial data at $S_1$ and very rapidly growing (in $k$) $\partial_u^k$-derivatives at $S_2$, one can show that $C^\infty$ gluing cannot be achieved. Note that in $1+1$ dimensions, $S_1$ and $S_2$ are points. 
\end{rk}

The null gluing problem for the Einstein vacuum equations was recently initiated by Aretakis, Czimek, and Rodnianski in a fundamental series of papers \cite{ACR1, ACR2, ACR3}. Their proof uses the inverse function theorem to reduce the nonlinear problem to a linear characteristic gluing problem for the linearized Einstein equations in double null gauge around Minkowski space, in the formalism of Dafermos--Holzegel--Rodnianski \cite{DHR19}. This linearized problem is carefully analyzed and the authors identify \emph{infinitely many} conserved charges which are obstructions to the linear gluing problem. All but ten of the charges turn out to be related to linearized gauge transformations (cf. the ``pure gauge solutions'' of \cite{DHR19}). An inverse function theorem argument then gives nonlinear gluing close to Minkowski space, provided that, \emph{a posteriori}, the 10 transported gauge-invariant charges at  $S_2$ agree with the prescribed charges on $S_2$ (which is in fact also perturbed to deal with the gauge-dependent charges).
Very recently,  Czimek and Rodnianski \cite{Czimek2022-cl} carefully exploited the \emph{nonlinear} structure of the null constraint equations to nonlinearly compensate for failure of matching of the linearly conserved charges. In this way, the authors prove obstruction-free gluing for characteristic and Cauchy data near Minkowski space.

In the present paper, we are however interested in a different regime of gluing. We wish to glue two specific null cones: a light cone in Minkowski space and a Reissner--Nordstr\"om event horizon, as a solution of the EMCSF null constraint system. On the one hand, this is a genuine ``large data'' gluing problem, as these cones are very dissimilar in a gauge invariant sense and there is no known spacetime around which one could reasonably linearize the equations. On the other hand, we study our problem in spherical symmetry, which makes it considerably more tractable. We refer to \cref{subsec:defn-sphere-data} below for a precise definition of characteristic gluing in spherical symmetry.

\subsubsection{Outline of the proof of \texorpdfstring{\cref{thm-informal-statement}}{Theorem 2}}
In the Einstein--Maxwell-charged scalar field model in spherical symmetry, the spacetime metric is written in double null gauge as
 \begin{equation*}
    g = -\Omega^2dudv +r^2g_{S^2},
\end{equation*}
where $\Omega^2$ is the lapse and $r$ the area-radius. We also have a complex-valued scalar field $\phi$ and a real-valued charge $Q$, which is related to the only nonzero component of the electromagnetic tensor $F$. We choose an electromagnetic gauge in which $A=A_u\,du$, where $A$ is a gauge potential for $F$. The dynamical variables to be glued along an outgoing cone (which we will call $C_{-1}\doteq\{u=-1\}$) are  $(r,\Omega^2,\phi,Q,A_u)$.
The charge $Q$ solves first order equations in $u$ and $v$, $A_u$ is computed from $Q$ via $F=dA$, and the variables $r$, $\Omega^2$, and $\phi$ solve coupled nonlinear wave equations involving also $Q$ and $A_u$. See already equations \eqref{eq:CSF-wave-equation}--\eqref{eq:Au-transport}. Since the value of $\Omega^2$ along any given null cone (or bifurcate null hypersurface) can be adjusted by reparametrizing the double null gauge, we impose that $\Omega^2\equiv 1$. 

We first consider \emph{Raychaudhuri's equation} (see already \eqref{eq:CSF-Raychaudhuri-v}), which reads in the gauge $\Omega^2\equiv 1$
\begin{equation}
    \partial_v^2 r = -r|\partial_v\phi|^2.\label{eq:Ray-simple}
\end{equation}
This equation gives a nonlinear constraint on $C_{-1}$ and completely determines $r$ on $C_{-1}$ given $r$ and $\partial_v r$ at one point of $C_{-1}$ and $\phi$ along $C_{-1}$. Thus, in the gauge  $\Omega^2\equiv  1$ along $C_{-1}$, up to specifying the dynamical quantities at a sphere, the free data in this problem is exactly $\phi$ on $C_{-1}$: All other dynamical quantities and their derivatives (both in the $u$ and $v$ coordinates) along $C_{-1}$ can be obtained from $\phi$ and the equations \eqref{eq:CSF-wave-equation}--\eqref{eq:CSF-Raychaudhuri-v}.

We will choose $\phi$ to be compactly supported on the textured segment in \cref{fig:main-thms} and set \[\partial_u\phi(0)=\cdots = \partial_u^k\phi(0)=0,\] where $k$ is the order at which we wish to glue. A first attempt to solve the gluing problem would be to set $(r,\Omega^2,Q,A_u)$ and derivatives to have their ``Minkowski values'' at the sphere $v=0$ and then prescribe $\phi(v)$ so that the dynamical variables reach their ``Reissner--Nordstr\"om values'' at $v=1$. However, specifying a ``Minkowski value'' for $\partial_vr$ is essentially another gauge choice, and the gauge invariance of the equations enables a much more convenient strategy. 

Given that $\phi$ vanishes to order $k$ at $v=0$, to know that the sphere $v=0$ is a sphere in Minkowski space to order $k$, we merely need to know that $r(0)>0$ and that the charge $Q$ and the Hawking mass (see already \eqref{eq:Hawking-mass}) both vanish. See already \cref{lem:gauge-Mink}. This reduces to the statement that in the gauge $\Omega^2\equiv1$,
\begin{equation*}
    \partial_ur(0)\partial_vr(0)=-\tfrac 14.
\end{equation*}
 Since $r$ solves a wave equation (see already \eqref{eq:CSF-equation-for-r}), $\partial_u r$ solves a first order equation in $v$, so it is determined on $C_{-1}$ by $\partial_u r(0)$ alone. Given $\phi$, we solve Raychaudhuri's equation \eqref{eq:Ray-simple} \emph{backwards}, i.e., we teleologically normalize $r$ at the final sphere by setting
\begin{align*}
r(1)&=r_+\doteq \left(1+\sqrt{1-\mathfrak q^2}\right)M\\
\partial_v r(1)&=0
\end{align*}
and then set 
\begin{equation*}
    \partial_ur(0)\doteq \frac{-\tfrac 14}{\partial_v r(0)}.
\end{equation*}
Therefore the only ``constraint'' is that $\partial_v r(0)>0$, which will be automatically satisfied by the monotonicity property of Raychaudhuri's equation as long as $r>0$.

The charge $Q$ is determined by \emph{Maxwell's equation} (see already \eqref{eq:Maxwell-v})
\begin{equation*}
    \label{eq:Maxwell-intro} \partial_v Q= \mathfrak e r^2 \Im(\phi\overline{\partial_v \phi}).
\end{equation*}
Integrating this forwards in $v$ yields the charge condition 
\begin{equation}
    \int_0^{1}\mathfrak e r^2 \Im(\phi\overline{\partial_v\phi})\,dv = \mathfrak q M.\label{eqn:intro-charge-integral}
\end{equation}
At this point we note that if $r(0)\ge \tfrac 12 M$, then the left-hand side of this equation is $\approx M^2\int \Im (\phi\overline{\partial_v \phi})$. So by modulating $\int \Im (\phi\overline{\partial_v\phi})$, we can hope to satisfy this equation just on the basis of scaling $\phi$ itself.

\begin{figure}[ht]
\centering{
\def\svgwidth{10pc}
\begingroup%
  \makeatletter%
  \providecommand\color[2][]{%
    \errmessage{(Inkscape) Color is used for the text in Inkscape, but the package 'color.sty' is not loaded}%
    \renewcommand\color[2][]{}%
  }%
  \providecommand\transparent[1]{%
    \errmessage{(Inkscape) Transparency is used (non-zero) for the text in Inkscape, but the package 'transparent.sty' is not loaded}%
    \renewcommand\transparent[1]{}%
  }%
  \providecommand\rotatebox[2]{#2}%
  \newcommand*\fsize{\dimexpr\f@size pt\relax}%
  \newcommand*\lineheight[1]{\fontsize{\fsize}{#1\fsize}\selectfont}%
  \ifx\svgwidth\undefined%
    \setlength{\unitlength}{200.01462441bp}%
    \ifx\svgscale\undefined%
      \relax%
    \else%
      \setlength{\unitlength}{\unitlength * \real{\svgscale}}%
    \fi%
  \else%
    \setlength{\unitlength}{\svgwidth}%
  \fi%
  \global\let\svgwidth\undefined%
  \global\let\svgscale\undefined%
  \makeatother%
  \begin{picture}(1,0.84943003)%
    \lineheight{1}%
    \setlength\tabcolsep{0pt}%
    \put(0,0){\includegraphics[width=\unitlength,page=1]{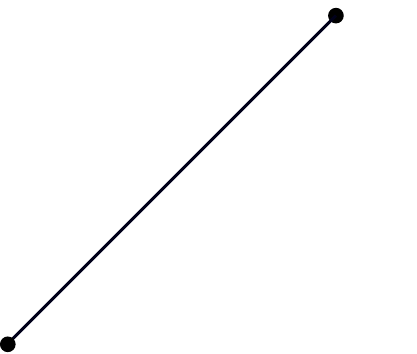}}%
    \put(0.00780951,0.37101765){\color[rgb]{0,0,0}\makebox(0,0)[lt]{\lineheight{1.25}\smash{\begin{tabular}[t]{l}$\alpha_1$\\\end{tabular}}}}%
    \put(0.27779716,0.60856881){\color[rgb]{0,0,0}\makebox(0,0)[lt]{\lineheight{1.25}\smash{\begin{tabular}[t]{l}$\alpha_2$\\\end{tabular}}}}%
    \put(0.17955727,0.00155271){\color[rgb]{0,0,0}\makebox(0,0)[lt]{\lineheight{1.25}\smash{\begin{tabular}[t]{l}Minkowski sphere\\\end{tabular}}}}%
    \put(0.76866685,0.59444863){\color[rgb]{0,0,0}\makebox(0,0)[lt]{\lineheight{1.25}\smash{\begin{tabular}[t]{l}RN horizon\\sphere\\\end{tabular}}}}%
    \put(0,0){\includegraphics[width=\unitlength,page=2]{three-bumps.pdf}}%
    \put(0.5215293,0.81480366){\color[rgb]{0,0,0}\makebox(0,0)[lt]{\lineheight{1.25}\smash{\begin{tabular}[t]{l}$\alpha_3$\\\end{tabular}}}}%
    \put(0,0){\includegraphics[width=\unitlength,page=3]{three-bumps.pdf}}%
  \end{picture}%
\endgroup%
}
\caption{Schematic illustration of the pulses.}
\label{Schw-gluing-diagram}
\end{figure}
Our ansatz for the scalar field will be 
\begin{equation*}
    \phi_\alpha=\sum_{1\le j\le 2k+1}\alpha_j\phi_j,
\end{equation*}
where $\alpha=(\alpha_1,\dotsc,\alpha_{2k+1})\in \Bbb R^{2k+1}$ and the $\phi_j$'s are smooth compactly supported complex-valued functions with disjoint supports. We assume $\mathfrak q\ne 0$ now, the $\mathfrak q=0$ case being in fact much easier. The charge condition \eqref{eqn:intro-charge-integral} is examined on every ray $\Bbb R_+\hat\alpha \in \Bbb R^{2k+1}$, $\hat\alpha\in S^{2k}$. We show that for a given choice of baseline profiles $\phi_j$, there is a smooth starshaped hypersurface $\mathfrak Q^{2k}\subset \Bbb R^{2k+1}$ which is isotopic to the unit sphere $S^{2k}$ and invariant under the antipodal map $\alpha\mapsto -\alpha$ such that \eqref{eqn:intro-charge-integral} holds for every $\alpha\in\mathfrak Q^{2k}$.

The condition that $M$ is large depending on $k$, $\mathfrak q$, and $\mathfrak e$ in \cref{thm-informal-statement} comes from natural conditions that arise when attempting to construct the hypersurface $\mathfrak Q^{2k}$. The charge condition \eqref{eqn:intro-charge-integral} implies $|\mathfrak e| M^2 |\alpha|^2\approx|\mathfrak q| M$ on $\mathfrak Q^{2k}$. However, to keep $r\ge \tfrac 12 M$ on $C_{-1}$, we find the condition $|\alpha|\les 1$, see already \cref{lem:radius-boostrap}. These conditions are consistent only if $|\mathfrak e|M\gtrsim |\mathfrak q|$. Furthermore, this condition is crucially used to propagate the condition $\partial_u r <0$, see already \cref{lem:no-antitrapped}.

The remaining equations ($2k$ real equations since the scalar field is complex) 
\begin{equation}
    \partial_u^i \phi_\alpha(1)=0\quad 1\le i\le k\label{phi-gluing-intro}
\end{equation}
can naturally be viewed as \emph{odd} equations as a function of $\alpha$. So when restricted to $\alpha\in \mathfrak Q^{2k}$, we can use the  classical Borsuk--Ulam theorem to find a simultaneous solution. Once we have an $\alpha\in\mathfrak Q^{2k}$ such that \eqref{phi-gluing-intro} is satisfied, $\phi_\alpha$ will glue all relevant quantities to $k$-th order, as desired. 

\begin{thm}[Borsuk--Ulam \cite{Borsuk1933}]\label{BorsukUlam}
If $f:S^k\to \Bbb R^k$ is a continuous odd function, i.e., $f(-x)=-f(x)$ for every $x\in S^k$, then $f$ has a root. 
\end{thm}

For a nice proof using only basic degree theory and transversality arguments, see Nirenberg's lecture notes \cite{Nirenberg}.

 \subsection{Retiring the third law of black hole thermodynamics}\label{retiring}
 
 In this section we give more details on the background and history of the third law of black hole thermodynamics put forth by Bardeen, Carter, and Hawking in \cite{BCH}, and how our present work fits into the picture. 
 
 While the \emph{zeroth}, \emph{first}, and \emph{second} laws of black hole thermodynamics are by now well understood  in the literature (see e.g.\ \cite{Wald-Living-reviews}), the validity of the third law has been a source of debate up until today.  
In the original form of Bardeen--Carter--Hawking (BCH), in analogy to  Nernst's version of the third law of classical thermodynamics \cite{Nernst}\footnote{For a discussion of various versions of the third law of classical thermodynamics, see \cite{Wald-Nernst}.}, it reads: 

\begin{center}
 \emph{It is impossible by any procedure, no matter how idealized, to reduce $\kappa$ to zero by a finite sequence of operations.}
\end{center}

A number of arguably pathological (e.g.\ singular or energy condition violating) examples of extremal black hole formation  were put forth in \cite{Kuchar, Delacruz-Israel, Proszynski,  Boulware, Farrugia-Hajicek,  Sullivan-Israel}, which Israel \cite{Israel-third-law, Israel-pdf}   took into account to make the third law more precise: 
\begin{center}
    \emph{A nonextremal black hole cannot become extremal
(i.e., lose its trapped surfaces) at a finite advanced time in any continuous process in which the stress-energy tensor of accreted matter stays bounded and satisfies the weak energy condition in a neighborhood of the outer apparent horizon.}
\end{center}

The parenthetical comment ``(i.e., lose its trapped surfaces)'' is an extra source of confusion which will be specifically addressed in \cref{sec:connectedness-a}.
We will now discuss the papers \cite{Kuchar, Delacruz-Israel, Proszynski,  Boulware, Farrugia-Hajicek,  Sullivan-Israel, Israel-third-law, Israel-pdf} and where the issues lie. 

\subsubsection{The singular massive dust shell model}\label{sec:singular-dust} It has been known since the 60's that an extremal black hole can be formed instantly by collapsing an infinitesimally thin shell of charged massive dust \cite{Kuchar, Delacruz-Israel, Proszynski,  Boulware}. Later, Farrugia and Hajicek \cite{Farrugia-Hajicek} showed how to ``turn a subextremal Reissner--Nordstr\"om spacetime into an extremal one'' by firing an appropriately charged singular massive shell into the black hole. The resulting spacetime metric is not $C^2$-regular. The Penrose diagram of the spacetime they construct is similar to our \cref{fig:first-page} (see \cite[p.~296 Fig.~2]{Farrugia-Hajicek}). In particular, we note the presence of a disconnected outermost apparent horizon in their example. Israel seemed to associate the disconnectedness of the apparent horizon with a singularity of the matter and/or spacetime: ``Violations can also be produced by any process that induces discontinuous behavior of the apparent horizon---for example, absorption of an infinitely thin massive shell, which will force this horizon to jump outward.'' See already \cref{sec:connectedness-a}.
On the basis of this, he dismissed this example in his formulation of the third law by explicitly requiring regularity. We note, however, that Farrugia and Hajicek suggest that their construction can in principle be desingularized---we do not know if this point was ever addressed again, because if true, it would seem to provide an alternative route to constructing a counterexample apart from our own.

\subsubsection{The charged null dust model} An interesting example motivating explicit mention of the weak energy condition in the third law was provided by Sullivan and Israel \cite{Sullivan-Israel} in spherical symmetry, with the charged null dust matter model. This matter model allows for \emph{dynamical} violations of the weak energy condition---even if the initial data satisfies the weak energy condition, the solution might violate it in the future. Sullivan and Israel showed that extremization is impossible in this model without such a violation, which can also be seen from Penrose diagrams. They interpreted this result as further evidence that the third law holds as long as the weak energy condition is demanded near the apparent horizon. We note, however, that Ori has proposed a different interpretation of the model studied by Sullivan and Israel which does not violate the weak energy condition \cite{Ori91}.

\subsubsection{``Losing trapped surfaces'' and connectedness of the outermost apparent horizon} 
\label{sec:connectedness-a}

We will now  clarify the issue of ``losing trapped surfaces'' appearing prominently in \cite{Israel-third-law, Israel-pdf} and the implicit assumption of connectedness of the outermost apparent horizon.

The black hole region in a subextremal Reissner--Nordstr\"om or Kerr spacetime is foliated by trapped spheres. Conversely, extremal Reissner--Nordstr\"om and Kerr black holes have no trapped surfaces, but the event horizon is a marginally trapped tube in both cases. As $|\mathfrak q|\to 1$ (where we take $\mathfrak q\doteq e/M$ for Reissner--Nordstr\"om and $\mathfrak q\doteq a/M$ for Kerr), $r_-\to r_+$, and one might be inclined to think that extremizing involves ``squeezing'' away the trapped region inside the black hole.  
However, it is an immediate consequence of Raychaudhuri's equation \cite{HE73, Wald-GR} that trapped surfaces persist in evolution as long as the spacetime satisfies the weak energy condition. 
 Since the typical explicit extremal black holes have no trapped surfaces (in particular none near the event horizon),
  one might wonder if Raychaudhuri's equation alone could be used to ``prove'' the third law. 
  
  This is what Israel attempted to do in \cite{Israel-third-law, Israel-pdf}. We will formalize his observation in \cref{def:becoming-extremal} and \cref{Israel-observation} below.
  However, as should be clear from our main theorem, this does not in fact capture the intended meaning of the third law.

In order to reconstruct Israel's argument mathematically, let us formulate the following definition. For precise definitions relating to spherical symmetry, see already \cref{sec:prelims}.

\begin{defn}\label{def:becoming-extremal}
Let $H$ be a connected dynamical apparent horizon, i.e., a connected, achronal curve in the $(1+1)$-dimensional reduction $(\mathcal Q,g_\mathcal Q)$ of a spherically symmetric spacetime $(\mathcal M,g)$, along which $\partial_vr $ vanishes identically. We say that $H$ \emph{becomes extremal in finite time in the sense of Israel} if
\begin{enumerate}
    \item $H$ is not completely contained in a null cone. 
    \item Let $\tau\mapsto H(\tau)$ be a parametrization of $H$. Then there exists a $\tau_0\in\Bbb R$ so that for all $\tau\ge \tau_0$, $\tau\mapsto H(\tau)$ is a future-directed constant $u$ curve. 
    \item There exists a $\tau_1>\tau_0$ and a neighborhood $\mathcal N$ of $H_{\tau\ge \tau_1}$ such that $\mathcal N\setminus H_{\tau\ge \tau_0}$ contains only strictly untrapped spheres ($\partial_v r>0$). 
\end{enumerate}
\end{defn}

\begin{rk}
The outermost apparent horizon $\mathcal A'$ (see already \cref{MFGHD}), if connected, is an example of a connected dynamical apparent horizon. 
\end{rk}

As a simple consequence of Raychaudhuri's equation in a spacetime satisfying the weak energy condition \cite{HE73, Wald-GR}, we have

\begin{prop}[Israel's observation]\label{Israel-observation}
Let $(\mathcal M,g)$ be a spherically symmetric black hole spacetime. If the spacetime satisfies the weak energy condition, has a nonempty trapped region, and a \ul{connected} outermost apparent horizon $\mathcal A'$ as defined in \cite{Kommemi13}, then the outermost apparent horizon $\mathcal A'$ does not become extremal in finite time in the sense of Israel. 
\end{prop}

However, it is clear that in view of our main theorem, the correct reading of this proposition is the \emph{contrapositive}, namely that  violations of the third law necessarily have a disconnected apparent horizon. This effect has nothing to do with singularities of spacetime or the matter model (and there was never actually any \emph{a priori} reason to believe that the outermost apparent horizon was connected). This situation is depicted in \cref{fig:disconnected}.  

 \begin{figure}[ht]
\centering{
\def\svgwidth{12pc}
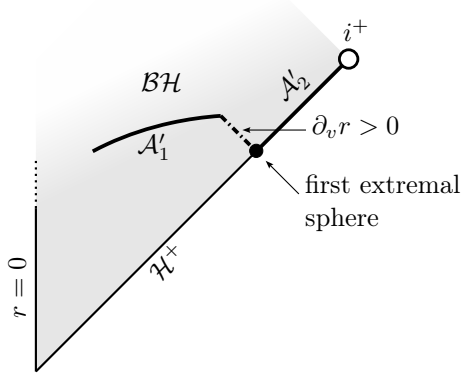}
\caption{Illustration of the contrapositive of \cref{Israel-observation}. The outermost apparent horizon $\mathcal A'=\mathcal A'_1\cup \mathcal A_2'$ becomes disconnected when a black hole with trapped surfaces ``becomes extremal,'' while the spacetime and matter fields remain regular. The trapped region begins to the north of $\mathcal A_1'$ and persists for all advanced time.}
\label{fig:disconnected}
\end{figure}

\subsubsection{Disproving the third law}\label{sec:disproof}

With this discussion out of the way, we present now a detailed version of our counterexample to the third law. It is essentially a corollary of the more general version of our main gluing result \cref{thm-informal-statement} with a Schwarzschild exterior sphere in place of a Minkowski sphere (see already \cref{sec:precise-version-main-theorems}) and will be given in \cref{proof-trapped-surfaces-ERN}. For an illustration of the spacetime, we refer the reader back to \cref{fig:first-page}.
\setcounter{thm}{0}
\begin{thm}[Gravitational collapse to ERN with a Schwarzschild piece]\label{cor:trapped-surfaces-ERN}
For any regularity index $k\in \Bbb N$, there exist spherically symmetric, asymptotically flat Cauchy data for the Einstein--Maxwell-charged scalar field system, with $\Sigma\cong \Bbb R^3$ and a regular center, such that the maximal  future globally hyperbolic development $(\mathcal M^4,g)$ has the following properties: 
 \begin{itemize}
 \item The spacetime satisfies all the conclusions of \cref{main-corollary} with $\mathfrak q=1$, including $C^k$-regularity of  all dynamical quantities. 
 
     \item The black hole region contains an isometrically embedded portion of a Schwarzschild exterior horizon neighborhood. In particular, there is a portion of a null cone behind the event horizon of $(\mathcal M,g)$ which can be identified with a portion of the apparent horizon of Schwarzschild.
     
     \item The ``Schwarzschild horizon'' piece is a part of the outermost apparent horizon $\mathcal A'$ of the spacetime. The set $\mathcal A'$ is disconnected and agrees with the event horizon $\mathcal H^+$ to the future of the first marginally trapped sphere on the event horizon.
     
     \item There is a neighborhood of the event horizon that contains no trapped surfaces. Nonetheless, the black hole region contains trapped surfaces. In fact, there are trapped surfaces at arbitrarily late advanced time in the interior of the black hole. 
 \end{itemize}
 \end{thm}
 
To reiterate, the scalar field collapses to form an \emph{exact} Schwarzschild spacetime, including the horizon, only to collapse further to form an \emph{exact} extremal Reissner--Norstr\"om for all late advanced time. The spacetime is regular (for any fixed $k\ge1$, one can construct an example which is $C^k$) and the matter model satisfies the dominant energy condition.

\subsubsection{Exceptionality and stability of third law violating solutions}
\label{subsubsec:exceptionality-of-third-law-violations}

The third law is manifestly concerned with exceptional behavior, which is why the phrases ``no matter how idealized'' \cite{BCH} or ``in any continuous process'' \cite{Israel-third-law} are specifically included in formulations of the third law. 
Indeed, keeping a horizon at exactly constant temperature (or equivalently constant surface gravity), any temperature, is of course exceptional. (Exactly stationary behavior on the horizon for all late advanced times is itself an infinite codimension phenomenon in the moduli space of solutions.) In view of our construction, the case of gravitational collapse to zero temperature in finite time is no more exceptional than any other \underline{fixed} temperature. 

We would also like to address the interesting question of whether creating \emph{asymptotically} extremal black holes should be viewed any differently from the subextremal case. Indeed, any mechanism which forms a black hole with \emph{exactly specified} parameters is inherently unstable, because a small perturbation can just change the parameters. As an example of this, we note the codimension-3 nonlinear stability of the Schwarzschild family by Dafermos--Holzegel--Rodnianski--Taylor \cite{DHRT}. In order to preserve the final black hole parameters, only a codimension-3 submanifold of the moduli space of data is admissible in their theorem.

The stability problem for extremal black holes is exceptional because they suffer from a linear instability known as the \emph{Aretakis instability} \cite{Aretakis-instability-1, Aretakis-instability-2, Aretakis-instability-3,apetroaie}. This instability is weak, and a restricted form of nonlinear stability is nevertheless conjectured to hold with the same codimensionality as in the subextremal case. See \cite[Section IV.2]{DHRT} for conjectures about stability of extremal black holes, \cite{A16, AAG20} for stability results on a nonlinear model problem, and numerical work \cite{Reall-numerical,LMR13} which is consistent with the above conjecture. The Aretakis instability should not be thought of as a manifestation of the third law and understanding its ramifications in the full nonlinear theory is a fundamental open problem in general relativity.

Therefore, asymptotic stability for any \underline{fixed} parameter ratio (up to and including extremality) should be formulated as a positive codimension statement. In our spherically symmetric setting, we are led to conjecture that for every solution constructed in \cref{main-corollary}, there exists a codimension-1 family of perturbations which asymptote to a Reissner--Nordstr\"om black hole with the same final parameter ratio. Since the conjectured codimension is the same for every ratio, we are then led to conclude that asymptotically extremal black holes are not qualitatively rarer than any fixed positive temperature.

    We hope that our construction demystifies the scenario of matter collapsing to exactly extremal black holes. More generally, the considerations in this paper open up a new window to studying critical behavior in gravitational collapse, which is fundamentally different from the regime studied in \cite{choptuik1993universality,Gundlach2007}.

\subsubsection{Aside: Extremal horizons with nearby trapped surfaces}
\label{rk:no-trapped-surfaces}
Though not directly relevant for the considerations of the present paper, we would like to point out that there is another issue with the attempt to characterize extremality by the lack of trapped surfaces near the horizon, i.e., by the third property of \cref{def:becoming-extremal}. In fact, it would appear that the property of having no trapped surfaces in the interior near the horizon is actually stronger than being extremal.

For a spacetime $(\mathcal M,g)$ with Killing field $K$, a Killing horizon $H$ is said to be \emph{extremal} if the surface gravity $\kappa$, defined by $\nabla_K K=\kappa K$ on $H$, vanishes identically. Equivalently, extremality means that $g(K,K)$ vanishes to at least second order along null geodesics crossing $H$ transversely. If $K$ is timelike to the past of $H$ and $g(K,K)$ vanishes to an \emph{even} order on $H$, then $K$ passes from timelike, to null, then back to timelike across $H$, and there are no strictly trapped surfaces near the horizon. This is precisely the situation for extremal Reissner--Nordstr\"om and Kerr black holes, where $g(K,K)$ vanishes to second order on the event horizon. 

However, there exist spacetimes for which $g(K,K)$ vanishes to an \emph{odd} order (at least three), in which case there may be trapped surfaces just behind the horizon. Indeed, in \cref{prop:isolated-extremal-horizon} of \cref{app:A} we construct an example of a stationary spacetime containing an extremal Killing horizon, with trapped surfaces just behind the horizon, and \ul{satisfying the dominant energy condition}. In this case $g(K,K)$ is exactly cubic in an ingoing null coordinate system. It would be interesting to construct such a spacetime with a specific matter model, or an extremal black hole with this behavior. 

While extremal Kerr, Reissner--Nordstr\"om, and other known examples are extremal in the sense of \cref{def:becoming-extremal}, it is far from obvious that all \emph{hairy} (i.e., carrying non-EM matter fields) extremal black holes should be free of trapped surfaces. In view of our example in \cref{app:A}, any mechanism which enforces this must necessarily be global in nature and/or depend on particular properties of the matter model in question. 

One could define the notion of a \emph{nondegenerate} extremal Killing horizon, i.e., the Killing field $K$ has the property that $g(K,K)$ vanishes only to second order, which would then be compatible with \cref{def:becoming-extremal}. See already \cref{rk:App1}. 

 For more discussion about possible definitions of extremality, see for instance \cite{Booth08, Booth16, Reall-numerical}.

 \subsection{Future boundary of the interior and Cauchy horizon gluing}\label{sec:interior}
The future boundary of the black hole region of dynamical black holes formed from gravitational collapse in the EMCSF system is known to be intricate (see e.g.\ \cite{MR2031855,Kommemi13,VdM18}). We refer to \cite{Kommemi13} for a detailed description of the most general possible structure of the interior, but see already \cref{Kommemi} for a summary of the most salient features. In this subsection we will first discuss the future boundary of the black hole interior in \cref{cor:trapped-surfaces-ERN}. Further, we will present
additional corollaries of our characteristic gluing method which provide examples of gravitational collapse to black holes with a piece of null boundary (a ``Cauchy horizon'') and a construction of spacetimes for which a Cauchy horizon closes off the interior region.

\subsubsection{Future boundary of the interior in \texorpdfstring{\cref{cor:trapped-surfaces-ERN}}{Theorem 1}}\label{figure-1-comment}
 For our main counterexample to the third law in \cref{cor:trapped-surfaces-ERN}, we obtain that the regular center $\Gamma$ extends into the black hole region. Regarding the future boundary of the spacetime, we do not know whether there exists a piece of possibly singular null boundary emanating from $i^+$ as in the subextremal case \cite{MR2031855,VdM18} or whether a  spacelike singularity emanates from $i^+$. Note that the result of  \cite{gajic-luk}, which shows the existence of a  Cauchy horizon emanating from $i^+$, does not apply directly since their analysis requires $|\mathfrak e |M\le 0.1$, whereas our construction requires $|\mathfrak e |M$ large. Nevertheless, one may speculate that a piece of Cauchy horizon occurs (for which the linear analysis of \cite{gajic1,gajic2} would be relevant), which could eventually turn into a spacelike singularity. (Note that one can readily set up the data such that the future boundary of the interior in \cref{cor:trapped-surfaces-ERN} has a piece of spacelike singularity. See however already \cref{subsubsec:BHinterior-CH-closes-off}.)

\subsubsection{Gravitational collapse with a piece of smooth Cauchy horizon} 
Another corollary of our method is the construction of regular one-ended Cauchy data which evolve to a subextremal or extremal black hole for which there exists a piece of Cauchy horizon emanating from $i^+$. We refer to \cref{trapped-surface-diagram} for the Penrose diagram of the spacetime constructed in \cref{cor:Cauchy-horizon-RN}. The proof of \cref{cor:Cauchy-horizon-RN} is given in \cref{subsec:collapse-with-CH}.

\begin{figure}[ht]
\centering{
\def\svgwidth{20pc}
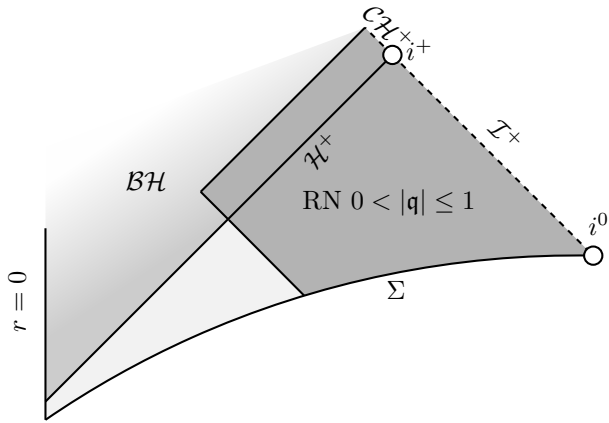}
\caption{Penrose diagram depicting \cref{cor:Cauchy-horizon-RN}: Gravitational collapse to Reissner--Nordström with nonempty piece of Cauchy horizon $\mathcal{CH}^+$.}
\label{trapped-surface-diagram}
\end{figure}

\begin{cor}[Gravitational collapse to RN with a smooth Cauchy horizon]\label{cor:Cauchy-horizon-RN}
For any regularity index $k\in \Bbb N$ and nonzero charge to mass ratio $\mathfrak q\in [-1,1]\setminus\{0\}$, there exist spherically symmetric, asymptotically flat Cauchy data for the Einstein--Maxwell-charged scalar field system in spherical symmetry, with $\Sigma\cong \Bbb R^3$ and a regular center, such that the maximal  future globally hyperbolic development $(\mathcal M^4,g)$ has the following properties: 
 \begin{itemize}
 \item The spacetime satisfies all the conclusions of \cref{main-corollary} with $\mathfrak q\ne 0$, including $C^k$-regularity of  all dynamical quantities. 
 \item The black hole region contains an isometrically embedded portion of a Reissner--Nordström Cauchy horizon neighborhood with charge to mass ratio $\mathfrak q$.
 \end{itemize}
 \end{cor}

\begin{rk}
When $|\mathfrak q|=1$, the spacetime constructed in \cref{cor:Cauchy-horizon-RN} does not contain trapped symmetry spheres in the dark shaded region in \cref{trapped-surface-diagram}. By a slight modification of the argument in \cref{no-trapped} below, this implies no trapped surfaces \emph{intersect} the dark shaded region. In particular, the trapped region (in the sense of \cite[p.~319]{HE73}) of the spacetime (if nonempty) avoids a whole double null neighborhood of the event horizon. Nevertheless, the event horizon agrees with the outermost apparent horizon for late advanced times. 
\end{rk}

\subsubsection{Black hole interiors for which the Cauchy horizon closes off spacetime} \label{subsubsec:BHinterior-CH-closes-off}
Our horizon gluing method can also be extended to glue  Reissner--Nordström interior spheres to a regular center along an ingoing cone, see already  \cref{thm:backwards-gluing}.  Using this, we construct asymptotically flat Cauchy data for which the future boundary of the black hole region $\mathcal{BH}$ is a Cauchy horizon $\mathcal{CH}^+$ which closes off spacetime.  We refer to \cref{spacetime-closed-off} for the Penrose diagram of the spacetime constructed in \cref{cor:Cauchy-horizon-closes--off}. The proof of \cref{cor:Cauchy-horizon-closes--off} is given in \cref{sec:closes-off-spacetime}.

  \begin{figure}[ht]
\centering{
\def\svgwidth{15pc}
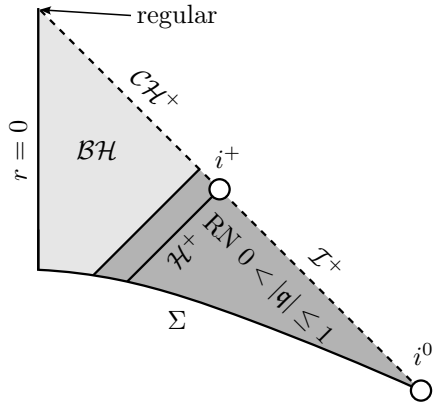}
\caption{Penrose diagram depicting \cref{cor:Cauchy-horizon-closes--off}: The Cauchy horizon is regular and closes off the spacetime in a regular fashion.}
\label{spacetime-closed-off}
\end{figure}

\begin{cor}[Cauchy horizon that closes off the spacetime]\label{cor:Cauchy-horizon-closes--off}
For any regularity index $k\in\Bbb N$, and nonzero charge to mass ratio $\mathfrak q\in [-1,1]\setminus\{0\}$, there exist spherically symmetric, asymptotically flat Cauchy data for the Einstein--Maxwell-charged scalar field system, with $\Sigma\cong \Bbb R^3$ and a regular center, such that the maximal  future globally hyperbolic development $(\mathcal M^4,g)$ has the following properties: 
  \begin{itemize}
\item All dynamical quantities are at least $C^{k}$-regular.
      \item The black hole region is non-empty, $\mathcal B\mathcal H\doteq \mathcal M\setminus J^-(\mathcal I^+)\neq \emptyset$.
      \item The future boundary of $\mathcal{BH}$ is a $C^k$-regular Cauchy horizon  $\mathcal{CH}^+$ which closes off spacetime. 
     \item  The black hole exterior is isometric to a Reissner--Nordström exterior with charge to mass ratio $\mathfrak q$. In particular, null infinity $\mathcal I^+$ is complete. 
     \item The spacetime does not contain antitrapped surfaces. 
     \item When $|\mathfrak q|=1$, the spacetime does not contain trapped surfaces. 
   \end{itemize}
 \end{cor}

 \begin{rk}
 In contrast to our previous constructions, the Cauchy surface $\Sigma$ in \cref{cor:Cauchy-horizon-closes--off}  could contain trapped surfaces and $\Sigma$ intersects the black hole region.  It would be interesting to construct a spacetime as in \cref{cor:Cauchy-horizon-closes--off} which depicts genuine gravitational collapse, i.e., for which $\Sigma\subset J^-(\mathcal I^+)$. 
 \end{rk}

In the subextremal case, the behavior exhibited by our construction can be seen as exceptional as one generically expects a Cauchy horizon which forms in gravitational collapse  to be a weak null singularity \cite{MR2031855,VdM18,LukOhI}. In particular, in the case where the Cauchy horizon $\mathcal{CH}^+$ is weakly singular,  Van~de~Moortel  \cite{Van_de_Moortel2019} showed that the Cauchy horizon $\mathcal{CH}^+$ cannot close off spacetime in the sense of \cref{spacetime-closed-off}. Thus, our construction in \cref{cor:Cauchy-horizon-closes--off} makes \cite{Van_de_Moortel2019} sharp in the sense that the singularity assumption of $\mathcal{CH}^+$ in  \cite{Van_de_Moortel2019} is needed.   
Restricted to the extremal case, however, on the basis of a more regular Cauchy horizon as in \cite{gajic-luk}, one may speculate that there exists a set of data (open as a subset of the positive codimension set of data settling down to ERN) for which the Cauchy horizon closes off spacetime as depicted in \cref{spacetime-closed-off}. 

 \FloatBarrier
 \subsection*{Acknowledgments}
The authors wish to express their gratitude to Mihalis Dafermos for suggesting the problem and many helpful discussions. We also thank Jonathan Luk, Hamed Masaood, Thomas Massoni, Georgios Moschidis, Harvey Reall, Igor Rodnianski, Jaydeep Singh, and Nina Zubrilina for helpful conversations.
C.K.~acknowledges support by a grant from the Institute for Advanced Study, by Dr.~Max Rössler, the Walter Haefner Foundation, and the ETH Zürich Foundation.  R.U.~acknowledges partial support from the grant NSF-1759835 and the hospitality of the Institute for Advanced Study, the Centro de Ciencias de Benasque, and the University of Cambridge.

\section{The characteristic initial value problem for the Einstein--Maxwell-charged scalar field system in spherical symmetry}
\label{sec:prelims}

In this section, we give a detailed explanation of the setup and characteristic initial value problem for the Einstein equations with charged scalar fields in spherical symmetry, with a view towards the characteristic gluing problem. See \cite{Kommemi13} for more details on the EMCSF system. 

\subsection{The Einstein--Maxwell-charged scalar field system in spherical symmetry}
\label{subsec:einstein-maxwell-sph-symm}
\subsubsection{Spherically symmetric spacetimes}

We say that a smooth, connected, time-oriented, four-dimensional Lorentzian manifold $(\mathcal M,g)$ is a \emph{spherically symmetric} spacetime with (possibly empty) center of symmetry $\Gamma \subset \mathcal M$ if $\mathcal M \setminus \Gamma $ splits diffeomorphically as $\mathcal M \setminus \Gamma \cong \mathring{\mathcal Q}\times S^2$ with metric 
\begin{align*}
    g = g_\mathcal{Q} + r^2 g_{S^2},
\end{align*}
where $(\mathcal Q,g_\mathcal Q)$ for $\mathcal  Q = \mathring{\mathcal Q} \cup \Gamma$ is a (1+1)-dimensional Lorentzian spacetime with (possibly empty) boundary $\Gamma$, $g_{ S^2}$ is the round metric on the sphere, and $r\colon\mathcal Q \to \R_{\geq 0}$ is a function which can be geometrically interpreted as the area radius of the orbits of the isometric $\mathrm{SO}(3)$ action on  $(\mathcal M,g)$. In mild abuse of notation, we denote with $\Gamma$ both the center of symmetry in $\mathcal M$ and its projection to $\mathcal Q$.  Moreover, if  $\Gamma$ is non-empty, we assume that the  $\mathrm{SO}(3)$ action fixes $\Gamma$ and that $\Gamma$ consists of one timelike geodesic along which $r=0$. 
We further assume that $(\mathcal Q,g_\mathcal Q)$ admits a global double-null foliation  (locally a double-null foliation always exists) with null coordinates $(u,v)$ such that the metric $g$ takes the form
\begin{align*}
g = -\Omega^2 du dv + r^2  g_{ S^2}
\end{align*}
for a nowhere vanishing function $\Omega^2= -2g_\mathcal{Q}(\partial_u,\partial_v)$ on $\mathcal Q$ and such that $\partial_u$ and $\partial_v$ are future-directed. We further assume that along the center $\Gamma$, the coordinate $v$ is outgoing and $u$ is ingoing, i.e.,  $\partial_v r \vert_\Gamma >0$, $\partial_u r\vert_\Gamma <0$.     While we introduced the above notions in the smooth category, we will also consider spacetimes which are less regular $(C^{k\geq 1})$. We note that all notions introduced above also apply in this less regular case. We will also make use of the Hawking mass $m \colon \mathcal M  \to \mathbb R$ defined as \begin{equation*}
    m \doteq \frac r 2 (1 - g(\nabla r, \nabla r))
\end{equation*}
which also can be viewed  as a function on $\mathcal Q$ as 
\begin{equation}\label{eq:Hawking-mass}
    m = \frac r2 \left( 1+  4\frac{\partial_ur\partial_vr}{\Omega^2}\right). 
\end{equation}

Finally, we note that the double null coordinates $(u,v)$ are not unique and for any smooth functions $U,V\colon \R \to \R$ with $U' , V' > 0$, we obtain new global double null coordinates $(\tilde u , \tilde v)   =(U(u), V(v))$ such that the metric $g= -\tilde \Omega^2 d \tilde u d \tilde v + r^2 g_{ S^2}$, where $\tilde \Omega^2(\tilde u , \tilde v) = (U' V')^{-1} \Omega^2(U^{-1}(\tilde u), V^{-1}(\tilde v)) $ and $r (\tilde u, \tilde v) = r(U^{-1} (\tilde u), V^{-1} (\tilde v))$.

\subsubsection{The Einstein--Maxwell-charged~scalar~field system}

For the Einstein--Maxwell-scalar field system \eqref{eq:Einstein-Maxwell-1}--\eqref{TCSF} in spherical symmetry, additionally to the spherically symmetric spacetime $(\mathcal M^4, g)$, we assume that the field $\phi$ is complex-valued and spherically symmetric, so that $\phi$ descends to a function $\mathcal Q\to \Bbb C$, and that   $F$ and $A$ are  spherically symmetric such that $F$ can be written as \begin{equation*}
    F =\frac{Q}{2r^2} \Omega^2 du \wedge dv \label{eq:Q-defn}
\end{equation*}
for charge  $Q : \mathcal Q \to \R$. The potential $1$-form reads
\begin{align}
    A= A_u du + A_v dv.\nonumber
\end{align} We also define the gauge covariant derivative operator $D=d+i\mathfrak e A$.
The Einstein--Maxwell-scalar field system is invariant with respect to the following gauge transformations 
\begin{equation}
    \phi \mapsto e^{-i \mathfrak e \chi} \phi, \quad A \mapsto A +d\chi \label{eq:EM-gauge-trafo}
\end{equation}
for real-valued functions $\chi=\chi(u,v)$, where $\mathfrak e$ is a dimensionful coupling constant representing the charge of the scalar field. More abstractly, the Einstein--Maxwell-scalar field system is a $U(1)$-gauge theory and we refer to \cite{Kommemi13} for more details. In order to break the symmetry  we will use the global electromagnetic gauge 
\begin{equation}
    A_v =0\label{EM-gauge}
\end{equation}
throughout the paper. In this gauge, the Einstein--Maxwell-scalar~field system \eqref{eq:Einstein-Maxwell-1}--\eqref{eq:Einstein-Maxwell-3} in spherical symmetry  reduces to the following set of equations. 

\vspace{3mm}

\noindent \textbf{Wave equations for scalar field and metric components:}
\begin{align}\label{eq:CSF-wave-equation}
&\partial_u \partial_v \phi = - \frac{\partial_u\phi\partial_v r}{r} - \frac{\partial_u r\partial_v\phi}{r} + \frac{i \mathfrak e \Omega^2 Q}{4 r^2} \phi - i \mathfrak e A_u \frac{\partial_v r}{r} \phi - i \mathfrak e A_u \partial_v \phi \\
\label{eq:CSF-equation-for-r}
&\partial_u \partial_v r = -\frac{\Omega^2}{4r} - \frac{\partial_ur \partial_vr}{r} + \frac{\Omega^2}{4 r^3} Q^2 \\
\label{eq:CSF-equation-for-Omega}
& \partial_u\partial_v\log(\Omega^2)= \frac{\Omega^2}{2r^2}+2\frac{\partial_u r\partial_v r}{r^2}   - \frac{\Omega^2}{r^4} Q^2 - 2 \Re(D_u \phi \overline{\partial_v \phi})
\end{align}
\textbf{Maxwell's equations:}
 \begin{align}
     &\partial_u Q = - \mathfrak e r^2 \Im(\phi \overline{D_u \phi})\label{eq:Maxwell-u}\\
     &\partial_v Q =  \mathfrak e r^2 \Im(\phi \overline{\partial_v \phi})\label{eq:Maxwell-v}\\
     &\partial_v A_u =  -\frac{Q \Omega^2}{2r^2}\label{eq:Au-transport}
 \end{align}
\textbf{Raychaudhuri's equations:}
\begin{align}
\partial_u\left(\frac{\partial_u r}{\Omega^2}\right)&=-\frac{r}{\Omega^2}|D_u\phi|^2 \label{eq:CSF-Raychaudhuri-u}\\
\label{eq:CSF-Raychaudhuri-v}
\partial_v\left(\frac{\partial_v r}{\Omega^2}\right)&=-\frac{r}{\Omega^2}|\partial_v\phi|^2
\end{align}
From these equations we easily derive
\begin{equation}\label{rdur}
    \partial_v(r\partial_u r)= -\frac{\Omega^2}{4}\left(1-\frac{Q^2}{r^2}\right)
\end{equation}
and
\begin{align}
    \partial_v \partial_u (r \phi) = -\frac{\Omega^2 m}{2 r^2}\phi +i \frac{\mathfrak e \Omega^2 Q}{4 r} \phi  +\frac{\Omega^2 Q^2}{4r^3}\phi  - i \mathfrak e A_u\partial_v (r \phi ),\label{rphi2}
\end{align}
as well as 
\begin{equation}\label{eq:equation-for-hawking-mass}
    \partial_v m = 2\Omega^{-2} r^2 (-\partial_u r) |\partial_v \phi|^2 +  \frac 12 \frac{Q^2}{r^2} \partial_v r 
\end{equation}
which will be useful later.

\subsection{The characteristic initial value problem}\label{subsubsec:char-data}

With the equations of the EMCSF system at hand, we can precisely define what we mean by a $C^k$ solution. We may for now restrict attention to solutions away from the center. 
\subsubsection{Bifurcate characteristic data}

\begin{defn}\label{defn:Ck-solution}
Let $k\in \Bbb N$. A \emph{$C^k$ solution} for the Einstein--Maxwell-charged scalar field system in the EM gauge \eqref{EM-gauge} consists of a domain $\mathcal Q\subset \Bbb R^{1+1}_{u,v}$ and functions $r\in C^{k+1}(\mathcal Q)$ and $\Omega^2,\phi,Q, A_u\in C^k(\mathcal Q)$, such that $r>0$, $\Omega^2>0$, $\phi$ is complex-valued, $\partial_v^{k+1}A_u\in C^0(\mathcal Q)$, and the functions satisfy\footnote{Note that the wave equations \eqref{eq:CSF-wave-equation} and \eqref{eq:CSF-equation-for-Omega} can readily be interpreted for $k=1$.} equations \eqref{eq:CSF-wave-equation}--\eqref{eq:CSF-Raychaudhuri-v}.
\end{defn}

Next, we formulate the characteristic initial value problem for this class of solutions. Let $\mathbb R^{1+1}_{u,v}$ denote the standard $(1+1)$-dimensional Minkowski space. We introduce the \emph{bifurcate null hypersurface} $C\cup\underline C\subset \mathbb R^{1+1}_{u,v}$, where 
\begin{align*}
    C&\doteq C_{-1} \doteq \{u=-1\}\cap \{v\ge 0\} \\
    \underline C&\doteq \underline C_0 \doteq \{v=0\}\cap \{u\ge -1\}.
\end{align*}
The special point $(-1,0)$ is called the \emph{bifurcation sphere}. We pose data for $\phi$, $Q$, $r$, $\Omega^2$ and $A_u$ for the Einstein--Maxwell-charged-scalar~field~system on $C\cup \underline C$. 

\begin{defn}\label{defn:ESF-char-data} Let $k\in\Bbb N$. 
A $C^k$ \emph{bifurcate characteristic initial data set} on $C\cup\underline C$ for the Einstein--Maxwell-charged scalar field system in the EM gauge \eqref{EM-gauge} consists of continuous functions $r>0$, $\Omega^2>0$, $\phi$ (complex-valued), $Q$, and $A_u$ on $C\cup\underline C$. It is required that $r\in C^{k+1}$, $\Omega^2\in C^k$, $\phi\in C^k$, $Q\in C^k$, and $A_u\in C^{k}$ on $C\cup\underline C$.\footnote{By ``$C^k$ on $C\cup\underline C$'' is meant that $v$ derivatives are continuous on $C$ and $u$ derivatives are continuous on $\underline C$.} Finally, the data are required to satisfy equations \eqref{eq:Maxwell-u}--\eqref{eq:CSF-Raychaudhuri-v}, which implies also $\partial_v^{k+1}A_u\in C^0(C)$.
\end{defn}

Given characteristic initial data on a portion of $C\cup\underline C$ containing the bifurcation sphere, we can solve in a full double null neighborhood to the future. The proof is a standard iteration argument.

\begin{prop}\label{local-existence-ss}
Given a $C^k$ bifurcate characteristic initial data set for the EMCSF system on 
\begin{equation*}
   \left( \{u=-1\}\times \{0\le v\le v_0\}\right)\cup\left(\{-1\le u\le u_0\}\times \{v=0\}\right)\subset C\cup\underline C,
\end{equation*}
where $u_0>-1$ and $v_0>0$, there exists a number $\delta>0$ and a unique spherically symmetric $C^k$ solution of the EMCSF system on 
\begin{equation*}
    \left(\{-1\le u \le -1+\delta\}\times \{0\le v\le v_0\}\right)\cup\left(\{-1\le u\le u_0\}\times \{0\le v\le \delta\}\right)
\end{equation*}
which extends the initial data on $C\cup\underline C$.
\end{prop}

\subsubsection{Determining transversal derivatives from tangential data}

Now that we know that the data on $C\cup\underline C$ extends to a solution of the system \eqref{eq:CSF-wave-equation}--\eqref{eq:CSF-Raychaudhuri-v}, we can use the equations to compute all the partial derivatives of the solution along $C\cup\underline C$. We describe a procedure for determining all $u$-derivatives on $C$ just in terms of $r,\Omega^2,\phi,Q$, and $A_u$ (as functions of $v$) and their $u$ derivatives at the bifurcation sphere.

\begin{prop}\label{prop:determining-everything}
Let $(r,\Omega^2,\phi,Q,A_u)$ be a $C^k$ \underline{bifurcate} characteristic initial data set as in \cref{defn:ESF-char-data}. Then the EMCSF system can be used to determine as many $u$-derivatives of $r,\Omega^2,\phi,Q$, and $A_u$ on $C$ as is consistent with \cref{defn:Ck-solution}, \ul{explicitly from the data on $C\cup\underline C$}.
\end{prop}
\begin{proof}
Since $(r,\Omega^2,\phi,Q,A_u)$ are all given on $\underline C$, we can compute as many $u$-derivatives of these quantities at the bifurcation sphere $(-1,0)$ as the regularity $k$ allows. We describe an inductive procedure for computing $u$-derivatives of $(r,\Omega^2,\phi,Q,A_u)$ on $C$, starting with $\partial_u r$. Since $\partial_u r(-1,0)$ is known, and the wave equation \eqref{eq:CSF-equation-for-r} can be written as
\begin{equation}
    \left(\partial_v  + \frac{\partial_vr}{r}\right)\partial_ur= -\frac{\Omega^2}{4r}+\frac{\Omega^2}{4r^3}Q^2, \nonumber
\end{equation}
where everything on the right-hand side is already known, $\partial_ur(-1,v)$ can be found by solving this ODE. In the same manner, $\partial_u\phi(-1,v)$ and then $\partial_u\log(\Omega^2)(-1,v)$ can  be found. To find $\partial_u Q(-1,v)$, differentiate \eqref{eq:Maxwell-u} in $v$ and then integrate. (Alternatively, differentiate \eqref{eq:Maxwell-v} in $u$.) Finally, $\partial_u A_u(-1,v)$ is found by differentiating \eqref{eq:Au-transport} in $v$ and then integrating. 

Proceeding in this way, by commuting all the equations with $\partial_u^i$, every partial derivative of $(r,\Omega^2,\phi,Q,A_u)$ which is consistent with the initial $C^k$ regularity can be found. We finally note that $\partial_u^{k+1}r(-1,v)$ is found from differentiating \eqref{eq:CSF-Raychaudhuri-u} an appropriate number of times, since the wave equation it satisfies is not consistent with the level of regularity of the rest of the dynamical variables. 
\end{proof}

\begin{rk}
Both \cref{local-existence-ss} and \cref{prop:determining-everything} exploit the \emph{null condition} satisfied by the EMCSF system in double null gauge. For a general nonlinear wave equation, the solution may not exist in a full double null neighborhood of the initial bifurcate null hypersurface as in \cref{local-existence-ss}. Indeed, the null condition means the transport equations for transversal derivatives in \cref{prop:determining-everything} are linear and hence do not blow up in finite time.
\end{rk}

\subsection{Sphere data and cone data} 

\subsubsection{Sphere data}

In order to define a notion of characteristic gluing later, we introduce a notion of \emph{sphere data} inspired by \cite{ACR1, ACR2}. Given a $C^k$ solution of the EMCSF system in spherical symmetry, for every $(u_0,v_0)\in \mathcal Q$ one can extract a list of numbers corresponding to $r(u_0,v_0)$, $\Omega^2(u_0,v_0)$, $\phi(u_0,v_0)$, $Q(u_0,v_0)$, $\partial_ur(u_0,v_0)$ etc. Our definition of sphere data formalizes this (long) list of numbers and incorporates the constraints  \eqref{eq:Maxwell-u}--\eqref{eq:CSF-Raychaudhuri-v}, so we may refer to the data induced by a $C^k$ solution on a sphere without reference to an actual solution of the equations themselves.

\begin{defn}\label{sphere-data}
Let $k\ge 1$. A \emph{sphere data set} with \emph{regularity index $k$} for the Einstein--Maxwell-charged scalar field in the EM gauge \eqref{EM-gauge} is the following list of numbers\footnote{One should think that formally $r(u_0,v_0)=\varrho$, $\phi(u_0,v_0)=\varphi$, $\Omega^2(u_0,v_0)=\omega$, $\partial_v^ir(u_0,v_0)=\varrho_v^i$, etc.}:
\begin{enumerate}
    \item $\varrho>0$, $\varrho_u^{1},\dotsc, \varrho_u^{k+1},\varrho_v^{1},\dotsc,\varrho_v^{k+1}\in \Bbb R$
    
    \item $\omega>0$, $\omega_u^{1},\dotsc,\omega_u^{k},\omega_v^{1},\dotsc,\omega_v^{k}\in\Bbb R$
    
    \item $\varphi,\varphi_u^{1},\dotsc,\varphi_u^{k},\varphi_v^{1},\dotsc,\varphi_v^{k}\in \Bbb C$
    
    \item $q, q_u^1,\dotsc, q_u^k,q_v^1,\dotsc,q_v^k\in \Bbb R$
    
    \item $a,a_u^1,\dotsc,a_u^k,a_v^1,\dotsc, a_v^k, a_v^{k+1}\in\Bbb R$
\end{enumerate}
subject to the following conditions:
\begin{enumerate}[(i)]
\item $\varrho_u^{i+2}$ can be expressed as a rational function of $\varrho_u^{j+1}$, $\omega_u^{j+1}$, $\varphi_u^{j+1}$, and $a_u^j$ for $0\le j\le i$ by formally differentiating \eqref{eq:CSF-Raychaudhuri-u},
\item $\varrho_v^{i+2}$ can be expressed as a rational function of $\varrho_v^{j+1}$, $\omega_v^{j+1}$, and $\varphi_v^{j+1}$ for $0\le j\le i$ by formally differentiating \eqref{eq:CSF-Raychaudhuri-v},
\item $q_u^{i+1}$ can be expressed as a polynomial of $\varrho_u^j$, $\varphi_u^j$, and $a_u^j$ for $0\le j\le i$ by formally differentiating \eqref{eq:Maxwell-u},
\item $q_v^{i+1}$ can be expressed as a polynomial of $\varrho_u^j$, and $\varphi_u^j$ for $0\le j\le i$ by formally differentiating \eqref{eq:Maxwell-v}, and 
\item $a_v^{i+1}$ can be expressed as a rational function of $\varrho_v^j$, $\omega_v^j$, and $q_v^j$ for $0\le j\le i$ by formally differentiating \eqref{eq:Au-transport},
\end{enumerate}
where we have adopted the convention that $\varrho_u^0=\varrho$, etc. We denote by $\mathcal D_k$   the set of such sphere data sets with regularity index $k$. 
\end{defn}

Gauge freedom is a very important aspect of the study of the EMCSF system. Our next definition records the gauge freedom present in sphere data. We need to consider both double null gauge transformations 
\begin{equation}
    u=f(U),\quad v=g(V),\nonumber
\end{equation}
where $f$ and $g$ are increasing functions on $\Bbb R$ and EM gauge transformations \eqref{eq:EM-gauge-trafo}
\begin{equation}
    \phi\mapsto e^{-i\mathfrak e \chi}\phi,\quad A\mapsto A+d\chi,\nonumber
\end{equation}
where $\chi$ is a function of $u$ alone, i.e.~$\partial_v\chi=0$, in order to satisfy \eqref{EM-gauge}. 

\begin{defn}
We define the \emph{full gauge group} of the Einstein--Maxwell-charged scalar field system in spherically symmetric double null gauge with the EM gauge condition \eqref{EM-gauge} as
\begin{equation}
    \mathcal G\doteq  \{(f,g):f,g\in \Diff_+(\Bbb R), f(0)=g(0)=0\}\times C^\infty(\Bbb R), \nonumber
\end{equation}
with the group multiplication given by\footnote{One can view this as a left semidirect product.}
\begin{equation}
    ((f_2,g_2),\chi_2)\cdot ((f_1,g_1),\chi_1)= ((f_2\circ f_1,g_2\circ g_1),\chi_2\circ f_1^{-1}+\chi_1).\nonumber
\end{equation}
The gauge group defines an action on sphere data as follows. Given sphere data $D\in \mathcal D_k$, assign functions $r(u,v)$, $\Omega^2(u,v)$, $\phi(u,v)$, $Q(u,v)$, and $A_u(u,v)$ whose jets agree with the sphere data $D$. For $\tau=((f,g),\chi)\in\mathcal G$, let
\begin{align}
   \label{gauge-transformation-1} \tilde r(u,v)&=r(f(u),g(v))\\
   \label{gauge-transformation-2}  \tilde\Omega^2(u,v)&= f'(u)g'(v)\Omega^2(f(u),g(v))\\
   \label{gauge-transformation-3}  \tilde \phi(u,v)&= e^{-i\mathfrak e \chi(f(u))} \phi(f(u),g(v)) \\
  \label{gauge-transformation-4}   \tilde Q(u,v)&= Q(f(u),g(v))\\
   \label{gauge-transformation-5}  \tilde A_u(u,v) & = f'(u)A_u(f(u))+f'(u)\chi'(f(u)).
\end{align}
The components of $\tau D$ are then defined by formally differentiating equations \eqref{gauge-transformation-1}--\eqref{gauge-transformation-5} and evaluating at $u=v=0$. For example, $\tau(\varrho)=\varrho$, $\tau(\varrho_v^1)= g'(0)\varrho_v^1$, and $\tau(\varphi_u^1)=(1-i\mathfrak e \chi'(0))e^{-i\mathfrak e \chi(0)}\varphi$.
\end{defn}

If one is given a bifurcate characteristic initial data set $(r,\Omega^2,\phi,Q,A_u)$, the lapse $\Omega^2$ can be set to unity on $C\cup\underline C$ by reparametrizing $u$ and $v$. In the sphere data setting, we have an analogous notion: 

\begin{defn}
A sphere data set $D\in \mathcal D_k$ is said to be \emph{lapse normalized} if $\omega=1$ and $\omega_u^i=\omega_v^i=0$ for $1\le i\le k$. Every sphere data set is gauge equivalent to a lapse normalized sphere data set. 
\end{defn}

\subsubsection{Cone data and seed data}\label{seed-data}

In the previous subsection, we saw how a $C^k$ solution $(r,\Omega^2,\phi,Q,A_u)$ on $\mathcal Q$ gives rise to a continuous map $\mathcal Q\to \mathcal D_k$. For the purpose of characteristic gluing, it is convenient to consider one-parameter families of sphere data which are to be thought of as being induced by constant $u$ cones in $\mathcal Q$. 

More precisely, if we consider a null cone $C\subset\mathcal Q$, parametrized by $v\in [v_1,v_2]$, then a solution of the EMCSF system induces a continuous map $D:[v_1,v_2]\to \mathcal D_k$ by sending each $v$ to its associated sphere data $D(v)$. In fact, this map can be produced by knowing only $D(v_1)$ and the values of $(r,\Omega^2,\phi,Q,A_u)$ on $C$. Arguing as in  \cref{prop:determining-everything} with $D(v_1)$ taking the role of the bifurcation sphere gives:

\begin{prop}\label{prop:null-data-existence}
Let $k\in \Bbb N$, $v_1<v_2\in\Bbb R$, $r, A_u\in C^{k+1}([v_1,v_2])$, and  $\Omega^2,\phi,Q\in C^{k}([v_1,v_2])$ which satisfy the constraints \eqref{eq:Maxwell-v}, \eqref{eq:Au-transport}, and \eqref{eq:CSF-Raychaudhuri-v} on $[v_1,v_2]$. Let $D_1\in \mathcal D_k$ such that all $v$-components of $D_1$ agree with the corresponding $v$-derivatives of $(r,\Omega^2,\phi,Q,A_u)$ at $v_1$. Then there exists a unique continuous function $D:[v_1,v_2]\to \mathcal D_k$ such that $D(v_1)=D_1$ and upon identification of the formal symbols $\varrho(D(v))$, $\varrho_u^1(D(v))$, etc., with the dynamical variables $(r,\Omega^2,\phi,Q,A_u)$ and their $u$- and $v$-derivatives, satisfies the EMCSF system and agrees with $(r,\Omega^2,\phi,Q,A_u)$ in the $v$-components for every $v\in [v_1,v_2]$.
\end{prop}

\begin{defn}
Let $k\in\Bbb N$ and $v_1<v_2\in\Bbb R$. A \emph{$C^k$ cone data set} for the Einstein--Maxwell-charged scalar field in spherical symmetry is a continuous function $D:[v_1,v_2]\to \mathcal D_k$ satisfying the conclusion of \cref{prop:null-data-existence}, i.e., formally satisfying the EMCSF system.  
\end{defn}

We now discuss a procedure for generating solutions of the ``tangential'' constraint equations, \eqref{eq:Maxwell-v}, \eqref{eq:Au-transport}, and \eqref{eq:CSF-Raychaudhuri-v}, which were required to be satisfied in the previous proposition. 

\begin{prop}[Seed data]\label{prop:seed-data}
Let $k\in \Bbb N$, $v_1<v_2\in\Bbb R$, and $D_1\in \mathcal D_k$ be lapse normalized. For any $\phi \in C^k([v_1,v_2])$ such that $\partial_v^i\phi(v_1)=\varphi_v^i(D_1)$ for $0\le i\le k$, there exist unique functions $r,A_u\in C^{k+1}([v_1,v_2])$ and $Q\in C^k([v_1,v_2])$ such that $(r,\Omega^2,\phi,Q,A_u)$ satisfies the hypotheses of \cref{prop:null-data-existence} with $\Omega^2(v)=1$ for every $v\in [v_1,v_2]$.
\end{prop}

\begin{proof}
When $\Omega^2\equiv 1$, Raychaudhuri's equation \eqref{eq:CSF-Raychaudhuri-v} reduces to 
\[\partial_v^2 r = - r|\partial_v\phi|^2,\]
which is a second order ODE for $r(v)$. Setting $r(v_1)=\varrho(D_1)$ and $\partial_v r(v_1)=\varrho_v^1(D_1)$, we obtain a unique solution $r\in C^{k+1}([v_1,v_2])$. The charge is obtained by integrating Maxwell's equation \eqref{eq:Maxwell-v}:
\[Q(v)= q(D_1)+\int_0^v \mathfrak e r^2(v') \Im(\phi(v')\overline{\partial_v \phi(v')})\,dv'.\]
Finally, the gauge potential is obtained by integrating \eqref{eq:Au-transport}:
\[A_u(v)= a(D_1)-\int_0^v \frac{Q(v')}{2r^2(v')}\,dv'.\]
The $v$-derivatives of $(r,\Omega^2,\phi,Q,A_u)$ agree with the $v$-components of $D_1$ by virtue of the definitions. 
\end{proof}

\section{The main gluing theorems}

In this section we give precise statements of our main theorems. In order to do this, we carefully define the notion of \emph{characteristic gluing}.

\subsection{Characteristic gluing in spherical symmetry}
\label{subsec:defn-sphere-data}

\begin{defn}[Characteristic gluing]
\label{defn:char-gluing}
Let $k\in\Bbb N$. Let $D_1,D_2\in\mathcal D_k$ be sphere data sets. We say that $D_1$ can be \emph{characteristically glued to $D_2$ to order $k$} in the Einstein--Maxwell-charged scalar field system in spherical symmetry if there exist $v_1<v_2$ and a $C^k$ cone data set $D:[v_1,v_2]\to \mathcal D_k$ such that $D(v_1)$ is gauge equivalent to $D_1$ and $D(v_2)$ is gauge equivalent to $D_2$. 
\end{defn}

\begin{rk}\label{gauge-remark}
It is clear that if $D_1$ and $D_2$ can be characteristically glued and $\tau_1,\tau_2\in \mathcal G$, then $\tau_1D_1$ and $\tau_2D_2$ can be characteristically glued. 
\end{rk}

\begin{rk}\label{rk:gluing-along-ingoing-data}
\cref{defn:char-gluing} on characteristic gluing along an outgoing cone has a natural analog defining characteristic gluing along an ingoing cone by parametrizing the cone data with $u$ and letting $v$ denote the transverse null coordinate, but keeping the definition of sphere data unchanged.
\end{rk}

By \cref{prop:seed-data}, characteristic gluing is equivalent to choosing an appropriate seed $\phi$ in the following sense. By applying a gauge transformation to $D_1$, we may assume it to be lapse normalized. Then cone data sets with $\Omega^2\equiv 1$ agreeing with $D_1$ at $v_1$ are parametrized precisely by functions $\phi\in C^k([v_1,v_2];\Bbb C)$ with the correct $v$-jet at $v_1$. Therefore, characteristic gluing reduces to finding $\phi$ so that the final data set $D(v_2)$ produced by \cref{prop:null-data-existence} is gauge equivalent to $D_2$. 

\subsection{Spacetime gluing from characteristic gluing}

If the two sphere data sets in \cref{defn:char-gluing} come from spheres in two spherically symmetric EMCSF spacetimes, we can use local well posedness for the EMCSF characteristic initial value problem, \cref{local-existence-ss}, to glue parts of the spacetimes themselves. This principle underlies all of our constructions in \cref{sec:construction-spacetime-cauchy-data}.

\begin{figure}[ht]
\centering{
\def\svgwidth{25pc}
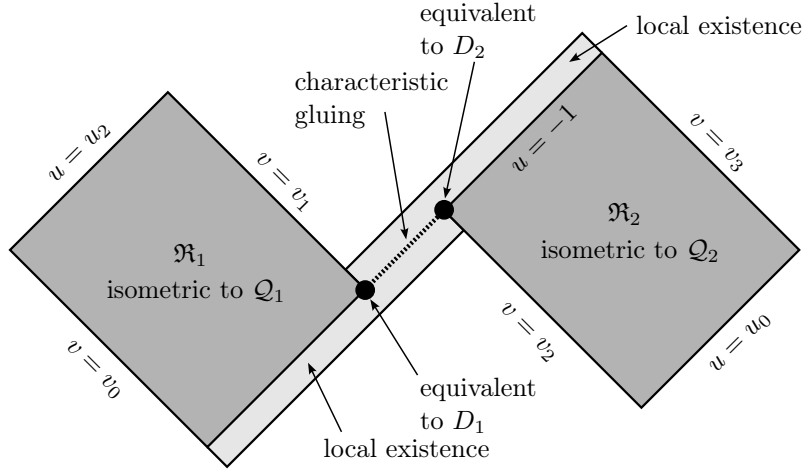}
\caption{Spacetime gluing obtained from characteristic gluing. The two spacetimes (dark gray) are glued along the cone $u=-1$. Note that the dark gray regions are causally disconnected except for the cone $u=-1$. Such a spacetime exists if and only if $D_1$ and $D_2$ can be characteristically glued.}
\label{fig:spacetime-gluing}
\end{figure}
\begin{prop}\label{prop:spacetime-gluing}
Let $(\mathcal Q_1,r_1,\Omega^2_1,\phi_1,Q_1, A_{u\,1})$ and $(\mathcal Q_2,r_2,\Omega^2_2, \phi_2,Q_2,A_{u\,2})$ be two $C^k$ solutions of the EMCSF system in spherical symmetry, where each $\mathcal Q_i$ is a double null rectangle, i.e., 
\begin{align*}
\mathcal Q_1&= [u_{0,1},u_{1,1}]\times [v_{0,1},v_{1,1}]\\
\mathcal Q_2&= [u_{0,2},u_{1,2}]\times [v_{0,2},v_{1,2}].
\end{align*}
Let $D_1$ be the sphere data induced by the first solution on $(u_{0,1}, v_{1,1})$ and $D_2$ be the sphere data induced by the second solution on $(u_{1,2}, v_{0,2})$. If $D_1$ can be characteristically glued to $D_2$ to order $k$, then there exists a spherically symmetric $C^k$ solution $(\mathcal Q, r,\Omega^2,\phi,Q,A_u)$ of the EMCSF system with the following property: There exists a global double null gauge $(u,v)$ on $\mathcal Q$ containing double null rectangles 
\begin{align*}
\mathfrak R_1&=[-1,u_2]\times [v_0,v_1],\\
\mathfrak R_2&=[u_0,-1]\times [v_2,v_3],
\end{align*}
such that the restricted solutions $(\mathfrak R_i,r,\Omega^2,\phi,Q,A_u)$
are isometric to the solutions $(\mathcal Q_i,r_i,\Omega^2_i,\phi_i,Q_i,A_{u\, i})$ for $i=1,2$, the sphere data induced on $(-1,v_1)$ is equal to $D_1$ to $k$-th order, and the sphere data induced on $(-1,v_2)$ is gauge equivalent to $D_2$ to $k$-th order.
\end{prop}

\begin{proof}
In this proof, we will refer to spherically symmetric solutions of the EMCSF system by their domains alone. 

By \cref{defn:char-gluing}, since $D_1$ and $D_2$ can be characteristically glued, we obtain $v_1<v_2$, functions $r,\Omega^2,\phi,Q,$ and $A_u$ on $[v_1,v_2]$, and a gauge transformation $\tau\in \mathcal D_k$ which acts on $D_2$. We now build the spacetime out of two pieces which will then be pasted along $u=-1$ and match to order $C^k$. See \cref{fig:spacetime-gluing-2}.

First, we prepare the given spacetimes. We relabel the double null gauge on $\mathcal Q_1$ by changing the southeast edge to be $u=-1$ and the northeast edge to be $v=v_1$. This also determines $u_2$ and $v_0$ and we apply no further gauge transformation to $\mathcal Q_1$. We denote this region by $\mathfrak R_1$. 

Next, the gauge transformation $\tau$ is extended and applied to $\mathcal Q_2$. We relabel the double null gauge to have $u=-1$ on the northwest edge and $v=v_2$ on the southwest edge. We denote this region by $\mathfrak R_2$.

We now construct the left half of \cref{fig:spacetime-gluing-2} as follows. Extend the cone $u=-1$ in $\mathfrak R_1$ until $v=v_3$, and extend the functions $(r,\Omega^2,\phi,Q,A_u)$ on $u=-1$ by taking them from the definition of characteristic gluing for $v\in [v_1,v_2]$, and then from the induced data on $u=-1$ in $\mathfrak R_2$ for $v\in [v_2,v_3]$. We now appeal to local existence, \cref{local-existence-ss}, the EMCSF system in spherical symmetry to construct the solution in a thin slab $\mathcal S_1$ to the future of \[(\{u=-1\}\times [v_1,v_3])\cup([-1,u_2]\times \{v=v_1\}).\]
This completes the construction of $\mathfrak R_1\cup \mathcal S_1$. 

The region $\mathfrak R_2\cup \mathcal S_2$ is constructed similarly, with the cone $u=-1$ now being extended backwards, first using the characteristic gluing data and then using the tangential data induced by $\mathfrak R_1$ on $u=-1$. Again, \cref{local-existence-ss} is used to construct the thin strip $\mathcal S_2$.

Finally, the spacetime is constructed by taking $\mathcal Q\doteq (\mathfrak R_1\cup \mathcal S_1)\cup (\mathfrak R_2\cup \mathcal S_1)$ and pasting $r,\Omega^2,\phi,Q,$ and $A_u$. From the construction, it is clear that the dynamical variables, together with all $v$-derivatives consistent with $C^k$ regularity are continuous on $\mathcal Q$. To show that all $u$-derivatives are continuous across $u=-1$, we observe that all transverse quantities are initialized consistently to $k$-th order at $(-1,v_1)$ and that the tangential data agrees by construction. Now \cref{prop:determining-everything} implies that the transverse derivatives through order $k$ are equal on $u=-1$ in both $\mathfrak R_1\cup \mathcal S_1$ and $\mathfrak R_2\cup \mathcal S_2$. This completes the proof. 
\end{proof}

\begin{figure}[ht]
\centering{
\def\svgwidth{30pc}
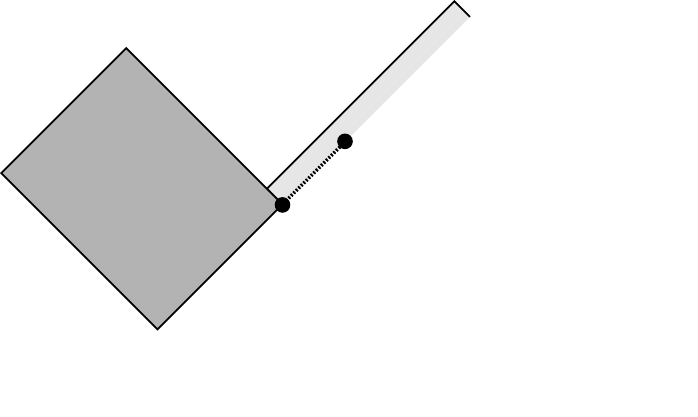}
\caption{Proof of \cref{prop:spacetime-gluing}.}
\label{fig:spacetime-gluing-2}
\end{figure}

\begin{rk}
If the characteristic gluing hypothesis is $C^k$ but no better and the original solutions $\mathcal Q_1$ and $\mathcal Q_2$ are more regular than $C^k$, then one expects $(k+1)$-th derivatives of dynamical quantities to jump across any of the null hypersurfaces bordering the light gray regions in \cref{fig:spacetime-gluing}. 
\end{rk}

\subsection{Sphere data in Minkowski, Schwarzschild, and Reissner--Nordstr\"om}
\label{subsec:sphere-data-mink-schwarz}
Before stating our main gluing results, we need to precisely define the terms \emph{Minkowski sphere}, \emph{Schwarzschild event horizon sphere}, and \emph{Reissner--Nordstr\"om event horizon sphere}. 

\begin{defn}[Minkowski sphere data]
Let $k\in \Bbb N$ and $R>0$. The unique lapse normalized sphere data set satisfying 
\begin{itemize}
    \item $\varrho =R$,
    \item $\varrho_u^1=-\tfrac 12$, 
    \item $\varrho_v^1=\frac 12$, and
    \item all other components zero,
\end{itemize}
is called the \emph{Minkowski sphere data of radius $R$} and is denoted by $D_{R,k}^\mathrm{M}$.
\end{defn}

\begin{defn}[Schwarzschild sphere data]
Let $k\in\Bbb N$, $R>0$, and $0\le 2M\le R$. The unique lapse normalized sphere data set satisfying 
\begin{itemize}
    \item $\varrho = R$,
    \item $\varrho_u^1=-\tfrac 12$
    \item $\varrho_v^1=\tfrac 12 (1-2M/R)$, and
    \item all other components zero,
\end{itemize}
is called the \emph{Schwarzschild sphere data of mass $M$ and radius R} and is denoted by $D^{\mathrm{S}}_{M,R,k}$. Note that $D^\mathrm{S}_{0,R,k}=D_{R,k}^\mathrm{M}$.
\end{defn}

\begin{defn}[Reissner--Nordstr\"om horizon sphere data]
\label{defn:RN-eh-data}
Let $k\in \Bbb N$, $M>0$, and $0\le |e|\le M$. The unique lapse normalized sphere data set satisfying 
\begin{itemize}
    \item $\varrho =r_+\doteq M+\sqrt{M^2-e^2}$, 
    \item $\varrho_u^1=-\tfrac 12$, 
    \item $\varrho_v^1=0$, 
    \item $q=e$, and
    \item all other components zero, 
\end{itemize}
is called the \emph{Reissner--Nordstr\"om horizon sphere data with parameters $M$ and $e$} and is denoted by $D^{\mathrm{RN}\mathcal H}_{M,e,k}$. Note that $D^{\mathrm{RN}\mathcal H}_{M,0,k}=D^\mathrm{S}_{M,2M,k}$.
\end{defn}

We will also define sphere data for general  Reissner--Nordström spheres. To do so, we extend the Hawking mass \eqref{eq:Hawking-mass} to a  function on sphere data sets $D\in \mathcal D_k$ by setting 
\begin{equation*}
    m(D) \doteq \frac{\varrho}{2}\left(1+\frac{4\varrho_u^1\varrho_v^1}{\omega}\right).
\end{equation*}
We also define the \emph{modified Hawking mass} of a spherically symmetric spacetime with charge by 
\begin{equation*}
    \varpi \doteq m +\frac{Q^2}{2r}
\end{equation*}
and extend it to sphere data sets by 
\begin{equation*}
    \varpi(D)\doteq  m(D)+\frac{q^2}{2\varrho}.
\end{equation*}
In a Reissner--Nordstr\"om spacetime of mass $M$ and charge $e$, any sphere data set $D$ associated to a symmetry sphere has $\varpi(D)=M$. Note that given $\varrho>0,\varrho_u^1<0,\omega$, and $q$, $\varrho_v^1$ is determined uniquely by $\varpi(D)$. 

Recall that the horizons of Reisser--Nordstr\"om with parameters $|e|\le M$ are located at 
\begin{equation*}
r_\pm = M\pm\sqrt{M^2-e^2}.
\end{equation*}

\begin{defn}[Reissner--Nordström sphere data]\label{defn:non-horizon-data} Let $k\in\Bbb N$,  $e\in\Bbb R$, and $R>0$ satisfy $M > e^2/(2R)$.  A lapse normalized sphere data set satisfying  
    \begin{itemize}
        \item $\varrho = R$,
        \item $q=e$,
        \item $\varpi = M$,
        \item $ |\varrho_v^1|=\tfrac 12$, or $ |\varrho_u^1|=\tfrac 12$, or $ \varrho_v^1=\varrho_u^1 = 0 $,
        \item all other components zero 
    \end{itemize}
    is called a \emph{Reissner--Nordstr\"om sphere data set of modified Hawking mass $M$, charge $e$, and radius $R$} and is denoted by $D^{\mathrm{RN}}_{M,e,R,k}$.
\end{defn}

\begin{rk}
    A Reissner--Nordstr\"om sphere data set of modified Hawking mass $M$, charge $e$, and radius $R$, $D^{\mathrm{RN}}_{M,e,R,k}$,  gives rise to unique sphere data  if   either $\varrho_v^1=\varrho_u^1=0$, or one additionally  specifies  $\operatorname{sgn}(\varrho_v^1)\in \{+,-\} $ or  $\operatorname{sgn}(\varrho_u^1)\in \{+,-\}$. 
\end{rk}

\subsection{Main gluing theorems}
\label{sec:precise-version-main-theorems}
With the previous definitions of \cref{subsec:defn-sphere-data} and \cref{subsec:sphere-data-mink-schwarz} at hand, we are now in a position to state our main gluing results. 

Our first gluing theorem concerns gluing a sphere in Minkowski space to a Schwarzschild event horizon with a \emph{real} scalar field. When the scalar field $\phi$ in the EMCSF system is real-valued, Maxwell's equation decouples from the rest of the system and the charge $Q$ is constant throughout the spacetime. Since $Q$ must vanish on any sphere in Minkowski space, it vanishes everywhere and the EMCSF system reduces to the Einstein-scalar field system.

\setcounter{thm}{1}
\begin{subtheorem}{thm}
\begin{thm}\label{main-thm-Schw}
For any $k\in \Bbb N$ and $0<R_i<2M_f$, the Minkowski sphere of radius $R_i$, $D_{R_i,k}^\mathrm{M}$, can be characteristically glued to the Schwarzschild event horizon sphere with mass $M_f$, $D_{M_f,k}^\mathrm{S}$, to order $C^k$ within the Einstein-scalar field model in spherical symmetry. 
\end{thm}
The proof of \cref{main-thm-Schw} is given in \cref{subsec:Schwarzschild-gluing-proof}. We have separated out Minkowski to Schwarzschild gluing as a special case because it is simpler and highlights our topological argument. We will actually use this special case as the first step to produce our counterexample to the third law in \cref{proof-trapped-surfaces-ERN}. 

Our second gluing theorem concerns gluing a sphere in the domain of outer communication of a Schwarzschild spacetime to a Reissner--Nordstr\"om event horizon with specified mass and charge to mass ratio.

\begin{thm}\label{main-thm-RN}
For any $k\in \Bbb N$, $\mathfrak q\in [-1,1]$, and $\mathfrak e\in \Bbb R\setminus\{0\}$, there exists a number $M_0(k, \mathfrak q,\mathfrak e)\ge 0$ such that if $M_f> M_0$, $0\le M_i\le \tfrac 18 M_f$, and $2M_i<R_i\le \tfrac 12 M_f$, then the Schwarzschild sphere of mass $M_i$ and radius $R_i$, $D^\mathrm{S}_{M_i,R_i,k}$, can be characteristically glued to the Reissner--Nordstr\"om event horizon with mass $M_f$ and charge to mass ratio $\mathfrak q$, $D_{M_f,\mathfrak qM_f,k}^{\mathrm{RN}\mathcal H}$, to order $C^k$ within the Einstein--Maxwell-charged scalar field model with coupling constant $\mathfrak e$. The associated characteristic data can be chosen to have no spherically symmetric  antitrapped surfaces, i.e.~$\partial_u r<0$ everywhere. 
\end{thm}
\end{subtheorem}
The proof of \cref{main-thm-RN} is given in \cref{sec:gluing-rn}.

\begin{rk}\label{ESF-extra-monotonicity}
The data constructed in the proof of \cref{main-thm-Schw} will automatically not contain spherically symmetric antitrapped surfaces because of a special monotonicity property in the absence of charge. Namely, 
\begin{equation}
    \label{eq:monotonicity-partialur}
 \partial_v(r\partial_u r)=-\frac{\Omega^2}{4},
\end{equation}
so $r\partial_ur$ is decreasing. In particular, since $r\partial_u r$ is negative in Minkowski space, the sign will propagate in view of \eqref{eq:monotonicity-partialur} for the Einstein-scalar field model.
\end{rk}

Our next gluing theorem supersedes \cref{main-thm-Schw} and \cref{main-thm-RN} by relaxing the requirement that the final sphere lie on the event horizon. The proof is slightly more involved than \cref{main-thm-RN} but has the same basic structure and is given in \cref{sec:non-horizon-gluing} below.

\setcounter{thm}{1}
\begin{subtheorem}{thm}
\setcounter{thm}{2}
\begin{thm}
\label{thm:interior-gluing}
    For any $k\in\Bbb N,\mathfrak q\in \Bbb R$,  $\mathfrak e\in\Bbb R\setminus\{0\}$ and $\mathfrak r >0$, there exists a number $M_0(k,\mathfrak q,\mathfrak e, \mathfrak r)>0$ such that if $M_f>M_0$ and
    \begin{equation}
        \label{eq:condition-on-Rf}
        R_f \geq\frac{M_f}{2} (1 + \mathfrak r)\mathfrak q^2,
    \end{equation}
     then there exists $R_i\in(0,R_f)$ such that the Minkowski sphere of radius $R_i$, $D^\mathrm{M}_{R_i,k}$, can be characteristically glued to the Reissner--Nordstr\"om sphere with modified Hawking mass $M_f$, charge $\mathfrak qM_f$,  and radius $R_f$, $D^\mathrm{RN}_{M_f,\mathfrak qM_f,R_f,k}$ with $\varrho_u^1 <0$, to order $C^k$ within the Einstein--Maxwell-charged scalar field system with coupling constant $\mathfrak e$. The associated characteristic data can be chosen to have no spherically symmetric antitrapped surfaces, i.e., $\partial_u r<0$ everywhere. 
\end{thm}
 
 \end{subtheorem}

\begin{rk}
 Reissner--Nordström spheres with modified Hawking mass $M$, charge $\mathfrak q M$ and radius $R\leq  \frac{M}{2} \mathfrak q^2$ have non-positive Hawking mass, $m \leq 0$. In this sense, the assumption $\mathfrak r>0$ in \cref{thm:interior-gluing} is necessary. Indeed, one immediately sees that \eqref{eq:condition-on-Rf} implies
 \begin{equation*}m \geq \frac{\mathfrak r
}{1+\mathfrak r} M_f,\end{equation*} 
so that $\mathfrak r>0$ ensures $m>0$.
\end{rk}

\begin{rk}
  \cref{thm:interior-gluing} also allows for gluing of Minkowski space to   Reissner--Nordström Cauchy horizons located at $r=r_-$. This is achieved by setting $\mathfrak r = \mathfrak q^2/4$ in   \cref{thm:interior-gluing}, see already the proof of \cref{cor:CH-closes-off}.
\end{rk}
 
While all the above theorems are stated as gluing results along outgoing cones, by mapping  $u\mapsto -v $ and $v\mapsto -u$, they also hold true for gluing along ingoing cones, recall \cref{rk:gluing-along-ingoing-data}. In particular, restating \cref{thm:interior-gluing} for gluing along ingoing cones gives
\setcounter{thm}{3}
\begin{manual}{2C$^\prime$}
\label{thm:backwards-gluing}
    For any $k\in\Bbb N,\mathfrak q\in \Bbb R$,  $\mathfrak e\in\Bbb R\setminus\{0\}$ and $\mathfrak r >0$, there exists a number $M_0(k,\mathfrak q,\mathfrak e, \mathfrak r)>0$ such that if $M_f>M_0$ and
    \begin{equation}
        \label{eq:condition-backwards-on-Rf}
        R_f \geq\frac{M_f}{2} (1 + \mathfrak r)\mathfrak q^2,
    \end{equation}
     then there exists $R_i\in(0,R_f)$ such that the Reissner--Nordstr\"om sphere with modified Hawking mass $M_f$, charge $\mathfrak qM_f$, and radius $R_f$, $ D^\mathrm{RN}_{M_f,\mathfrak qM_f,R_f,k}$ with $\varrho_v^1 >0$, can be characteristically glued along an ingoing cone to the Minkowski sphere of radius $R_i$, $D^\mathrm{M}_{R_i,k}$, to order $C^k$ within the Einstein--Maxwell-charged scalar field system with coupling constant $\mathfrak e$. The associated characteristic data can be chosen to have no spherically symmetric trapped surfaces, i.e., $\partial_v r>0$ everywhere. 
\end{manual}

\section{Proofs of the main gluing theorems}

\label{sec:gluing-section}

We begin with two lemmas which identify the orbits of Schwarzschild and Reissner--Nordstr\"om sphere data under the action of the full gauge group. This essentially amounts to a version of Birkhoff's theorem for sphere data.

\begin{lem}[Schwarzschild exterior sphere identification]\label{lem:gauge-Mink}
If $D\in \mathcal D_k$ satisfies 
\begin{itemize}
    \item $\varrho =R>0$,
    \item $\varrho_u^1<0$,
    \item $\varrho_v^1>0$,
    \item $\tfrac 12\varrho(1+4 \varrho_u^1\varrho_v^1)=M,$
    \item $q=0$, and 
    \item $\varphi_u^i=\varphi^i_v=0$ for $0\le i\le k$,
\end{itemize}
then $R>2M$ and $D$ is equivalent to $D^\mathrm{S}_{M,R,k}$ up to a gauge transformation. 
\end{lem}
\begin{proof}
First, we observe that by the relations obtained from Maxwell's equations, $q_u^i=q_v^i=0$ for $1\le i\le k$. Since $\varphi_u^i=\varphi^i_v=0$, we can perform an EM gauge transformation to make $a_u^i=0$ for $0\le i\le k$. Also, $a_v^i=0$ for $1\le i\le k$ from $F=d(A_u\,du)$. Next, we can normalize the lapse. Finally, $R>2M$ follows from the definitions and $\varrho_u^1\varrho_v^1<0$.
\end{proof}

\begin{lem}[Reissner--Nordstr\"om horizon sphere identification] \label{lem:gauge-Schw}
If $D\in \mathcal D_k$ satisfies
\begin{itemize}
\item $\varrho = (1+\sqrt{1-\mathfrak q^2})M$ for $\mathfrak q\in [-1,1]$ and $M>0$,  
\item $\varrho_u^1<0$,
\item $\varrho_v^1 = 0$,
\item $q=\mathfrak qM$, and 
\item $\varphi_u^i=\varphi^i_v=0$ for $0\le i\le k$,
\end{itemize}
then $D$ is equivalent to $D^\mathrm{RN}_{M,\mathfrak qM,k}$ up to a gauge transformation. 
\end{lem}
\begin{proof}
As before, the charge vanishes to all orders and we normalize the gauge potential and lapse. We then use the additional double null gauge freedom $u\mapsto \lambda u$, $v\mapsto \lambda^{-1}v$ to make $\varrho_u^1=-\tfrac 12$.
\end{proof}

\begin{rk}
Without the condition $\varrho_u^1<0$ in the previous lemma, the sphere data in the extremal case could also arise from the \emph{Bertotti--Robinson} universe. 
\end{rk}

With these lemmas and \cref{gauge-remark} in mind, we follow the strategy discussed in \cref{subsec:defn-sphere-data}. We fix the interval $[0,1]$, set $\Omega^2\equiv 1$, and solve Raychaudhuri's equation, Maxwell's equation, and the transport equation for transverse derivatives of $\phi$ with appropriate initial and final values. We do \textbf{not} have to track transverse derivatives of $\partial_u r$, $\Omega^2$, $Q$, or $A_u$, because these will be ``gauged away'' at the end of the proof. 

\subsection{Proof of \texorpdfstring{\cref{main-thm-Schw}}{Theorem 2A}} 
\label{subsec:Schwarzschild-gluing-proof}

In this subsection we prove \cref{main-thm-Schw}. We first note that if the scalar field is chosen to be real-valued, the Einstein--Maxwell-charged scalar field system collapses to the Einstein-scalar field system. If the initial data has no charge ($Q(0)=0$), then this is equivalent to setting $\mathfrak e=0$ and $A_u$ and all its derivatives to be identically zero. 

We will first set up our scalar field ansatz as a collection of pulses. To do so, let  
\[0=v_0<v_1<\cdots< v_k<v_{k+1}=1\]
be an arbitrary partition of $[0,1]$. For each $1\le j\le k+1$, fix a nontrivial bump function \[\chi_j\in C_c^\infty((v_{j-1},v_j);\Bbb R).\] In the rest of this section, the functions $\chi_1,\dotsc,\chi_{k+1}$ are fixed and our constructions depend on these choices. 

Let $\alpha=(\alpha_1,\dotsc,\alpha_{k+1})\in\Bbb R^{k+1}$ and set 
\begin{equation}
    \phi_\alpha(v)\doteq \phi(v;\alpha)\doteq \sum_{1\le j\le k+1}\alpha_j\chi_j(v).\label{def-phi}
\end{equation}
We set $\Omega^2(v;\alpha)\equiv 1$ along $[0,1]$ and define $r(v;\alpha)$ as the unique solution of Raychauduri's equation \eqref{eq:CSF-Raychaudhuri-v} with this scalar field ansatz,
\begin{equation}
    \partial_v^2 r(v;\alpha) = - r(v;\alpha)(\partial_v\phi_\alpha(v))^2\label{eqn:uncharged-Ray-proof},
\end{equation}
 with prescribed ``final values''
\begin{align*}
    r(1;\alpha)&=2M_f\\
    \partial_v r(1;\alpha)&=0. 
\end{align*}

Let $0<\ve<2M_f-R_i$. By Cauchy stability and monotonicity properties of Raychaudhuri's equation \eqref{eqn:uncharged-Ray-proof}, there exists a $\delta>0$ such that for every $0<|\alpha|\le \delta$, 
\begin{align*}
  \sup_{[0,1]} |r(\cdot;\alpha)-2M_f|&\le \ve, \\
  \inf_{[0,1]} \partial_v r(\cdot;\alpha)&\ge 0,\\
   \partial_v r(0;\alpha)& > 0.
\end{align*}
The final inequality follows from the fact that $\alpha\ne 0$.

We now consider the sphere $S^k_\delta\doteq\{\alpha\in\Bbb R^{k+1}:|\alpha|=\delta\}$. For each $\alpha\in S_\delta^k$, define $D_\alpha(0)\in \mathcal D_k$ by setting
\begin{itemize}
    \item $\varrho=r(0;\alpha)>0$,
    \item $\varrho_v^1= \partial_vr(0;\alpha)>0$,
    \item $\varrho_u^1=-\tfrac 14 (\varrho_v^1)^{-1}$,
    \item $\omega=1$, and 
    \item all other components to zero.
\end{itemize}
By \cref{lem:gauge-Mink}, $D_\alpha(0)$ is equivalent to $D^\mathrm{M}_{r(0;\alpha),k}$ up to a gauge transformation. 

For each $\alpha\in S_\delta^k$, we now apply \cref{prop:null-data-existence} and \cref{prop:seed-data} to uniquely determine cone data \[D_\alpha:[0,1]\to \mathcal D_k,\]
 with initialization $D_\alpha(0)$ above and seed data $\phi_\alpha$ given by \eqref{def-phi}. By standard ODE theory, $D_\alpha(v)$ is jointly continuous in $v$ and $\alpha$. Note that $\varrho(D_\alpha(v))=r(v;\alpha)$ and $\varphi(D_\alpha(v))=\phi(v;\alpha)$ by definition. We now use the notation 
 \[\partial_u^i \phi(v;\alpha)\doteq \varphi_u^i(D_\alpha(v))\]
 for $i=1,\dotsc,k$ to denote the transverse derivatives of the scalar field obtained by \cref{prop:null-data-existence}.

By construction, the data set $D_\alpha(1)$ satisfies
\begin{itemize}
\item $\varrho =2M_f,$
\item $\varrho_u^1 <0$,
\item $\varrho_v^1=0$,
\item $\omega=1$, and
\item $\varphi_v^i=0$ for $0\le i\le k$.
\end{itemize}
The second property follows from the initialization of $\varrho_u^1$ in $D_\alpha(0)$ and the monotonicity of 
\[(r\partial_u r)(v;\alpha)\doteq \varrho (D_\alpha(v))\varrho_u^1(D_\alpha(v))\] in the Einstein-scalar field system discussed in \cref{ESF-extra-monotonicity}. 

In order to glue to Schwarzschild at $v=1$, by \cref{lem:gauge-Schw}, it suffices to find an $\alpha_*\in S_\delta^k$ for which additionally 
\[\partial_u\phi(1;\alpha_*)=\cdots=\partial_u^k\phi(1;\alpha_*)=0.\]
The following discrete symmetry of the Einstein-scalar field system plays a decisive role in finding $\alpha_*$.

A function $f(v;\alpha)$ is \emph{even in $\alpha$} if $f(v;-\alpha)=f(v;\alpha)$ and \emph{odd in $\alpha$} if $f(v;-\alpha)=-f(v;\alpha)$.

\begin{lem}\label{lem:odd}
As functions on $[0,1]\times S^k_\delta$, the metric coefficients $r(v;\alpha)$, $\Omega^2(v;\alpha)$ and all their ingoing and outgoing derivatives are even functions of $\alpha$. The scalar field $\phi(v;\alpha)$ and all its ingoing and outgoing derivatives are odd functions of $\alpha$. In particular, the map 
\begin{align}
 \label{eqn:F-defn}   F:S^k_\delta&\to \Bbb R^k\\
   \alpha &\mapsto \left(\partial_u\phi(1;\alpha),\dotsc ,\partial_u^k\phi(1;\alpha)\right)\nonumber
\end{align}
is continuous and odd. 
\end{lem}
\begin{proof}
The scalar field itself is odd by the definition \eqref{def-phi}. Since Raychaudhuri's equation 
\eqref{eqn:uncharged-Ray-proof} involves the square of $\partial_v\phi(v;\alpha)$, $r(v;\alpha)$ will be automatically even.  Next, $\partial_u r(v;\alpha)$ is found by integrating the wave equation for the radius \eqref{eq:CSF-equation-for-r}, forwards in $v$
with initial value determined by $D_\alpha(0)$. Since $\phi$ enters into this equation with an even power (namely zero), $\partial_u r(v;\alpha)$ will also be even.  The wave equation for $r\phi$ in the Einstein-scalar field model can be derived from \eqref{rphi2} and reads
\begin{equation*}
    \partial_u \partial_v(r\phi)= -\frac{\Omega^2 m}{2r^3}\,r\phi,  
\end{equation*}
and the right-hand side is odd in $\alpha$ (the Hawking mass is constructed from metric coefficients so is also even). Recall from \cref{prop:null-data-existence} that this wave equation is used to compute $\varphi_u^i(D_\alpha(v))$. By inspection $\partial_u (r\phi)$ is odd, whence $\partial_u\phi(v;\alpha)$ is also odd. The proof now follows by inductively following the procedure of \cref{prop:determining-everything}, taking note of the fact that the transport equations for ingoing derivatives of $r$ and $\Omega^2$ only involve even powers of $\phi$ and its derivatives, whereas the transport equations for ingoing derivatives of $\phi$ only involve odd powers. 

The claim about the map $F$ follows from the oddness of ingoing derivatives of $\phi$ and the continuity of all dynamical quantities in $\alpha$, per standard ODE theory. 
\end{proof}

We now complete the proof of \cref{main-thm-Schw}. By the Borsuk--Ulam theorem stated as \cref{BorsukUlam}, $F(\alpha_*)=0$ for some $\alpha_*\in S_\delta^k$, where $F$ is as in \eqref{eqn:F-defn}. By \cref{lem:gauge-Schw}, $D_{\alpha_*}(1)$ is gauge equivalent to $D^\mathrm{S}_{M_f,k}$. 

So far we have glued $D^\mathrm{M}_{r(0;\alpha),k}$ to $D^\mathrm{S}_{M_f,k}$, and since $r(0;\alpha)> R_i$, we extend the data trivially in order to glue $D^\mathrm{M}_{R_i,k}$ to $D^\mathrm{S}_{M_f,k}$, which concludes the proof of \cref{main-thm-Schw}. \qed

\subsection{Proof of \texorpdfstring{\cref{main-thm-RN}}{Theorem 2B}}
\label{sec:gluing-rn}

In this subsection we prove \cref{main-thm-RN}.
 We assume that $\mathfrak q\ne 0$, the $\mathfrak q=0$ version of this result being essentially a repeat of the arguments in the previous section combined with the new initialization of $\partial_u r(0;\alpha)$ in \eqref{ni-init} below. 
 
In this subsection we adopt the notational convention that $A\les B$ means $A\le CB$, where $C$ is a constant that depends only on $k$  and the baseline scalar field profile, but not on $\mathfrak q$, $\mathfrak e$, $M_i$, $M_f$, or $\alpha$. The notation $A\approx B$ means $A\les B$ and $B\les A$.    

Let  
\begin{equation*}0=v_0<v_1<\cdots< v_{2k}<v_{2k+1}=1\end{equation*}
be an arbitrary partition of $[0,1]$. For each $1\le j\le 2k+1$, fix a nontrivial bump function \[\chi_j\in C_c^\infty((v_{j-1},v_j);\Bbb R).\] In the rest of this section, the functions $\chi_1,\dotsc,\chi_{2k+1}$ are fixed and our constructions depend on these choices.

For $\alpha=(\alpha_1,\dotsc,\alpha_{2k+1})\in\Bbb R^{2k+1}$, set 
\begin{equation} 
    \phi_\alpha(v)\doteq\phi(v;\alpha)\doteq\sum_{1\le j\le 2k+1}\alpha_j \chi_j(v)e^{-iv}.\label{eq:defn-phi-RN}
\end{equation}

\begin{rk}\label{signs-issue}
If $\mathfrak e>0$, this choice of $\phi$ will make $Q\ge 0$, which is consistent with $\mathfrak q>0$. If $\mathfrak e>0$ and $\mathfrak q<0$, then we replace $-iv$ in the exponential with $+iv$. Similarly, the cases $\mathfrak e<0$, $\mathfrak q>0$ and $\mathfrak e<0$, $\mathfrak q<0$ can be handled. Therefore, we assume without loss of generality that $\mathfrak e>0$, $\mathfrak q>0$.
\end{rk}

For $\hat\alpha\in S^{2k}$ (the unit sphere in $\Bbb R^{2k+1}$) and $\beta\ge 0$, it is convenient to define $r(v;\beta,\hat\alpha)=r(v;\beta\hat\alpha)$, etc.
We again set $\Omega^2(v;\alpha)\equiv 1$ and study the equations \eqref{eq:CSF-Raychaudhuri-v} and \eqref{eq:Maxwell-v} for $v\in [0,1]$ with the $\phi_\alpha$ ansatz:
\begin{align}
 \label{Ray-4}  \partial_v^2 r(v;\alpha)&=- |\alpha|^2r(v;\alpha) |\partial_v\phi_{\hat\alpha}(v)|^2,\\
   \label{maxwell-4} \partial_v Q(v;\alpha)&= \mathfrak e |\alpha|^2 r(v;\alpha)^2 \Im(\phi_{\hat\alpha}(v)\overline{\partial_v\phi_{\hat\alpha}(v)}).
\end{align}
In addition, we again define $r$ at $v=1$ by 
\begin{align}
     r(1;\alpha)&= r_+, \nonumber \\
    \partial_v r(1;\alpha)&=0, \nonumber
    \end{align}
   and   $Q$ at $v=0$ by  
\begin{equation}
    Q(0;\alpha)=0, \label{eq:initialization-of-Q}
\end{equation}
which together with  \eqref{Ray-4}  and \eqref{maxwell-4} uniquely determine $r$ and $Q$ on $ [0,1]$. Note that we will initialize $\partial_u r$ only later in \eqref{ni-init}. 

We first note that basic calculations yield
\begin{equation}
    |\partial_v\phi_{\hat\alpha}|^2 = \sum_{1\le j\le 2k+1}\hat\alpha_j^2 \left(\chi_j^2+\chi_j'^2\right)\nonumber
\end{equation}
and 
\begin{equation}
    \Im(\phi_{\hat\alpha}\overline{\partial_v\phi_{\hat\alpha}})=\sum_{1\le j\le 2k+1}\hat\alpha_j^2\chi_j^2.\nonumber
\end{equation} 
Therefore, 
\begin{equation}
    \int_0^{1}|\partial_v\phi_{\hat\alpha}|^2 \,dv\approx \int_{0}^{1}\Im(\phi_{\hat\alpha}\overline{\partial_v\phi_{\hat\alpha}})\,dv\approx 1\nonumber
\end{equation}
for any $\hat\alpha\in S^{2k}$.

\begin{lem}\label{lem:radius-boostrap}
There exists a constant $0<c\les 1$ such that if $0<\beta \le c$, then for any $\hat\alpha\in S^{2k}$, $r(\cdot;\beta \hat \alpha)$ satisfies \begin{align}
 r(v;\beta \hat \alpha)&\ge \tfrac 12 r_+\label{eq:r-lower-bound}\\
    \partial_v r(v;\beta \hat \alpha)&\ge 0\label{no-trapped-surfaces}
\end{align} 
for $v\in [0,1]$, where \begin{equation*}
    r_+\doteq \left(1+\sqrt{1-\mathfrak q^2}\right)M_f.
\end{equation*}
Furthermore,  
\begin{align}
\partial_v r (0;\beta \hat \alpha) >0\label{condition1-radius-intial}.
\end{align}
\end{lem}
\begin{proof}
This is a simple bootstrap argument in $v$. Assume that on $[v_0,1]\subset [0,1]$, we have 
\begin{align*}
    \inf_{[v_0,1]}r &\ge 0\\
    \inf_{[v_0,1]}\partial_v r&\ge 0.
\end{align*}
This is clear for $v_0$ close to $1$ by Cauchy stability. From Raychaudhuri's equation \eqref{Ray-4}, $r\ge 0$ implies $\partial_v r$ is monotone decreasing, hence is bounded above by $\partial_vr(v_0)$, which can be estimated by
\begin{equation}
    \partial_v r(v_0)=\int_{v_0}^{1}\beta^2 r |\partial_v\phi_{\hat\alpha}|^2\,dv \les \beta^2 r_+,\label{lambda-estimate}
\end{equation}
since $r\le r_+$ on $[v_0,1]$. It follows that 
\begin{equation}
    r(v_0)= r_+-\int_{v_0}^{1}\partial_v r\,dv \ge r_+-C\beta^2r_+
\end{equation}
for some $C\lesssim 1$. Choosing $\beta>0$ sufficiently small shows $ r(v_0) \ge \tfrac 12 r_+$ which improves the bootstrap assumptions and proves the desired estimate \eqref{eq:r-lower-bound}. Finally, note that \eqref{condition1-radius-intial} holds true as $\partial_vr $ is monotone decreasing and $r$ is not constant ($\beta > 0$ and the scalar field is not identically zero). 
\end{proof}

\begin{lem} \label{lem:charge-increasing}
  By potentially making the constant $c$ from \cref{lem:radius-boostrap} smaller, we have that for any $0<\beta \le  c$  and $\hat\alpha\in S^{2k}$, the following estimate holds
\begin{equation*}
    \frac{\partial}{\partial\beta} Q(1;\beta,\hat\alpha)>0.
\end{equation*}
\end{lem}
\begin{proof}
Integrating Maxwell's equation \eqref{maxwell-4} and using \eqref{eq:initialization-of-Q}, we find 
\begin{equation*}
    Q(1;\beta,\hat\alpha)=\int_0^{1}\mathfrak e\beta^2 r^2 \Im(\phi_{\hat\alpha}\overline{\partial_v\phi_{\hat\alpha}})\,dv.
\end{equation*}
A direct computation yields 
\begin{equation*}
    \partial_\beta Q(1;\beta,\hat\alpha)=2\mathfrak e\beta\int_0^{1} (r^2+\beta r\partial_\beta r) \Im(\phi_{\hat\alpha}\overline{\partial_v\phi_{\hat\alpha}})\,dv. 
\end{equation*}
Note that $\Im(\phi_{\hat\alpha}\overline{\partial_v\phi_{\hat\alpha}})\ge 0$ pointwise and is not identically zero. Since $0<\beta\leq c$, we use \cref{lem:radius-boostrap} to estimate
\begin{equation*}
    r^2+\beta r\partial_\beta r \ge \tfrac 14 r_+^2 -C \beta r_+^2=r_+^2(\tfrac 14 - C\beta),
\end{equation*}
 where we also used $|\partial_\beta r|\les r_+$ which follows directly from differentiating \eqref{Ray-4} with respect to $\beta=|\alpha|$. Therefore, by choosing $c$ even smaller, we obtain $\partial_\beta Q(1;\beta,\hat\alpha)>0$. 
\end{proof}

\begin{lem}\label{betaQ}
If $\mathfrak eM_f/\mathfrak q$ is sufficiently large depending only on $k$ and the choice of profiles, then there is a smooth function $ \beta_Q:S^{2k}\to (0,\infty)$ so that $Q(1;\beta_Q(\hat \alpha),\hat\alpha)=\mathfrak qM_f$ for every $\hat\alpha\in S^{2k}$, which also satisfies  
\begin{align}
    \beta_Q(\hat \alpha) &\approx \frac{\sqrt{\mathfrak qM_f}}{\sqrt{\mathfrak e}r_+}\label{beta-estimate}\\
    \beta_Q(-\hat \alpha ) &= \beta_Q(\hat \alpha) \label{eq:beta-is-odd}
\end{align}
for every $\hat \alpha \in S^{2k}$.
\end{lem}
\begin{proof}
As in the proof of \cref{lem:charge-increasing} we have
\begin{equation*}
    Q(1;\beta,\hat\alpha)=\mathfrak e \beta^2\int_0^{1}r^2 \Im(\phi_{\hat\alpha}\overline{\partial_v\phi_{\hat\alpha}}).
\end{equation*}
If $\beta$ is sufficiently small so that \cref{lem:radius-boostrap} and \cref{lem:charge-increasing} apply, we estimate 
\begin{equation*}\label{eq:estimate-on-Q}
 Q(1;\beta,\hat\alpha)\approx \mathfrak e \beta^2 r_+^2.
\end{equation*}
For $\mathfrak e M_f/\mathfrak q $ sufficiently large as in the assumption, we apply now the intermediate value theorem,  to obtain a $\beta_Q(\hat\alpha)$ satisfying  $0< \beta_Q(\hat\alpha) \leq c$ such that
\begin{equation}\label{eq:defin-of-Q}
Q(1;\beta_Q,\hat\alpha)=\mathfrak qM_f.
\end{equation} 
Note that $\beta_Q(\hat\alpha)$ is unique since $Q(1;\cdot,\hat\alpha)$ is strictly increasing as shown in \cref{lem:charge-increasing}. Moreover, since $Q(1; \cdot,\cdot )$ is smooth (note that $\hat \alpha\in S^{2k}$ and $\beta>0$ enter as smooth parameters in \eqref{maxwell-4} which defines $Q$), a direct application of the implicit function theorem using that $\partial_\beta Q(1;\cdot,\hat\alpha)\ne 0$  shows that $\beta_Q:S^{2k}\to (0,\infty)$ is smooth. 

Moreover,  by \eqref{eq:estimate-on-Q} and \eqref{eq:defin-of-Q}, $\beta_Q$ satisfies 
\begin{equation*}
    \mathfrak e \beta^2_Q r_+^2 \approx \mathfrak q M_f
\end{equation*}
which shows \eqref{beta-estimate}.
Finally, note that  $Q(1;\beta,-\hat\alpha)=Q(1;\beta,\hat\alpha)$, from which \eqref{eq:beta-is-odd} follows.
\end{proof}

\begin{lem}\label{lem:defn-of-mathfrak-Q}
 Let $\mathfrak e M_f/\mathfrak q$ be sufficiently large (depending only on $k$  and the choice of profiles) so that 
 \cref{betaQ} applies. Then 
 \begin{align*}
    p_Q: S^{2k} & \to \mathfrak  Q^{2k}\\
    \hat \alpha &\mapsto \beta_Q(\hat \alpha) \hat \alpha
\end{align*}
is a diffeomorphism, where 
 \begin{equation*}
    \mathfrak Q^{2k}\doteq \{\beta_Q(\hat \alpha) \hat \alpha \colon \hat \alpha \in S^{2k} \} \subset \R^{2k+1}
\end{equation*} is the radial graph of $\beta_Q$. 
 Moreover,  $  \mathfrak Q^{2k}$ is invariant under the antipodal map $A(\alpha)=-\alpha$ and $p_Q$ commutes with the antipodal map.
\end{lem}
\begin{proof} By definition of $ \mathfrak Q^{2k}$ and the facts that $\beta_Q$ is smooth, positive, and invariant under the antipodal map as proved in \cref{betaQ}, the stated properties of  $ \mathfrak Q^{2k}$ and $p_Q$ follow readily.
\end{proof}

Having identified the set $\mathfrak Q^{2k}$ which guarantees gluing of the charge $Q$, for the rest of the section we will always take $\alpha\in \mathfrak Q^{2k}$. Recall from \eqref{beta-estimate} that for every $\alpha \in\mathfrak Q^{2k} $: \begin{equation}
    \label{eq:norm-of-a}
|\alpha|\approx\frac{\sqrt{\mathfrak qM_f}}{\sqrt{\mathfrak e}r_+} .\end{equation} 

Before proceeding to choose sphere data, we will need to examine the equation for $\partial_u r$ because this will place a further restriction on $\alpha$ which must be taken into account before setting up the topological argument. We continue by using the definition of the Hawking mass $m$ in  \eqref{eq:Hawking-mass}, to impose the condition 
\begin{align}
    m(0;\alpha)&=M_i\label{condition-mass-initial}\nonumber 
\end{align}
 by initializing 
\begin{equation}
    \partial_ur(0;\alpha)= -\left(1-\frac{2M_i}{r(0;\alpha)}\right)\frac{1}{4\partial_v r(0;\alpha)}.\label{ni-init}
\end{equation}
The transverse derivative $\partial_u r(v;\alpha)$ is now determined by solving \eqref{eq:CSF-equation-for-r},
\begin{equation}
\label{r-wave-4}     \partial_v \partial_u r(v;\alpha)=- \frac{1}{4r(v;\alpha)^2}-\frac{\partial_ur(v;\alpha)\partial_vr(v;\alpha)}{r(v;\alpha)^2}+\frac{Q(v;\alpha)^2}{4r(v;\alpha)^3},
\end{equation}
with initialization \eqref{ni-init}.

Note that \eqref{ni-init} is  well-defined by \eqref{condition1-radius-intial} and \eqref{eq:r-lower-bound} from \cref{lem:radius-boostrap}. Furthermore,
\[1-\frac{2M_i}{r(0;\alpha)}\ge 1-\frac{4M_i}{M_f}>0,\]
so
\begin{equation}
\partial_ur(0;\alpha)<0.\label{ni-negative}
\end{equation}

Having initialized $\partial_u r $ at $v=0$, we determine $\partial_u r(v;\alpha)$ using \eqref{r-wave-4}, and we will now show that for  $\mathfrak eM_f/\mathfrak q$ sufficiently large, $\partial_u r(v;\alpha) <0$ for all $v\in [0,1]$.

\begin{lem}\label{lem:no-antitrapped}
If $\mathfrak eM_f/\mathfrak q$ is sufficiently large depending only on $k$ and the choice of profiles and if $0\le M_i\le \tfrac 18 M_f$, then
\begin{equation}\label{eq:lemma-no-antitrapped}
    \sup_{v\in [0,1]}\partial_u r(v;\alpha) <0
\end{equation}
for every $\alpha\in \mathfrak Q^{2k}$.
\end{lem}
\begin{proof}
Since $r>0$ on $[0,1]$, it suffices to show that 
\begin{equation}
    \sup_{[0,1]}r\partial_u r<0.\nonumber
\end{equation}
First, by  \eqref{rdur},
\begin{equation}
    |\partial_v(r\partial_u r)|= \left|\frac 14 \left(1-\frac{Q^2}{r^2}\right)\right|\les 1, \label{eq:partialv-r-partialur}
\end{equation}
as 
\begin{equation*}
    Q(v;\alpha) \leq  Q(1;\alpha)  = \mathfrak q M_f  \lesssim r(v;\alpha),
   \end{equation*}
   where we used \eqref{eq:r-lower-bound}.
Integrating \eqref{eq:partialv-r-partialur}, we have 
\begin{equation}
  \sup_{v\in[0,1]}  r(v)\partial_ur(v)\le r(0)\partial_ur(0)+C_1,\label{antitrapped-eq-1}
\end{equation}
where $C_1\les 1$ is a constant.
Analogously to \eqref{lambda-estimate}, we estimate 
\begin{equation*}\partial_v r(0;\alpha)\lesssim |\alpha|^2 r_+ \les \frac{\mathfrak q}{\mathfrak e},\end{equation*}
where we used \eqref{eq:norm-of-a}. 
Now, using \eqref{ni-init},
\begin{align*}
    -r(0)\partial_u r(0) &=\frac{r(0)-2M_i}{4\partial_v r(0)}\\
    &\gtrsim \frac{\mathfrak e}{\mathfrak q} (\tfrac 12 M_f-2M_i)\\
    &\gtrsim \frac{\mathfrak e}{\mathfrak q}M_f.
\end{align*}
Therefore, we improve \eqref{antitrapped-eq-1} to 
\begin{equation*}
    \sup_{v\in[0,1]}   r(v)\partial_ur(v) \le -C_2\frac{\mathfrak e}{\mathfrak q}M_f+C_1
\end{equation*} 
for some $C_2\lesssim 1$. Thus,  if $\mathfrak eM_f/\mathfrak q$ is sufficiently large we obtain \eqref{eq:lemma-no-antitrapped}. 
\end{proof}

To continue the proof of \cref{main-thm-RN}, we  now put our construction into the framework of the sphere data in  \cref{subsec:sphere-data-mink-schwarz}.
For each $\alpha\in\mathfrak Q^{2k}$, define $D_\alpha(0)\in\mathcal D_k$ by setting 
\begin{itemize}
    \item $\varrho=r(0;\alpha)\ge \tfrac 12 r_+$ (see \eqref{eq:r-lower-bound}),
    \item $\varrho_v^1= \partial_vr(0;\alpha)>0$ (see \eqref{condition1-radius-intial}),
    \item $\varrho_u^1=\partial_u r(0;\alpha)<0$ (see \eqref{ni-init} and \eqref{ni-negative}),
    \item $\omega=1$, and 
    \item all other components to zero.
\end{itemize}
By \cref{lem:gauge-Mink}, $D_\alpha(0)$ is equivalent to $D^\mathrm{S}_{M_i,r(0;\alpha),k}$ up to a gauge transformation. 

For each $\alpha\in \mathfrak Q^{2k}$, we now apply \cref{prop:null-data-existence} and \cref{prop:seed-data} to uniquely determine cone data \[D_\alpha:[0,1]\to \mathcal D_k,\]
 with initialization $D_\alpha(0)$ above and seed data $\phi_\alpha$ given by \eqref{eq:defn-phi-RN}. By standard ODE theory, $D_\alpha(v)$ is jointly continuous in $v$ and $\alpha$. Note that $\varrho(D_\alpha(v))=r(v;\alpha)$, $\varphi(D_\alpha(v))=\phi(v;\alpha)$, and $q(D_\alpha(v))=Q(v;\alpha)$ by definition. As in the proof of \cref{main-thm-Schw}, we use the notation 
 \[\partial_u^i \phi(v;\alpha)\doteq \varphi_u^i(D_\alpha(v))\]
 for $i=1,\dotsc,k$ to denote the transverse derivatives of the scalar field obtained by \cref{prop:null-data-existence}. Note also that 
 \[\partial_u r(v;\alpha)=\varrho_u^1(D_\alpha(v)),\]
 where $\partial_u r(v;\alpha)$ is as in \eqref{eq:CSF-equation-for-r} above. 

By construction, the data set $D_\alpha(1)$ satisfies
\begin{itemize}
\item $\varrho =2M_f,$
\item $\varrho_u^1 <0$ (see \cref{lem:no-antitrapped}),
\item $\varrho_v^1=0$, 
\item $\omega=1$,
\item $q=\mathfrak qM_f$ (definition of $\mathfrak Q^{2k}$), and
\item $\varphi_v^i=0$ for $0\le i\le k$.
\end{itemize}

In order to glue to the appropriate Reissner--Nordstr\"om event horizon sphere, by \cref{lem:gauge-Schw}, it suffices to find an $\alpha_*\in \mathfrak Q^{2k}$ for which additionally 
\[\partial_u\phi(1;\alpha_*)=\cdots=\partial_u^k\phi(1;\alpha_*)=0.\]
Analogously to \cref{lem:odd} we first establish
\begin{lem}\label{CSF-odd}
The metric coefficients $r(v;\alpha)$, $\Omega^2(v;\alpha)$, the electromagnetic quantities $Q(v;\alpha)$, $A_u(v;\alpha)$, and all their ingoing and outgoing derivatives are even functions of $\alpha$. The scalar field $\phi(v;\alpha)$ and all its ingoing and outgoing derivatives are odd functions of $\alpha$.
\end{lem}
\begin{proof}
The proof is essentially the same as \cref{lem:odd}, noting that equations \eqref{eq:Maxwell-u}, \eqref{eq:Maxwell-v}, and \eqref{eq:Au-transport} are also even in $\phi$. 
\end{proof}

We now complete the proof of \cref{main-thm-RN}. Recall from  \cref{lem:defn-of-mathfrak-Q} that $  p_Q: S^{2k}  \to \mathfrak  Q^{2k}$ is a diffeomorphism which commutes with the antipodal map. 
We now argue similarly to \cref{subsec:Schwarzschild-gluing-proof}. By \cref{CSF-odd}, the function 
\begin{align*}
    F:\mathfrak Q^{2k}&\to \Bbb C^k\\
   \alpha &\mapsto \left(\partial_u\phi(1;\alpha),\dotsc ,\partial_u^k\phi(1;\alpha)\right)
\end{align*}
is continuous and odd. Therefore, the Borsuk--Ulam theorem, stated as  \cref{BorsukUlam}, applied to 
\[(\Re F^1,\Im F^1,\dotsc, \Re F^k,\Im F^k)\circ p_Q :S^{2k}\to \Bbb R^{2k},\]
where $F^i$ is the $i$th component of $F$, shows that there is an $\alpha_*\in\mathfrak Q^{2k}$ such that $F(\alpha_*)=0$. By \cref{lem:gauge-Schw}, $D_{\alpha_*}(1)$ is gauge equivalent to $D^\mathrm{RN\mathcal H}_{M_f, \mathfrak q M_f,k}$ which concludes the gluing construction. Since we have already established that $\partial_u r <0$ for all $v\in [0,1]$ in \cref{lem:no-antitrapped}, this concludes the proof of \cref{main-thm-RN}. \qed

\subsection{Proof of \texorpdfstring{\cref{thm:interior-gluing}}{Theorem 2C}} \label{sec:non-horizon-gluing}

In this section we extend our characteristic gluing result \cref{main-thm-RN}  to allow for sphere data at the final sphere which is not necessarily located on a horizon. Recall \cref{defn:non-horizon-data} for the definition of general Reissner--Nordström sphere data.  As the steps in the proof below are direct generalizations of the proof of \cref{main-thm-RN}, our presentation here will have fewer details.

\begin{proof}[Proof of \cref{thm:interior-gluing}]
We only consider the case $\mathfrak q \neq 0$, the case $\mathfrak q=0$ being strictly easier and requiring only ``gluing 3'' below. Without loss of generality, we may also assume $R_f \leq3  M_f$  as for $ r\geq  3 M_f$ we can extend trivially with Reissner--Nordström data satisfying $\partial_vr >0$ and $\partial_u r <0$. In the following proof, we use the convention that all constants appearing in $\lesssim, \gtrsim$ and $\approx$ to also depend on $\mathfrak q, \mathfrak r$ and $\mathfrak e$. The theorem is proved as a consequence of the following three intermediate gluings:
    \begin{enumerate}
        \item $D^\mathrm{M}_{R_i,k}$ is glued to $D^\mathrm{RN}_{M',Q_f,R_1,k}$ with a complex scalar field,
        \item $D^\mathrm{RN}_{M',Q_f,R_1,k}$ is glued to $D^\mathrm{RN}_{M',Q_f,R_2,k}$ trivially (i.e., with identically vanishing scalar field), and 
        \item $D^\mathrm{RN}_{M',Q_f,R_2,k}$ is glued to $D^\mathrm{RN}_{M_f,Q_f,R_f,k}$ with a real scalar field,
    \end{enumerate}
    where $R_i \doteq R_f - M_f^{3/4}$, $0<M'<M_f$ is an intermediate modified Hawking mass, $Q_f\doteq\mathfrak qM_f$, $R_1,R_2$ are intermediate radii which satisfy $R_i<R_1<R_2$. 

    \textbf{Gluing 1.}  In the interval $v\in [0,1]$ we impose the ansatz \eqref{eq:defn-phi-RN}. At $v=0$, we set
\begin{align}\label{eq:initializations-interior-gluing}
        r(0)&=R_i, \quad   m(0)=Q(0)=0, \quad \partial_ur(0)=-\frac{1}{2M_f^{1/2}}, \quad \partial_vr(0)=\frac{M_f^{1/2}}{2}.
    \end{align}
    
    The pulse parameters $\alpha_*$ which achieve gluing of transverse derivatives of $\phi$ are determined by the procedure of \cref{sec:gluing-rn}, with charge condition $Q(1;\alpha)= Q_f$. As in \cref{sec:gluing-rn} we find that the gluing can be performed with parameters satisfying $|\alpha_*|^2 \lesssim M^{-1}_f$. Using this estimate on $\alpha_*$,  we obtain from  Raychaudhuri's equation \eqref{eq:CSF-Raychaudhuri-v} and \eqref{eq:initializations-interior-gluing} that $\frac 12 M_f^{1/2}\geq \partial_v r \geq \frac 14 M_f^{1/2}$ for every $v\in[0,1]$  by choosing $M_0(k,\mathfrak q, \mathfrak e, \mathfrak r )$ sufficiently large. This also implies $R_i  \leq r \leq R_i + \frac 12 M_f^{ 1/2}$.    
 Using $r \geq R_f - M_f^{3/4}$ and the estimate analogous to  \eqref{eq:partialv-r-partialur} we  infer $|r\partial_u r- r(0) \partial_u r(0) |\lesssim 1$  for every $v\in [0,1]$, i.e., $0 < -\partial_u r\leq M_f^{-1/2}$. 
 We now estimate the  Hawking mass at $v=1$ by integrating \eqref{eq:equation-for-hawking-mass},
 \begin{equation} \label{eq:estimate-on-hawing_mass} m  (1) = \int_0^1 2  r^2 (-\partial_u r ) |\partial_v\phi|^2\, d v  + \int_0^1  \frac{Q^2}{2r^2}\partial_v r\, dv \lesssim M_f^2 M_f^{-1/2}    M_f^{-1} + M_f^{1/2}\lesssim  M_f^{1/2}.
 \end{equation} 
Setting $R_1=r(1)$ and $M'=m(1)+Q_f^2/(2R_1)$, we have shown that $R_i < R_1 \leq R_i + \frac 12 M_f^{1/2}$. The condition \eqref{eq:condition-on-Rf} shows that $Q_f^2/(2R_f) \leq 
M_f/ (1+\mathfrak r) $. In particular, since $R_1 \geq R_f - M_f^{3/4}$ we estimate  \begin{equation}M'= m(1)+\frac{Q_f^2}{2 R_1} \leq \frac 12 \left(1 + \frac{1}{1+\mathfrak r}\right)   M_f = \frac{2+\mathfrak r }{2 + 2 \mathfrak r}M_f \label{eq:M1}\end{equation}
by possibly taking $M_0(k,\mathfrak q,\mathfrak e, \mathfrak r)$ larger. This completes the first gluing step. 
 
    \textbf{Gluing 3.} It is more convenient to now carry out the third gluing step and simply ensure that $R_2>R_1$. We use a collection of  $k+1$ real-valued pulses as in \eqref{def-phi} on $v\in [0,1]$. We impose
    \begin{equation}
        r(1)=R_f,\quad \partial_ur(1)=-M_f,\quad Q(1)=Q_f,\quad \varpi(1)=M_f.
    \end{equation}
   This uniquely determines  $\partial_v r(1)$ which can have either sign but satisfies $|\partial_v r(1)|\lesssim M_f^{-1}$. We also note that as long as $|\alpha|^2 \leq M_f^{-3/2}$, we have  $|\partial_v r|\lesssim M_f^{-1/2}$ and thus $|r-R_f|\lesssim  M_f^{-1/2}$ on $[0,1]$. This also gives $-\partial_u r \approx M_f$. Using
  \begin{equation*}
    \varpi(1)- \varpi(0) = \int_0^12 r^2 (-\partial_u r ) |\partial_v\phi|^2 \,d v
  \end{equation*} and \eqref{eq:M1},
  we write the mass condition $\varpi(1)=M_f$ and $\varpi(0) = M'$ as a sphere of $\alpha$'s  ($ |\alpha|^2\approx     M_f^{-2}$) for which we will apply the Borsuk--Ulam argument. We use here that  $ M_f - M' = M_f  \mathfrak r /(2+2\mathfrak r)$. With $ |\alpha_*|^2\approx     M_f^{-2}$  we have the improved estimate  $|\partial_v r |\lesssim M_f^{-1}$ for $v\in [0,1]$ and thus, $|r(0) -R_f|\lesssim M_f^{-1}$. Taking now $M_0(k,\mathfrak q,\mathfrak e, \mathfrak r)$ sufficiently large makes $R_2\doteq r(0)>R_1$. 

\textbf{Gluing 2.} By the previous constructions, we have $R_1<R_2$, $\varpi(0)=\varpi(1)=M'$, $Q(0)=Q(1)=Q_f$, $\partial_vr(0)>0$, and $\partial_ur(0)<0$. Now $D^\mathrm{RN}_{M',Q_f,R_k}$ can be trivially glued to $D^\mathrm{RN}_{M',Q_f,R_2,k}$ by choosing $\phi\equiv 0$, and we must merely ensure that $\partial_ur<0$ along the way. Since $\partial_vr>0$ by Raychauduri's equation, this amounts to proving $\frac{2m}{r}<1$. Indeed,
\begin{equation*} m(v)  \leq m(0) + \int_0^1 \frac{Q_f^2}{2 r^2} \partial_v r\, dv = m(0) + \int_{R_1}^{R_2} \frac{Q_f^2}{2r^2}\,dr  \lesssim m(0)+ (R_2 - R_1) \lesssim   M_f^{1/2} +M_f^{3/4}, \end{equation*} 
where we used \eqref{eq:estimate-on-hawing_mass}.  In particular,  by choosing $M_0(k,\mathfrak q, \mathfrak e,\mathfrak r) $ larger, we can make $m(v)/M_f$ arbitrarily small and thus $\partial_u r <0$ throughout gluing~2. 
\end{proof}

\section{Constructing the spacetimes and Cauchy data}\label{sec:construction-spacetime-cauchy-data}

In this final section we will prove our main result \cref{cor:trapped-surfaces-ERN} as well as  \cref{main-corollary}, \cref{cor:Cauchy-horizon-RN}, and \cref{cor:Cauchy-horizon-closes--off}.

\subsection{Maximal future developments of asymptotically flat data for EMCSF}\label{MFGHD}

Our theorems and corollaries in this paper are stated in the framework of the Cauchy problem for the Einstein--Maxwell-charged scalar field system.  We recall that Cauchy data for the EMCSF system consist of the usual Cauchy data $(\Sigma,g_0,k_0)$ for the Einstein equations, where $\Sigma$ is a 3-manifold, $g_0$ a Riemannian metric on $\Sigma$, and $k_0$ a symmetric $2$-tensor field, together with initial data for the matter fields, namely initial electric and magnetic fields, $E_0$ and $B_0$, and finally the scalar field $\phi_0$ and its ``time derivative'' $\phi_1$. (See e.g.\ \cite[Section VI.10]{YCB-book} for a treatment of the Einstein--Maxwell Cauchy problem.) Associated to a Cauchy data set is a unique maximal future globally hyperbolic development $(\mathcal M^4,g,F,A,\phi)$ \cite{MR53338,CBG69}. If the Cauchy data are moreover spherically symmetric, then the maximal development will be spherically symmetric by uniqueness. 

We will not, however, actually construct our spacetimes by directly evolving Cauchy data. Rather, we construct the spacetimes \emph{teleologically} by gluing together explicit spacetimes with the help of our characteristic gluing results and \cref{prop:spacetime-gluing}. In each case, a Cauchy hypersurface $\Sigma$ is then found, within the spacetime, whose future domain of dependence contains the physically relevant region, and contains no antitrapped spheres. At this point, all attention is restricted to this future domain of dependence. \emph{A posteriori}, by the existence and uniqueness theory for the maximal globally hyperbolic development, the spacetime will then be contained in the maximal development of the induced data on the Cauchy hypersurface $\Sigma$. 

\begin{figure}[ht]
\centering{
\def\svgwidth{18pc}
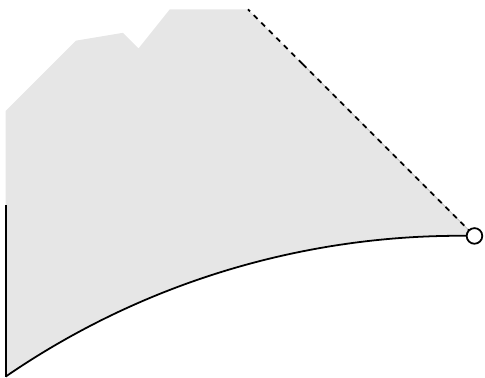}
\caption{General structure of the MFGHD of asymptotically flat Cauchy data $\Sigma$ in the EMCSF system in spherical symmetry \cite{Kommemi13}. What is depicted is the quotient manifold $\mathcal Q$ as a bounded subset of $\Bbb R^{1+1}_{u,v}$ with boundary suitably labeled. Note that various components of the diagram can be empty.}
\label{Kommemi}
\end{figure}

Since our examples are maximal globally hyperbolic developments of asymptotically flat, spherically symmetric Cauchy data for the EMCSF system with no antitrapped spheres of symmetry, we can make use of a general characterization of the boundary of spacetime in this context appearing in \cite{Kommemi13}. In particular one can rigorously associate a global Penrose diagram, and unambiguously identify a nonempty null boundary component \emph{future null infinity} $\mathcal I^+$, \emph{domain of outer communication} $J^-(\mathcal I^+)$, (possibly empty) \emph{black hole region} $\mathcal{BH}\doteq \mathcal M\setminus J^-(\mathcal I^+)$, (possibly empty) \emph{event horizon} $\mathcal H^+\doteq\partial(\mathcal{BH})$,  (possibly empty) \emph{Cauchy horizon} $\mathcal{CH}^+$, (possibly empty) $r=0$ singularity $\mathcal S$, and (possibly empty) null boundary component $\mathcal N$ emanating from a (possibly absent) ``locally naked'' singularity at the center. The Penrose diagram $\mathcal Q\subset \Bbb R^{1+1}_{u,v}$ can be viewed as a global double null chart for the spacetime, with $v$ the ``outgoing'' null coordinate and $u$ the ``ingoing'' coordinate. See \cref{Kommemi}.\footnote{Note that the above general boundary decomposition in particular proves that one cannot form a globally naked singularity once a marginally trapped surface has developed in the spacetime, which already rules out naked singularity formation by supercharging a black hole in spherical symmetry, see \cite[Section 1.9]{Kommemi13}. It is thus not at all surprising that ongoing numerical searches for these continue to be futile.} 

For use in the statement and proof of \cref{trapped-surfaces-ERN-proof} below, we recall that in the EMCSF system in spherical symmetry, the \emph{apparent horizon} is defined by 
\[\mathcal A\doteq\{\partial_v r=0\}\subset\mathcal{BH}.\] 
Since $\mathcal A$ might have a complicated structure (in particular, it might have nonempty interior), we define an appropriate notion of boundary as follows. The \emph{outermost apparent horizon} $\mathcal A'$ consists of those points $p\in\mathcal A$ whose past-directed ingoing null segment lies in the strictly untrapped region $\{\partial_v r>0\}$ and eventually exits the black hole region, i.e., enters $J^-(\mathcal I^+)$. $\mathcal A'$ is a possibly disconnected achronal curve in the $(1+1)$-dimensional reduction $\mathcal Q$ of $\mathcal M$. Note, as depicted in \cref{Kommemi}, that $\mathcal A'$ does not necessarily asymptote to future timelike infinity $i^+$.

For definiteness, we will make extensive use of these notions in our theorems and corollaries. However, our notation and usage should be sufficiently familiar to readers acquainted with standard concepts in general relativity so that they may read our diagrams and understand our theorems without specific reference to \cite{Kommemi13}. 

We also note that when referring to spherically symmetric subsets of $(\mathcal M,g)$, such as the event horizon $\mathcal H^+$, we may view them as objects in $\mathcal M$ or in the reduced space $\mathcal Q$. The context will make it clear which point of view we are taking. 

\begin{rk}
In \cref{app-B}, we show by a barrier argument that since $\partial_u r<0$ in a spacetime satisfying the hypotheses of \cite{Kommemi13}, there are also no nonspherically symmetric antitrapped surfaces. 
\end{rk}

\subsection{Construction of gravitational collapse to Reissner--Nordstr\"om}
\label{subsec:grav-collapse-RN}

We now state a more precise version of \cref{main-corollary} as follows. 
\setcounter{cor}{0}
\begin{cor}\label{corollary-proof}
For any $k\in\Bbb N$, $\mathfrak q\in [-1,1]\setminus \{0\}$, and $\mathfrak e\in\Bbb R\setminus\{0\}$, let $M_0(k,\mathfrak q,\mathfrak e)$ be as in \cref{main-thm-RN}. Then for any $M\ge M_0$ there exist asymptotically flat, spherically symmetric Cauchy data $(\Sigma,g_0,k_0,E_0,B_0,\phi_0,\phi_1)$ for the EMCSF system, with $\Sigma\cong \Bbb R^3$ and a regular center, such that the maximal future globally hyperbolic development $(\mathcal M^4,g,F,A,\phi)$ has the following properties:
\begin{itemize}
    \item All dynamical quantities are at least $C^{k}$-regular.
    \item Null infinity $\mathcal I^+$ is complete.
    \item The black hole region is nonempty, $\mathcal{BH}\doteq \mathcal M\setminus J^-(\mathcal I^+)\ne\emptyset$.
    \item The Cauchy surface $\Sigma$ lies in the domain of outer communication $J^-(\mathcal I^+)$. In particular, it does not intersect the event horizon $\mathcal H^+\doteq\partial(\mathcal{BH})$.
    \item The initial data hypersurface does not contain trapped surfaces.
    \item The spacetime does not contain antitrapped surfaces. 
    \item For sufficiently late advanced times $v\ge v_0$, the domain of outer communication, including the event horizon, is isometric to that of a Reissner--Nordstr\"om solution with mass $M$ charge to mass ratio $\mathfrak q$. For $v\ge v_0$, the event horizon of the spacetime can be identified with the event horizon of Reissner--Nordstr\"om. 
\end{itemize}
\end{cor}
\begin{rk} \label{rk:Schwarzschild-case}
A similar statement can be made with $\mathfrak q=0$ for the Einstein-scalar field model, using instead \cref{main-thm-Schw}. In that case, there will also be no assumption made on the mass. 
\end{rk}

\begin{figure}[ht]
\centering{
\def\svgwidth{29pc}
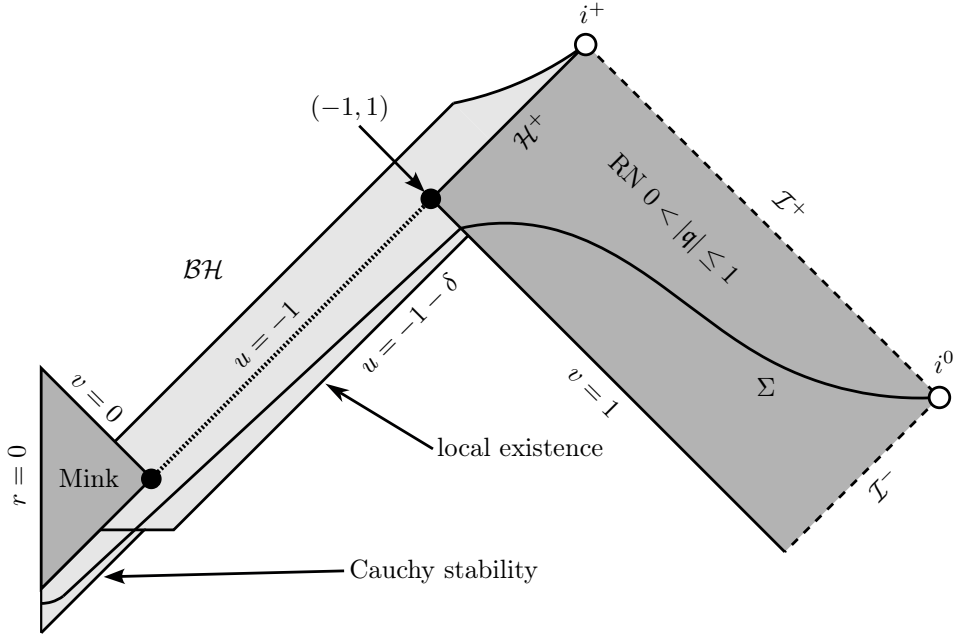}
\caption{Penrose diagram for the proof of \cref{corollary-proof}.}
\label{Penrose-diagram-proof}
\end{figure}
\begin{proof}
We refer the reader to \cref{Penrose-diagram-proof} for a visual guide to the proof. Using \cref{main-thm-RN} with regularity index $k+1$ (see footnote below) and \cref{prop:spacetime-gluing}, a portion of Minkowski space
\begin{align*}
t+r&\le \tfrac 12 M,\\
t-r&\ge -\tfrac 12 M,
\end{align*}
can be glued to a Reissner--Nordstr\"om solution with parameters $M$ and $\mathfrak qM$. Note that as depicted, one can solve for a complete future neighborhood of the event horizon, which might not be a complete double null neighborhood.

 Since we are in spherical symmetry, standard techniques (see \cite[Section 5]{TheBVPaper} or \cite[Section 3]{luk-oh-yang}) allow the ``local existence'' region emanating from the Reissner--Nordstr\"om portion of the spacetime to be extended all the way up to the center.\footnote{The wave equation in spherical symmetry loses one derivative at the center when compared to characteristic data. Therefore, to obtain a globally $C^k$ solution, we take $C^{k+1}$ characteristic data.} (In this figure, this region is denoted ``Cauchy stability'' for reasons that will become clear below.)
 
 We now identify a spacelike curve $\Sigma$ connecting spacelike infinity $i^0$ in the exactly Reissner--Nordstr\"om region to the center, to the past of the cone $u=-1$. The curve $\Sigma$ can be chosen so the induced data on it is asymptotically flat near $i^0$. For example, it may be taken to be a constant $t$ curve near $i^0$ in standard coordinates. Furthermore, by having $\Sigma$ hug the gluing region closely enough, we are guaranteed to have no spherically symmetric antitrapped surfaces on $\Sigma$.

Completeness of null infinity $\mathcal I^+$ is inherited from the exact Reissner--Nordstr\"om solution. By inspecting \cref{Penrose-diagram-proof}, we see that the null hypersurface $C_{-1}$ is the event horizon $\mathcal H^+=\partial J^-(\mathcal I^+)$ of the spacetime and that $\Sigma$ can be arranged to lie in the domain of outer communication $J^-(\mathcal I^+)$. The statement about trapped surfaces follows from \cref{no-antitrapped} below.

We now consider the (unique) maximal future globally hyperbolic development $(\mathcal M^4,g,F,A,\phi)$ of the induced data  $(\Sigma,g_0,k_0,E_0,B_0,\phi_0,\phi_1)$ on $\Sigma$. By uniqueness of the MFGHD, it contains the domain of dependence of $\Sigma$ in the gluing spacetime (and thus all shaded regions to the future of $\Sigma$ in \cref{Penrose-diagram-proof}). Therefore, by construction, $(\mathcal M^4,g,F,A,\phi)$ has all the properties listed in the statement of \cref{corollary-proof}. Note that the property of having no antitrapped symmetry spheres is propagated to the whole development by Raychaudhuri's equation \eqref{eq:CSF-Raychaudhuri-u}. By \cref{no-antitrapped}, the spacetime does not contain any nonspherically symmetric antitrapped surfaces either.  This concludes the proof. 
\end{proof}

The above proof made use of spherical symmetry in the local existence region and the region up to the center. In view of potentially extending our work to the  Einstein vacuum equations in the future, we give a second construction of these regions which does not invoke spherical symmetry. First, the ``local existence region'' can be constructed outside of spherical symmetry by the well-known theorem of Luk \cite{Luk-local-existence}. Once such a region has been constructed, we can use the fact that it lies ``outside'' of a Minkowski region to   construct the rest of the spacetime, up to the center, by Cauchy stability: 

\begin{lem}
Let $B_{r_0}$ and $B_{r_1}$ denote the (open) balls of radii $r_0>0$ and $r_1>r_0$ in $\Bbb R^3$, respectively. Consider on $B_{r_1}$ data for the Einstein--Maxwell-charged scalar field system corresponding to Minkowski space, $(\delta,0,0,0,0,0)$. Let $D\doteq(g_0,k_0,E_0,B_0,\phi_0,\phi_1)$ be a $C^k$ (for $k\in\Bbb N$ sufficiently large and not assumed to be spherically symmetric) initial data set for the Einstein--Maxwell-charged scalar field system defined on $B_{r_1}$ which agrees with the Minkowski data set on $B_{r_0}$. Then the maximal globally hyperbolic development of $D$ contains the Minkowski cone over $B_{r_0}$ ``in its interior'' in the following sense:

There exists an $\ve>0$ and a development $(g,F,A,\phi)$ of the data $D$ on $K_{r_0+\ve}\doteq\{t+r< r_0+\ve\}\cap \{t\ge 0\}\subset \Bbb R^{3+1}$ so that the development of the Minkowski portion of the data is defined on $K_{r_0}\doteq \{t+r<r_0\}\cap \{t\ge 0\}$ and is the Minkowski metric in those coordinates. 
\end{lem}
\begin{proof}
Since this is a standard  Cauchy stability argument   we merely sketch the proof. For $0<\ve<\frac{r_1-r_0}{2}$, let $\theta_\ve$ be a cutoff function which is equal to one on $B_{r_0+\ve}$ and vanishes outside $B_{r_0+2\ve}$. On $B_{r_1}$, we consider the ``initial data set''
\begin{equation}
    D_\ve\doteq(\theta_\ve g_0 +(1-\theta_\ve)\delta,\theta_\ve k_0 , \theta_\ve E_0,\theta_\ve B_0,\theta_\ve \phi_0,\theta_\ve \phi_1).\nonumber
\end{equation}
This does not solve the constraints everywhere, but it does solve them on $B_{r_0+\ve}$, where it equals $D$. We assume that $k\ge 5$ and show that $D_\ve$ is $O(\ve)$-close to the Minkowski data set in $H^4$. Then Cauchy stability for the reduced Einstein equations (in harmonic coordinates) will show that a solution to the reduced equations with data $D_\ve$ exists on $K_{r_0+2\ve}$ for $\ve$ sufficiently small. By domain of dependence arguments, a genuine solution will then exist on a smaller domain which still contains the entirety of $\overline{K_{r_0}}$ in its interior. 

To show that $D_\ve$ is close to Minkowski data we must check it componentwise. For brevity, we only check $\theta_\ve k_0$. Note first that \begin{equation}
    \|\theta_\ve k_0\|_{H^4}\les \|\theta_\ve k_0\|_{C^4}.\nonumber
\end{equation}
Now since $k_0$ vanishes on $B_{r_0}$ and is at least $C^5$, Taylor's theorem implies 
\begin{equation}
    |\partial_r^i\slashed\nabla^j k_0|\les \max\{0,r-r_0\}^{5-i-j},\nonumber
\end{equation}
if $0\le i+j\le 5$. In the region where either $\theta_\ve$ or $\partial_r\theta_i$ are nonvanishing, $\max\{0,r-r_0\}\les \ve$. It follows that 
\begin{equation}
    \|\theta_\ve k_0\|_{H^4}\les \sum_{0\le i+j\le 4}\sup_{B_{r_1}}|\partial_r^i\slashed\nabla^j(\theta_\ve k_0)|\les \ve,\nonumber
\end{equation}
which proves the claim and hence the lemma. 
\end{proof}

\subsection{Construction of counterexample to the third law}\label{proof-trapped-surfaces-ERN}

In this section we prove \cref{cor:trapped-surfaces-ERN} with an analogous approach as in the proof of \cref{corollary-proof}. 
We first restate the result in more detail.  
\setcounter{thm}{0}
\begin{thm}\label{trapped-surfaces-ERN-proof} For any $k\in\Bbb N$ and $\mathfrak e\in\Bbb R\setminus\{0\}$, there exist asymptotically flat, spherically symmetric Cauchy data $(\Sigma,g_0,k_0,E_0,B_0,\phi_0,\phi_1)$, with $\Sigma\cong \Bbb R^3$ and a regular center, for the EMCSF system such that the maximal future globally hyperbolic development $(\mathcal M^4,g,F,A,\phi)$ has the following properties:
\begin{itemize}
          \item All dynamical quantities are at least $C^{k}$-regular.

    \item The spacetime and Cauchy data satisfy all the conclusions of \cref{corollary-proof} with $\mathfrak q=1$ and final mass $M_f\ge  M_0(1,\mathfrak e, k)+8$.
    
    \item The spacetime contains a double null rectangle of the form $\mathfrak R\doteq \{-2\le u\le -1\}\cap\{1\le v\le 2\}$ which is isometric to a double null rectangle in a Schwarzschild spacetime of mass $1$. 
    
    \item The cone $\{u=-1\}\cap\mathfrak R$ lies in the outermost apparent horizon $\mathcal A'$ of the spacetime and is isometric to an appropriate portion of the $r=2$ hypersurface in the Schwarzschild spacetime of mass $1$. 

 \item The outermost apparent horizon $\mathcal A'$ is disconnected.
    
    \item The spacetime contains trapped surfaces in the black hole region, for all arbitrarily late advanced time. More precisely, 
    for every symmetry sphere $S_{u,v} \subset \mathcal H^+$, $J^+(S_{u,v})$ contains a trapped sphere. 
    \item There exists a neighborhood $\mathcal U$ of $\mathcal H^+$ in $\mathcal M$ such that there are no trapped surfaces $S\subset\mathcal U$. 
\end{itemize}
\end{thm}


\begin{figure}[ht]
\centering{
\def\svgwidth{30pc}
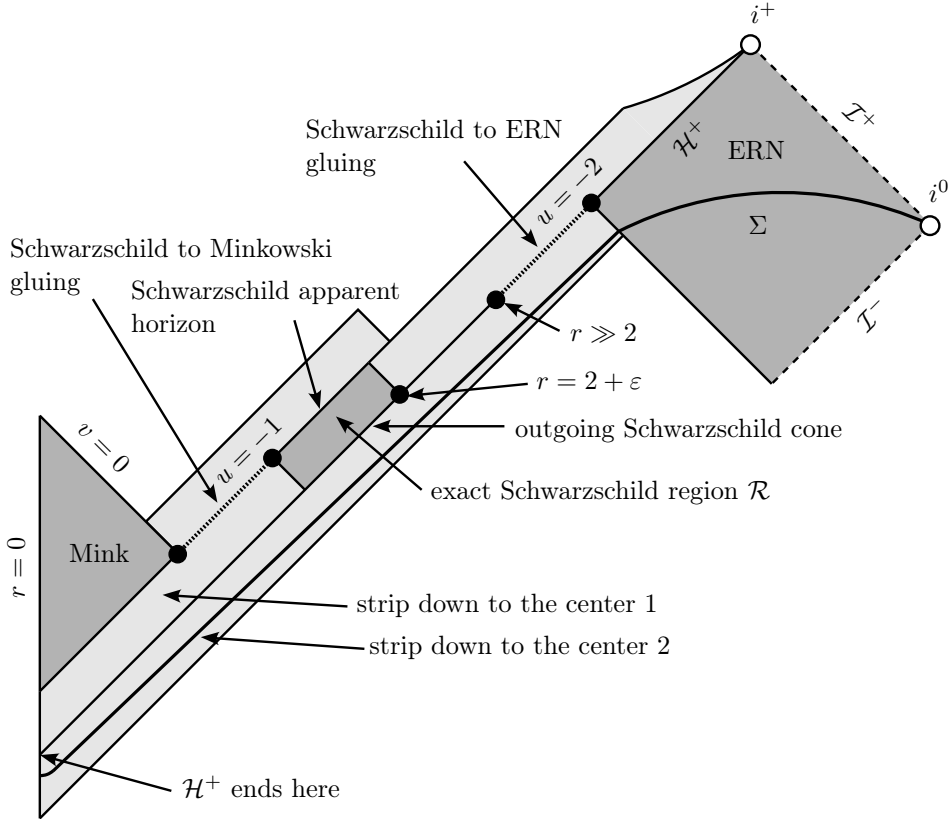}
\caption{Penrose diagram for the proof of \cref{trapped-surfaces-ERN-proof}.}
\label{trapped-surfaces-ERN-fig}
\end{figure}


\begin{proof}We refer to \cref{trapped-surfaces-ERN-fig} for a Penrose diagram illustrating the proof. 
The proof begins as the proof of \cref{corollary-proof} (recall also \cref{prop:spacetime-gluing}), by gluing a Minkowski cone to a Schwarzschild event horizon of unit mass along $\{ u=-1\}$. Then, attach a double null rectangle $\mathfrak R$ of Schwarzschild along the hypersurface $r=2$, as in \cref{corollary-proof}, but stop after a finite advanced time $v=2$.  Now place $u=-2$ so that 
\[\sup_{\{u=-2\}\cap\mathfrak R}r= 2+\ve\le 3.\] 
For $\ve$ sufficiently small, the first strip down to the center can be constructed as in the proof of \cref{corollary-proof}. Now let $M_f\ge M_0+8$ and extend the cone $u=-2$ to the future with trivial scalar field until $r=\tfrac 12(M_0+8)\gg 3$. Then using \cref{main-thm-RN}, extremal Reissner--Nordstr\"om of mass $M_f$ can be attached. We again solve backward up to the center as in \cref{corollary-proof} and have now constructed the spacetime depicted in \cref{trapped-surfaces-ERN-fig}. 

As in the proof of \cref{corollary-proof}, we again find an asymptotically flat spacelike curve $\Sigma$ connecting $i^0$ with the center and lying entirely in $J^-(\mathcal I^+)$.
The maximal future globally hyperbolic development  $(\mathcal M,g,F,A,\phi)$ of the induced data on $\Sigma$ contains the domain of dependence of $\Sigma$ in the spacetime constructed above (and thus all shaded regions to the future of $\Sigma$ in \cref{trapped-surfaces-ERN-fig}) and satisfies all the conclusions of \cref{corollary-proof} with $\mathfrak q=1$ and final mass $M_f\ge  M_0(1,\mathfrak e, k)+8$.  By construction, $\mathcal M$  contains the double null rectangle $\mathfrak R$ which satisfies the stated properties. Further, the cone $\{u=-1\}\cap\mathfrak R$ lies in the apparent horizon $\mathcal A$ of $(\mathcal M,g)$ and  $\{u=-1\}\cap\mathfrak R$ is isometric to an appropriate portion of the $r=2$ hypersurface in the Schwarzschild spacetime of mass $1$. 

We readily see that $(\mathcal M,g)$ contains trapped surfaces in any (future) neighborhood of $\{u=-1\}\cap\mathfrak R$ as $\partial_v r  =0 $ along $\{u=-1\}\cap\mathfrak R$ and \eqref{rdur} evaluated on $\{u=-1\}\cap\mathfrak R$  gives  
\begin{equation*}\partial_u ( r\partial_v r) = -\frac{\Omega^2}{4}.\end{equation*} 
To prove that trapped surfaces exist for arbitrarily late advanced time, we invoke the general boundary characterization of \cite{Kommemi13}. If the $r=0$ singularity $\mathcal S$ is empty, then the outgoing cone starting from one of these trapped spheres terminates on the Cauchy horizon $\mathcal{CH}^+$ and the claim is clearly true by Raychaudhuri's equation \eqref{eq:CSF-Raychaudhuri-v}. If $\mathcal S$ is nonempty, then every outgoing null cone which terminates on $\mathcal S$ is eventually trapped since $r$ extends continuously by zero on $\mathcal S$. Furthermore, $\mathcal S$ terminates at the Cauchy horizon $\mathcal{CH}^+$ or future timelike infinity $i^+$, so the claim is also true in this case.

We now show that there exists a neighborhood $\mathcal U$ of $\mathcal H^+$ in $\mathcal M$  which does not contain \emph{spherically symmetric} trapped surfaces. It suffices to show that there is a neighborhood $\mathcal V$ of $\mathcal H^+$ in $\mathcal Q$ such that $\partial_v r>0$ on $\mathcal V\setminus\mathcal H^+$, where we use the same symbol for the event horizon in $\mathcal M$ and $\mathcal Q$. Let $p\in\mathcal H^+$ be any sphere after the final gluing sphere, see \cref{trapped-surfaces-ERN-fig}. Then $r(p)=Q(p)=M_f$, $\partial_v r(p)=0$, and $\phi(p)=0$. Reparametrize the double null gauge so that $\Omega\equiv 1$ on the ingoing cone $\underline C$ passing through $p$. By the wave equation for the radius \eqref{eq:CSF-equation-for-r},
\begin{equation*}
    \partial_u\partial_v r(p)=-\frac{1}{4M_f}+\frac{M^2_f}{4M^3_f}=0. 
\end{equation*}
Differentiating  \eqref{eq:CSF-equation-for-r} in $u$, we find
\begin{equation*}
    \partial_u^2\partial_v r = \frac{\partial_u r}{4r^2}-\partial_u(\partial_u\log r)\partial_v r - (\partial_u\log r) \partial_u\partial_v r - \frac{3Q^2\partial_u r}{4r^4}+\frac{Q\partial_u Q}{2r^3}.
\end{equation*}
Evaluating at $p$, we find $\partial_u Q(p)=0$ by Maxwell's equation \eqref{eq:Maxwell-u}, so we have 
\begin{equation*}
    \partial_u^2 \partial_v r(p)=\frac{\partial_u r(p)}{4M_f^2} - \frac{3M^2_f\partial_u r(p)}{4M^4_f} = -\frac{2\partial_u r(p)}{M^2_f}>0.
\end{equation*}
Therefore, $\partial_v r$ becomes immediately positive for all points along $\underline C$ sufficiently close to the event horizon but not on it (see also \cref{fig:disconnected}).\footnote{This calculation is related to the discussion in \cref{rk:no-trapped-surfaces} above and \cref{app:A} below. In fact, we have effectively just proved the claim in \cref{rk:App1}.} 

By the monotonicity of Raychaudhuri's equation \eqref{eq:CSF-Raychaudhuri-v} and since $p\in \mathcal H^+$ after the final gluing sphere was arbitrary, this shows that there exists a neighborhood $\mathcal V$ of $\mathcal H^+$ contained in $\mathcal Q$ that does not contain trapped symmetry spheres except for $\mathcal H^+$ itself.  That there are also no nonspherically symmetric trapped surfaces in $\mathcal U\doteq \mathcal V \times S^2$ now follows immediately from \cref{no-trapped} below.

The claim about the disconnectedness of the outermost apparent horizon $\mathcal A'$ now follows from the fact that $\mathcal A'\cap \mathcal H^+$ is one connected component of $\mathcal A'$ which does not contain $\{u=-1\} \cap \mathfrak R \subset \mathcal A'$. This concludes the proof. \end{proof}

\subsection{Construction of collapse to Reissner--Nordstr\"om with piece of Cauchy horizon}
\label{subsec:collapse-with-CH}
In this section, we show that a mild modification of the proof of \cref{corollary-proof} allows us to construct examples of gravitational collapse such that the black hole region admits a piece of future boundary which is a Cauchy horizon which is isometric to a subextremal or extremal Reissner--Nordström Cauchy horizon.  

\begin{cor}\label{RN-with-CH-proof} For any $k\in\Bbb N$, $\mathfrak q\in [-1,1]\setminus \{0\}$, and $\mathfrak e\in\Bbb R\setminus\{0\}$, let $M_0(k,\mathfrak q,\mathfrak e, 1/2)$ be as in \cref{thm:interior-gluing}. Then for any $M\ge M_0$ there exist asymptotically flat, spherically symmetric Cauchy data $(\Sigma,g_0,k_0,E_0,B_0,\phi_0,\phi_1)$, with $\Sigma\cong \Bbb R^3$ and a regular center, for the EMCSF system such that the maximal future globally hyperbolic development $(\mathcal M^4,g,F,A,\phi)$ has the following properties:
\begin{itemize}
\item  All dynamical quantities are at least $C^{k}$-regular.
\item The spacetime and Cauchy data satisfy all the conclusions of \cref{corollary-proof}.
\item The black hole region contains an isometrically embedded portion of a Reissner--Nordström Cauchy horizon neighborhood with parameters $M$ and $\mathfrak q M$, in particular $\mathcal{CH}^+\ne\emptyset$. 
\end{itemize}
\end{cor}

\begin{figure}[ht]
\centering{
\def\svgwidth{20pc}
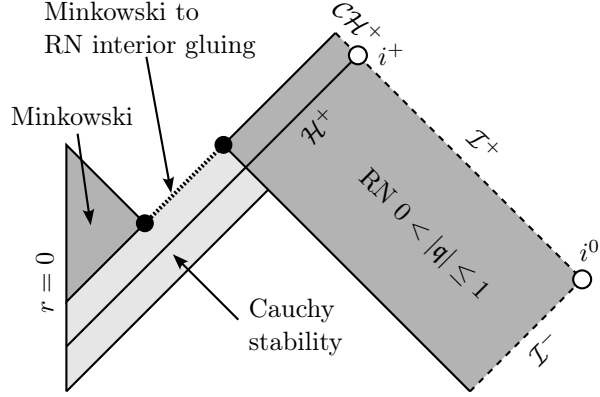}
\caption{Penrose diagram depicting the proof of \cref{RN-with-CH-proof}.}
\label{fig:spacetime-with-ch}
\end{figure}

\begin{proof}
The proof is completely analogous to the proof of \cref{corollary-proof}. We apply the gluing construction of \cref{thm:interior-gluing} to glue a sphere in Minkowski space to a Reissner--Nordström interior sphere with radius $R_f<r_+$  and $r_+ - R_f$ small. Indeed, this can be achieved by setting $\mathfrak r = \frac 12$ in \cref{thm:interior-gluing} as then $\tfrac 12 M_f ( 1+ \mathfrak r )\mathfrak q^2 \leq \frac 34 M_f < M_f \leq r_+$.
We then apply the local existence and Cauchy stability argument as in the proof of \cref{corollary-proof}. We note that the $u$-width of the local existence and Cauchy stability argument remains uniform as $R_f\to r_+$ so by choosing $R_f$ sufficiently close to $r_+$, we guarantee that we find a Cauchy hypersurface $\Sigma$ which does not intersect the event horizon. We refer to \cref{fig:spacetime-with-ch} for the Penrose diagram explaining the proof.
\end{proof}

\begin{rk}
As in \cref{rk:Schwarzschild-case}, we note that a similar statement with a piece of Schwarzschild interior including the $\{r =0\}$ singularity can be made with $\mathfrak q =0$. 
\end{rk}

\subsection{Construction of black hole interior for which the Cauchy horizon closes off spacetime}
\label{sec:closes-off-spacetime}
We now give our construction of a spacetime for which the Cauchy horizon closes off the black hole region. 

\begin{cor}\label{cor:CH-closes-off}
For any $k\in\Bbb N$, $\mathfrak q \in [-1,1]\setminus\{0\}$, $\mathfrak e\in\Bbb R\setminus\{0\}$, let $\tilde M_0(k,\mathfrak q,\mathfrak e, \mathfrak q^2/4)$ be as in \cref{thm:backwards-gluing}. Then for any $M\ge \tilde M_0$ there exist asymptotically flat, spherically symmetric Cauchy data $(\Sigma,g_0,k_0,E_0,B_0,\phi_0,\phi_1)$, with $\Sigma\cong \Bbb R^3$ and a regular center, for the EMCSF system such that the maximal future globally hyperbolic development $(\mathcal M^4,g,F,A,\phi)$ has the following properties:
\begin{itemize}
    \item All dynamical quantities are at least $C^{k}$-regular.
    \item The spacetime does not contain antitrapped surfaces.
    \item The black hole region is nonempty, $\mathcal{BH}\doteq \mathcal M\setminus J^-(\mathcal I^+)\ne\emptyset$.
    \item The future boundary of the black hole region is a $C^k$-regular Cauchy horizon $\mathcal{CH}^+$ which closes off spacetime, i.e., $\mathcal N\cup\mathcal S=\emptyset$ in \cref{Kommemi}. 
     \item The exterior region is isometric to a Reissner--Nordström exterior with mass $M$ and charge $\mathfrak q M$. In particular, future null infinity $\mathcal I^+$ is complete.
\end{itemize}
\end{cor}

\begin{figure}[ht]
\centering{
\def\svgwidth{13pc}
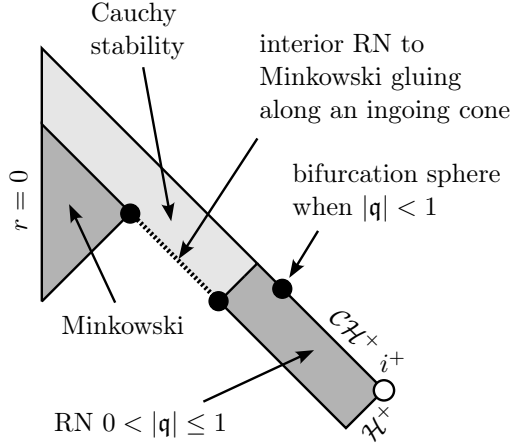}
\caption{Penrose diagram depicting the proof of \cref{cor:CH-closes-off}.}
\label{fig:CH-gluing-proof}
\end{figure}

\begin{proof}
Analogous to the proof of \cref{RN-with-CH-proof} we glue a Reissner--Nordström interior sphere with  $R_f< r_-$ and $r_- - R_f $ small to a sphere in Minkowski space along an \emph{ingoing} cone using \cref{thm:backwards-gluing}. We can choose $R_f$ arbitrarily close to $r_-$ in \cref{thm:backwards-gluing} by setting $\mathfrak r = \mathfrak q^2/4$. Indeed, in this case 
\begin{equation*}
    r_- - \frac{M_f}{2}\left(1+ \frac{\mathfrak q^2}{4} \right)\mathfrak q^2  =  M_f \left(1-\sqrt{1-\mathfrak q^2}\right) - \frac{M_f}{2}\left(1+ \frac{\mathfrak q^2}{4} \right)\mathfrak q^2 =M_f \left( 1- \sqrt{1-\mathfrak q^2} - \frac{\mathfrak q^2}{2}- \frac{\mathfrak q^4}{8} \right) \geq M_f \frac{\mathfrak q^6}{16},
\end{equation*} 
where in the last step we used the Taylor expansion of $\sqrt{1-\mathfrak q^2}$ around $\mathfrak q =0$.   The rest of the proof is now analogous to \cref{RN-with-CH-proof} and can be read off from \cref{fig:CH-gluing-proof}. We note that an isometric copy of the Reissner--Nordström exterior can be attached to the past of $\mathcal H^+$ in \cref{fig:CH-gluing-proof}.
\end{proof}

\begin{appendix}
\section{An isolated extremal horizon with nearby trapped surfaces}\label{app:A}

In this appendix we show that, in the context of the dominant energy condition, there is no local mechanism forcing a stationary extremal Killing horizon to have no trapped surfaces ``just inside'' of the horizon. We also refer back to \cref{rk:no-trapped-surfaces}.

\begin{prop}\label{prop:isolated-extremal-horizon}
There exists a $C^\infty$ spherically symmetric spacetime $(\mathcal M^4,g)$ with a complete null hypersurface $H\subset \mathcal M$ and a Killing vector field $T$ with the following properties. The Killing field $T$ is spherically symmetric, timelike in $I^-(H)$, spacelike in $I^+(H)$, null and tangent along $H$, where it also satisfies $\nabla_TT=0$, i.e., its integral curves are affinely parametrized null generators of $H$. Furthermore, $(\mathcal M,g)$ contains no antitrapped symmetry spheres, i.e., $\partial_u r <0$, and satisfies the dominant energy condition. Therefore, $H$ is an extremal Killing horizon and $I^+(H)$ is foliated by trapped symmetry spheres. 
\end{prop}

We recall that a spacetime $(\mathcal M,g)$ satisfies the \emph{dominant energy condition} 
if for all future directed causal vectors $X\in T\mathcal M$, $-G(\cdot,X)^\sharp$ is future directed causal or zero. Here $G$ denotes the Einstein tensor of $g$, 
\[G(g)\doteq \Ric(g)-\tfrac 12 R(g)g.\]

\begin{proof}
The spacetime is given by the spherically symmetric ansatz
\begin{align*}
    \mathcal M&=\mathcal Q\times S^2\\
    g&= g_\mathcal Q+r^2 g_{S^2},
\end{align*}
where 
\[\mathcal Q=\{(t,u)\in\Bbb R^2:t\in\Bbb R, -\ve <u<\ve\}\]
 for $\ve$ to be chosen later, and $r=r(u)$. Let $f=f(u)$ and set
\[g_\mathcal Q = fdt^2-2dtdu.\]
The vector field $\underline L=\partial_u$ is geodesic and null and we declare it to be future directed. The Killing vector field $T=\partial_t$ satisfies $g(T,T)=f$. Letting $f(u)=u^2F(u)$ for a smooth function $F(u)$ makes $H=\{u=0\}$ an extremal Killing horizon and $\partial_t$ is future directed where it is causal. The conjugate null vector to $\underline L$ is $L=\partial_t + \tfrac 12 f \partial_u$ such that $g(\underline L, L) = -1$. The symmetry spheres $S_{t_0,u_0}=\{t=t_0\}\cap\{u=u_0\}$ are trapped if 
\begin{align*}
    Lr&<0\\
    \underline Lr&<0,
\end{align*}
which can be more simply written as
\begin{align*}
 f(u)r'(u)    &<0\\
   r'(u) &<0.
\end{align*}
From this we see that $r'(u)<0$ implies no antitrapped spheres of symmetry and $f(u)<0$ for $u<0$ and $f(u)>0$ for $u>0$ implies the symmetry spheres to the past (respectively, future) of $H$ are untrapped (respectively, trapped). This also makes $T$ timelike to the past of $H$. Since we require $f(u)=u^2F(u)$ but also that $f$ changes sign, we in fact have $f(u)=u^3\tilde F(u)$.

We will now see which restrictions on $f$, $r$, and $\ve$ enforce the dominant energy condition. The Einstein tensor of $g$ is given by 
\begin{equation}
    G=-\theta g_\mathcal Q - \frac{2r''}{r}du^2 +\zeta r^2 g_{S^2},\label{A:Einstein}
\end{equation}
where
\begin{equation*}\theta\doteq \frac{1+(r')^2 f +rf'r'+2frr''}{r^2},\quad \zeta\doteq -\tfrac 12   f'' -  \frac{r'}{r} f'-\frac{r''}{r} f .
\end{equation*}
For $f(u)=u^3\tilde F(u)$ and $r(u)$ fixed  and $\ve>0$ sufficiently small, we have $\theta(u)>0$ and $|\zeta(u)|\ll \theta(u)$ for $|u|<\ve$. 

Let $X$ be a future causal vector, that is
\begin{equation} \label{eq:X-causal-vector} g_\mathcal Q (X,X) + r^2 g_{S^2}(X,X) \leq 0, \quad g_{\mathcal Q}(\underline L + L, X) < 0.\end{equation} To show that $-G(\cdot,X)^\sharp$ is causal or zero, it suffices to show that 
\begin{equation}
    g_{\mu\nu}G^\mu{}_\rho G^\nu{}_\sigma X^\rho X^\sigma \le 0 .\label{A:causal}
\end{equation}
To simplify the calculation, we assume $r''$ vanishes identically and then the left-hand side of \eqref{A:causal}, using \eqref{A:Einstein} and \eqref{eq:X-causal-vector}, can be estimated as 
\begin{equation*}
g_{\mu\nu}G^\mu{}_\rho G^\nu{}_\sigma X^\rho X^\sigma= \theta^2 g_\mathcal Q(X,X)+ \zeta^2 r^2 g_{S^2}(X,X)\leq    (\zeta^2 - \theta^2) r^2 g_{S^2}(X,X).\end{equation*}  
 Since   $ \zeta^2 -\theta^2 \le 0$, this proves that $-G(\cdot,X)^\sharp$ is causal. 
To show that $-G(\cdot,X)^\sharp$ is future directed we compute using \eqref{eq:X-causal-vector} \begin{equation*} g (\underline L + L , -G(\cdot,X)^\sharp )=-G(\underline L + L  ,X) = \theta g_{\mathcal Q} (\underline L+L ,X) < 0 . \end{equation*}
 
Finally, an explicit example of a metric satisfying all of our conditions is
\begin{equation*}
g=u^3dt^2 - 2dtdu + (1-u)^2g_{S^2}.\qedhere
\end{equation*}
\end{proof}

\begin{rk}\label{rk:App1}
Extremal Reissner--Nordstr\"om has $f(u)\sim -u^2$. One might say that an extremal horizon constructed in the above manner with $f(u)$ vanishing faster than $u^2$ is a \emph{degenerate extremal horizon}.   
\end{rk}

\section{General trapped and antitrapped surfaces in spherically symmetric spacetimes}\label{app-B}

In this appendix we infer the absence of nonspherically symmetric trapped or antitrapped surfaces from the absence of spherically symmetric trapped or antitrapped surfaces.

Our definition of trapped surface is completely standard, see \cref{def:trapped} below. (Note that we assume trapped surfaces to be closed and strictly trapped.) Our definition of antitrapped is as in \cite{TheBVPaper, Kommemi13}, i.e., an antitrapped surface is closed and past weakly outer trapped, see \cref{defn:antitrapped} below.

\begin{prop}\label{no-trapped}
Let $(\mathcal M^4,g)$ be a spherically symmetric spacetime as defined in \cref{subsec:einstein-maxwell-sph-symm}. Then there are no trapped surfaces contained in the sets
\begin{align}
A&\doteq \{p\in \mathcal M:\partial_u r \ge 0\},\\
B&\doteq \{p\in \mathcal M:\partial_v r\ge 0\}.
\end{align}

\end{prop}

\begin{rk}
Note that there could be trapped surfaces contained in $A\cup B$. There might also be trapped surfaces which merely intersect $A$ or $B$.   
\end{rk}

\begin{prop}\label{no-antitrapped}
Let $(\mathcal M^4,g,F,A,\phi)$ be a spherically symmetric spacetime
arising as the maximal future globally hyperbolic development from one-ended asymptotically flat Cauchy data for the EMCSF system with no antitrapped spheres of symmetry as in \cite{Kommemi13}. Then:
\begin{enumerate}
\item If $S$ is a trapped surface in $\mathcal M$, then $S\cap J^-(\mathcal I^+)=\emptyset$. 

\item $\mathcal M$ does not contain any antitrapped surfaces. 
\end{enumerate}
\end{prop}

\begin{rk}
Under stronger assumptions on $\mathcal I^+$, the first part of the previous proposition would follow from a classical result of Hawking \cite[Proposition 9.2.1]{Hawking72BH, HE73}. 
\end{rk}

For the proofs, we recall some facts from Lorentzian geometry \cite{Galloway}. Let $H$ be a null hypersurface in a spacetime $(\mathcal M^4,g)$, i.e., $H$ is a 3-dimensional submanifold of $\mathcal M$ and admits a future-directed normal vector field $L$ which is null and whose integral curves can be reparametrized to be null geodesics. We say that $L$ is a \emph{(future-directed) null generator} of $H$.

 The \emph{second fundamental form of $H$ with respect to $L$} is given by 
 \begin{equation}
     B^L(X,Y)=g(\nabla_XL,Y)
 \end{equation}
 for $X,Y\in TH$. If $e_1$ and $e_2$ are an orthonormal pair of spacelike vectors at $p\in H$, we define the \emph{null expansion of $H$ with respect to $L$} by
 \begin{equation}
     \theta^L=B^L(e_1,e_1)+B^L(e_2,e_2)\label{theta-defn}
 \end{equation}
 at $p$, and this definition is independent of the pair $e_1$ and $e_2$. If $\tilde L$ is another future-directed null generator of $H$, then there is a positive function $f$ on $H$ such that $\tilde L=fL$. In this case, we have 
 \begin{equation}
     \theta^{\tilde L}= f\theta^L.\label{homogen}
 \end{equation}

 \begin{lem}[Comparison principle for null hypersurfaces]\label{comparison-principle}
Let $H_1$ and $H_2$ be null hypersurfaces in $(\mathcal M^4,g)$, with $H_1$ to the future of $H_2$ and generated by $L_1$ and $L_2$, respectively. If $H_1$ and $H_2$ are tangent at a point $p$, and $L_1(p)=L_2(p)$, then 
\begin{equation}
\theta_{H_1}^{L_1}(p)\ge \theta_{H_2}^{L_2}(p).\label{B-inequality}
\end{equation}
 \end{lem}
\begin{proof}
By \eqref{homogen}, it suffices to prove \eqref{B-inequality} with respect to some choice of null generators of $H_1$ and $H_2$ which agree at $p$. Let $(t,x,y,z)$ be normal coordinates for $g$ based at $p$ so that $\partial_t$ is future-directed and $\{\tfrac 12(\partial_t +\partial_x),\partial_y,\partial_z\}$ spans $T_pH_1=T_pH_2$. We introduce approximate null coordinates $u=t-x$ and $v=t+x$, so that
\begin{equation*}
    \partial_u = \tfrac12(\partial_t-\partial_x), \quad
    \partial_v = \tfrac12(\partial_t+\partial_x).
\end{equation*}
Note that $\partial_u$ and $\partial_v$ are only guaranteed to be null at $p$. 

By the implicit function theorem, there exist functions $f_1(v,y,z)$ and $f_2(v,y,z)$ defined near $p$, so that, upon defining
\begin{equation*}
    \zeta_1(u,v,y,z)\doteq f_1(v,y,z)-u, \quad \zeta_2(u,v,y,z)\doteq f_2(v,y,z)-u,
\end{equation*}
we have $H_i=\{\zeta_i=0\}$ for $i=1,2$. Note that $f_1(p)=f_2(p)=0$ and that $p$ is a critical point for $f_1$ and $f_2$.  The vector fields $Z_i=\grd \zeta_i$ are null on $H_i$ and define there future-directed null generators. In particular, we have $Z_1(p)=Z_2(p)=\partial_v|_p$.

We first show that $f_1\ge f_2$ near $p$. If a point $q=(u,v,y,z)$ lies to the past of $H_1$, then $\zeta_1(q)\ge 0$. If $q\in H_2$, then $\zeta_2(q)=0$, so combining these inequalities yields 
\[f_1(v,y,z)= \zeta_1(q)+u\ge \zeta_2(q)+u=f_2(v,y,z),\] as claimed.

We now show that \begin{equation}
    B^{Z_1}_{H_1}(\partial_y,\partial_y)(p)\ge B^{Z_2}_{H_2}(\partial_y,\partial_y)(p),\label{proof-ineq}
\end{equation}
the corresponding statement and proof for $\partial_z$ being the same. By \eqref{theta-defn} this will complete the proof. Since $f_1\ge f_2$ near $p$, $p$ is a local minimum for $f_1-f_2$. It follows that 
\begin{equation}\partial_y^2 (f_1-f_2)(p)\ge 0\label{2nd-deriv}\end{equation}
by the second derivative test. Since we are working in a normal coordinate system, 
\[B^{Z_i}_{H_i}(\partial_y,\partial_y)(p)=g(\nabla_{\partial_y}\nabla\zeta_i,\partial_y)(p)=\partial^2_y f_i(p),\]
whence \eqref{2nd-deriv} proves \eqref{proof-ineq}, which completes the proof. 
\end{proof}

\begin{defn}\label{def:trapped}
A closed spacelike 2-surface $S$ in a spacetime $(\mathcal M^4,g)$ is always the intersection of two locally defined null hypersurfaces. We say that $S$ is \emph{trapped} if both of these hypersurfaces have negative future null expansion along $S$. 
\end{defn}

\begin{proof}[Proof of \cref{no-trapped}]
We show that there is no trapped surface $S\subset B$. The argument for $S\subset A$ is analogous after noting that $A\cap \Gamma=\emptyset$ by our definition of spherical symmetry and convention for $u$. 

Let $S\subset \{\partial_v r\ge 0\}$ be a closed 2-surface. Let $\pi:\mathcal M\to \mathcal Q$ be the projection of the spherically symmetric spacetime to its Penrose diagram. Then $\pi(S)$ is a compact subset of $\mathcal Q$ and hence $u$ attains a minimum $u_0$ on $\pi(S)$. 

Therefore, there exists a symmetry sphere $S_{u_0,v_0}$ on which $\partial_v r\ge 0$ such that $S$ lies to the future of $C_{u_0}$ and is tangent to this cone at a point $p\in S_{u_0,v_0}$. Note that $p\notin\Gamma$ because $C_{u_0}$ is not regular there. The condition $\partial_v r\ge 0$ means $C_{u_0}$ has nonnegative future expansion. By \cref{comparison-principle}, one of the two null hypersurfaces emanating from $S$ also has nonnegative future expansion, so $S$ is not trapped.   
\end{proof}

\begin{defn}\label{defn:antitrapped}
Let $(\mathcal M^4,g)$ be a spacetime satisfying the hypotheses of \cref{no-antitrapped}.
A closed spacelike 2-surface $S$ which bounds a compact spacelike hypersurface $\Omega$ is said to be \emph{antitrapped} if its future-directed inward null expansion is nonnegative. Here the (locally defined) inward null hypersurface $H_{\mathrm{in}}$ emanating from $S$ is chosen to be the one which smoothly extends the boundary of the causal past of $\Omega$. 
\end{defn}

\begin{proof}[Proof of \cref{no-antitrapped}]
1. Since $r\to \infty$ at $\mathcal I^+$ \cite{Kommemi13}, Raychaudhuri's equation \eqref{eq:CSF-Raychaudhuri-v} implies $\partial_v r>0$ in $J^-(\mathcal I^+)$. Let $S$ be a closed 2-surface such that $S\cap J^-(\mathcal I^+)\ne\emptyset$. Let $\pi:\mathcal M\to \mathcal Q$ be the projection to the Penrose diagram. Then $u$ attains a minimum $u_0$ on $\pi(S)$. By the causal properties of $J^-(\mathcal I^+)$, there exists a symmetry sphere $S_{u_0,v_0}\subset J^-(\mathcal I^+)$ such that $S$ lies to the future of $C_{u_0}$ and is tangent to the cone at $p\in S_{u_0,v_0}$. Arguing as in the proof of \cref{no-trapped}, we see that one of the null hypersurfaces emanating from $S$ has positive future expansion, so $S$ is not trapped. 

2. Let $\pi:\mathcal M\to \mathcal Q$ be again the projection. Then $v$ attains a maximum $v_0$ on $\pi(S)$ and again there exists a non-central symmetry sphere $S_{u_0,v_0}$ such that $\partial_u r(u_0,v_0)<0$, $S$ lies to the past of $C_{v_0}$, and is tangent to the cone at a point $p\in S_{u_0,v_0}$. Now $C_{v_0}$ is tangent to $H_\mathrm{in}$ at $p$ and lies to the future, so by \cref{comparison-principle}, $H_{\mathrm{in}}$ has negative null expansion at $p$. Therefore, $S$ is not antitrapped. 
\end{proof}

 \end{appendix}
\printbibliography[heading=bibintoc] 

\end{document}